\documentclass[aps,prx,reprint,superscriptaddress,nofootinbib,longbibliography,floatfix]{revtex4-2}

\usepackage[T1]{fontenc}
\usepackage{amsmath,amssymb,mathtools,bm,amsthm}
\usepackage{graphicx}
\usepackage{booktabs,array}
\usepackage{algpseudocode}
\usepackage{microtype}
\usepackage{xcolor}
\usepackage{tikz}
\usepackage{pgfplots}
\pgfplotsset{compat=1.18}
\usetikzlibrary{arrows.meta,positioning,calc,fit,backgrounds,shapes.geometric}
\definecolor{gkpblue}{HTML}{2F6B9A}
\definecolor{gkpteal}{HTML}{2A9D8F}
\definecolor{gkporange}{HTML}{E76F51}
\definecolor{gkppurple}{HTML}{6A4C93}
\definecolor{gkpgreen}{HTML}{7CB342}
\usepackage{braket}
\usepackage{hyperref}
\hypersetup{colorlinks=true,linkcolor=blue,citecolor=blue,urlcolor=blue}
\allowdisplaybreaks[2]

\newcommand{\F}{\mathbb{F}}
\newcommand{\R}{\mathbb{R}}
\newcommand{\Z}{\mathbb{Z}}
\newcommand{\I}{\mathbb{I}}
\newcommand{\DV}{\mathrm{DV}}
\newcommand{\obs}{\mathrm{obs}}

\newcommand{\meas}{\mathrm{meas}}

\newcommand{\fus}{\mathrm{fus}}
\newcommand{\Bell}{\mathrm{Bell}}
\newcommand{\GKP}{\mathrm{GKP}}
\newcommand{\locera}{\bot_{\rm loc}}
\newcommand{\accera}{\bot_{\rm acc}}
\newcommand{\confera}{\bot_{\rm conf}}
\newcommand{\LA}{\mathrm{LA}}

\newcommand{\Tr}{\operatorname{Tr}}
\newcommand{\syn}{\operatorname{syn}}
\newcommand{\spanZ}{\operatorname{span}_{\Z}}
\newcommand{\spanF}{\operatorname{span}_{\F_2}}

\newcommand{\argmax}{\operatorname*{arg\,max}}

\newtheorem{theorem}{Theorem}[section]

\theoremstyle{definition}

\theoremstyle{remark}

\newcounter{prxalgorithm}
\renewcommand{\theprxalgorithm}{\arabic{prxalgorithm}}
\newenvironment{prxalgorithm}[1]{%
  \begin{figure}[!t]
  \refstepcounter{prxalgorithm}%
  \begin{minipage}{\columnwidth}
  \noindent\rule{\columnwidth}{0.5pt}\par\nobreak
  \vskip 3pt plus 1pt minus 3pt\nobreak
  \noindent\textbf{Algorithm \theprxalgorithm.}~\textbf{#1}\par\nobreak
  \vskip 3pt plus 1pt minus 3pt\nobreak
}{%
  \par\nobreak
  \vskip 3pt plus 1pt minus 3pt\nobreak
  \noindent\rule{\columnwidth}{0.5pt}%
  \end{minipage}
  \end{figure}
}
\algrenewcommand\algorithmicrequire{\textbf{Input:}}
\algrenewcommand\algorithmicensure{\textbf{Output:}}

\begin{document}

\title{Measurement-Based Loss Tolerance in Graph--GKP Codes through Syndrome-Resolved Pauli-Frame Decoding}

\author{Seid Koudia}
\affiliation{Interdisciplinary Centre for Security, Reliability and Trust (SnT), University of Luxembourg, Luxembourg}
\author{Symeon Chatzinotas}
\affiliation{Interdisciplinary Centre for Security, Reliability and Trust (SnT), University of Luxembourg, Luxembourg}
\date{\today}

\begin{abstract}
Graph codes offer multiple physical representatives of logical observables, while Gottesman–Kitaev–Preskill (GKP) codes retain analog information about bosonic displacement noise. We develop a causal framework that unifies these mechanisms for measurement‑based loss tolerance under pure loss followed by quantum‑limited amplification. In this framework, each local GKP recovery produces a refreshed logical block, a continuous syndrome record, and a confidence score for the inferred Pauli class. Low‑confidence outcomes are deliberately converted into located erasures, so the availability pattern is generated directly from the bosonic data rather than sampled independently. Both accepted and rejected syndromes contribute to a syndrome‑resolved posterior over the graph branch, which determines accessible logical representatives and the outgoing logical Pauli frame.
We derive decoder‑conditioned branch restriction, signed‑outcome reconstruction, Pauli‑frame updating, recursive concatenation of graph‑GKP modules, and syndrome‑resolved logical fusion. Numerical simulations across several squeezing levels identify task‑dependent loss‑tolerance behavior and finite‑depth pseudothresholds for square and hexagonal GKP lattices. The resulting graph‑GKP interface provides a unified causal control layer for fault‑tolerant MBQC, fusion‑based computation and all photonic repeaters.
\end{abstract}

\maketitle

\section{Introduction}
\label{sec:introduction}

Fault-tolerant quantum computation requires architectures that suppress physical noise while remaining compatible with realistic preparation, measurements, feedforward, and classical control. This is particularly demanding in photonic and bosonic platforms, where loss, finite squeezing, detector imperfections, mode mismatch, and probabilistic entangling operations must be incorporated into the computational model. Recent studies similarly emphasize the need for physically faithful optical models, experimentally accessible recovery operations, and decoders matched to the dominant noise mechanisms \cite{Vischi2024Photonic,Hoch2025AdaptiveBoson,EisertPreskill2025,ZengEtAl2025}.

Measurement-based quantum computation (MBQC) provides a natural setting for this integration. An entangled resource state is prepared in advance, while computation proceeds through local measurements, classical feedforward, and adaptive basis choices \cite{raussendorf2001,raussendorf2003}. In photonic implementations, successful operation therefore depends on the ability of the classical controller to identify usable measurement patterns, reroute around unavailable resource nodes, and update the logical Pauli frame.

Graph-state encodings provide such flexibility because a logical observable may admit several physically distinct representatives on the same graph. The controller can consequently replace an unavailable representative by another one supported on the surviving resource. This principle underlies loss-tolerant graph-state computation and adaptive measurement protocols \cite{varnava2006,morleyshort2019,bell2023}, and is also relevant to photonic networks and repeater architectures affected by transmission loss, probabilistic state preparation, and incomplete Bell measurements \cite{LiEtAl2025,koudia2024quantum}.

Gottesman--Kitaev--Preskill (GKP) codes provide a complementary form of protection by encoding logical information in the phase space of bosonic modes \cite{gkp2001,larsen2025integrated}. Modular quadrature measurements retain continuous information about displacement errors, which can substantially improve logical decoding compared with hard binary decisions \cite{conrad2022,lin2023}. Their performance is intrinsically linked to the underlying bosonic channel, including pure-loss and amplification processes \cite{ZhengEtAl2025,WangJiang2025}. GKP states also belong to the broader class of non-Gaussian resources relevant to bosonic information processing \cite{Koudia2021NGaussCausal,Koudia2019CausalDisc}.

Graph-code pathfinding and analog GKP decoding have nevertheless largely been treated as separate control layers. Graph decoders select logical representatives compatible with a declared availability pattern \cite{bell2023}, whereas GKP decoders infer logical sectors from continuous measurement records \cite{conrad2022,lin2023}. A simple sequential composition is insufficient when availability is itself determined by thresholding those records. In that case, the selected graph branch is a data-dependent event, and the likelihood of both accepted and rejected local decisions must remain part of the global decoder.

We develop a causal graph--GKP framework for this setting. Physical attenuation is modeled through pure loss followed by quantum-limited amplification, yielding displacement statistics processed by the local GKP decoder. Local recovery is performed before the graph controller selects the destructive outer measurements. Each recovered block produces an analog syndrome, an inferred residual Pauli class, and a confidence score. Blocks with insufficient confidence are deliberately converted into located erasures; thus, unavailability is generated from the bosonic measurement record rather than sampled independently from an additional deletion channel. This ordering distinguishes three operational events: local decoder rejection, failure to find an accessible graph representative, and insufficient confidence in the logical frame inferred from an accessible representative. Local rejection may be bypassed through graph pathfinding, whereas accessibility and frame-confidence failures directly affect the requested logical operation. The decoder therefore retains the selected branch, the available outer syndrome information, and the corresponding posterior over the outgoing logical Pauli frame. Exact wrapped logical-coset maximum likelihood defines the statistical decoding target, while the executable decoder uses a covariance-weighted closest-coset approximation. The resulting selected-event likelihood provides a common interface for local GKP recovery, logical Pauli measurements, adaptive transport to non-Pauli terminal measurements, recursive graph-code concatenation, and logical parity fusion. We instantiate the construction using cube and decorated-pentagon modules and extend it to recursively concatenated resources. Numerical simulations compare square and hexagonal GKP lattices, finite graph families, several concatenation depths, and adaptive and transversal fusion strategies, yielding finite-size pseudothresholds for logical measurements and their associated Pauli-frame layer.

The paper is organized as follows. Section~\ref{sec:gkp_machinery} introduces the GKP conventions, loss-and-amplification model, and local closest-coset decoder. Section~\ref{sec:graph_gkp_codes} develops the graph--GKP construction. Section~\ref{sec:logical_measurements} treats logical Pauli and non-Pauli measurements, while Sec.~\ref{sec:analog_qec} derives the branch-conditioned decoder and Pauli-frame updates. Sections~\ref{sec:modularization} and \ref{sec:fusion} address recursive concatenation and logical fusion, respectively. Section~\ref{sec:discussion} discusses applications to fault-tolerant MBQC, fusion-based computation, photonic repeaters, and multimode bosonic architectures. 

\section{Bosonic phase space and GKP codes}
\label{sec:gkp_machinery}

This section fixes the continuous-variable convention used in every subsequent derivation  \cite{gkp2001,conrad2022}.

\subsection{Phase-space and GKP lattice conventions}

For $m$ oscillator modes, define
\begin{equation*}
\hat{\bm x}=(\hat q_1,\ldots,\hat q_m,\hat p_1,\ldots,\hat p_m)^T,
\qquad
J_{2m}=\begin{pmatrix}0&I_m\\-I_m&0\end{pmatrix},
\end{equation*}
so that, with $\hbar=1$,
\begin{equation*}
[\hat x_j,\hat x_k]=i(J_{2m})_{jk}.
\end{equation*}
For a dimensionless phase-space vector $\xi\in\R^{2m}$, the displacement operator is
\begin{equation}
D(\xi)=\exp\!\left[-i\sqrt{2\pi}\,\xi^T J_{2m}\hat{\bm x}\right].
\label{eq:displacement}
\end{equation}
It shifts the quadratures by
\begin{equation*}
D(\xi)^\dagger\hat{\bm x}D(\xi)=\hat{\bm x}+\sqrt{2\pi}\,\xi,
\end{equation*}
and obeys the Weyl relations
\begin{align}
D(\xi)D(\eta)
&=e^{-i\pi\xi^TJ_{2m}\eta}D(\xi+\eta),
\label{eq:weyl_product}\\
D(\xi)D(\eta)
&=e^{-i2\pi\xi^TJ_{2m}\eta}D(\eta)D(\xi).
\label{eq:weyl_comm}
\end{align}
Thus two displacements commute exactly when their symplectic product is integral. A Gaussian unitary with symplectic matrix $S\in\mathrm{Sp}_{2m}(\R)$ acts by $U_SD(\xi)U_S^\dagger=D(S\xi)$.

\label{subsec:gkp_lattice}

Let $M\in\R^{2m\times2m}$ be full rank, with row vectors $\xi_1^T,\ldots,\xi_{2m}^T$. In the row convention used here,
\begin{equation*}
L(M)=\left\{\xi\in\R^{2m}:\xi^T=a^TM,\ a\in\Z^{2m}\right\}
\end{equation*}
is the displacement lattice generated by $M$. Its symplectic Gram matrix is
\begin{equation*}
A=MJ_{2m}M^T.
\end{equation*}
The displacement generators commute if and only if $A$ is integer valued. Importantly, the lattice alone does not specify every stabilizer sign when $A$ is not even. Let $A_{\triangle}$ denote the strictly lower-triangular part of $A$. Repeated use of Eq.~\eqref{eq:weyl_product} gives, for $\xi(a)=M^Ta$,
\begin{equation}
\prod_{j=1}^{2m}D(\xi_j)^{a_j}
=e^{i\pi a^TA_{\triangle}a}D\!\left(\xi(a)\right).
\label{eq:lattice_product_phase}
\end{equation}
We therefore define the phase function
\begin{equation}
\phi_M(\xi(a))=\pi a^TA_{\triangle}a\pmod{2\pi}
\label{eq:phi_M}
\end{equation}
and the GKP stabilizer group
\begin{equation}
\mathcal S(M,\phi_M)
=\left\{e^{i\phi_M(\xi)}D(\xi):\xi\in L(M)\right\}.
\label{eq:gkp_stabilizer_group}
\end{equation}
The ideal code space is the simultaneous $+1$ eigenspace of this group. If $M'=UM$ with $U\in\mathrm{GL}_{2m}(\Z)$ unimodular, then $M'$ generates the same geometric lattice, but the phase sector must be transported with the basis change; it must not in general be reset to zero.

The symplectic dual lattice is
\begin{equation*}
L^\perp
=\left\{\xi^\perp\in\R^{2m}:
(\xi^\perp)^TJ_{2m}\xi\in\Z\ \text{for all }\xi\in L\right\}.
\end{equation*}
A canonical dual generator is
\begin{equation*}
M^\perp=(J_{2m}M^T)^{-1}=M^{-T}J_{2m}^T.
\end{equation*}
For a GKP code, $L\subseteq L^\perp$. The quotient
\begin{equation*}
L^\perp/L\simeq C(\mathcal S)/\mathcal S
\end{equation*}
lists the logically distinct displacement operators. Its cardinality and the code-space dimension are
\begin{align}
|L^\perp/L|&=|\det A|=|\det M|^2,\nonumber\\
\dim\mathcal H_C&=|\det M|=\sqrt{|\det A|}.
\label{eq:gkp_dimension}
\end{align}
These are established lattice identities; they are included because the graph--GKP compilation below is an explicit superlattice construction and its logical dimension will be checked with Eq.~\eqref{eq:gkp_dimension}.

\subsection{Single-mode lattices and finite-energy GKP states}
\label{subsec:square_hex_gkp}

A scaled one-qubit GKP code has type $d=2$, equivalently $A=2J_2$. The two local lattices used in the numerical figures are generated by
\begin{align*}
M_{\square}&=\sqrt2\,I_2,\\
M_{\mathrm{hex}}&=3^{-1/4}
\begin{pmatrix}2&0\\1&\sqrt3\end{pmatrix}.\end{align*}
Direct multiplication gives
\begin{equation*}
M_{\square}J_2M_{\square}^T
=M_{\mathrm{hex}}J_2M_{\mathrm{hex}}^T
=2J_2.
\end{equation*}
Let $\xi_1^T,\xi_2^T$ be the rows of either matrix and choose the quotient representatives
\begin{equation}
e=\frac{\xi_1}{2},
\qquad
f=\frac{\xi_2}{2}.
\label{eq:local_e_f}
\end{equation}
Then $e,f\in L^\perp$, $2e,2f\in L$, and $e^TJ_2f=1/2$. We use
\begin{equation}
\overline X=D(e),
\qquad
\overline Z=D(f),
\qquad
\overline Y=D(e+f),
\label{eq:gkp_logical_paulis_ef}
\end{equation}
which satisfy the qubit Pauli algebra by Eq.~\eqref{eq:weyl_comm}. For the square code this specializes to
\begin{equation}
\overline X=e^{-i\sqrt\pi\hat p},
\qquad
\overline Z=e^{i\sqrt\pi\hat q},
\label{eq:square_logical_paulis}
\end{equation}
while the stabilizers are $e^{-i2\sqrt\pi\hat p}$ and $e^{i2\sqrt\pi\hat q}$. The Euclidean GKP distance
\begin{equation*}
\Delta(L)=\min_{x\in L^\perp\setminus L}\|x\|_2
\end{equation*}
is $\Delta_{\square}=2^{-1/2}$ and $\Delta_{\mathrm{hex}}=3^{-1/4}$ in the coordinates of Eq.~\eqref{eq:displacement}. The hexagonal advantage in the plots therefore originates at the local lattice level from a larger shortest nontrivial logical displacement; it is not introduced as an abstract ``lattice factor.''

\label{subsec:gkp_states}

Ideal GKP states are non-normalizable lattice-comb eigenstates.  Finite-energy codewords are commonly regularized as
\begin{equation}
|\widetilde\psi_{\beta}\rangle
=\mathcal N_{\beta}e^{-\beta\hat N}|\psi_L\rangle,
\qquad
\hat N=\sum_{j=1}^m\hat a_j^\dagger\hat a_j,
\label{eq:finite_energy_regularization}
\end{equation}
or by equivalent approximate comb models \cite{gkp2001,conrad2022}.  Those state-preparation models are not simulated here.  The numerical analysis begins after preparation and readout have been reduced to a calibrated Gaussian displacement channel with width $\widetilde\sigma$ and covariance contributions collected in Eq.~\eqref{eq:loss_amplification_total_covariance}.  We therefore do not identify $\widetilde\sigma$ with a state-envelope parameter without specifying a preparation circuit.

\subsection{Gaussian channel, closest-coset decoding, and local erasure conversion}
\label{subsec:local_gkp_mld}

The effective Gaussian displacement channel in the adopted normalization is
\begin{equation*}
\mathcal N_{\widetilde\sigma}(\rho)
=\int_{\R^{2m}}d^{2m}\epsilon\,
P_{\widetilde\sigma}(\epsilon)
D(\epsilon)\rho D(\epsilon)^\dagger,
\end{equation*}
with
\begin{equation*}
P_{\widetilde\sigma}(\epsilon)
=\frac{1}{(2\pi\widetilde\sigma^2)^m}
\exp\!\left[-\frac{\|\epsilon\|_2^2}{2\widetilde\sigma^2}\right].
\end{equation*}
Because Eq.~\eqref{eq:displacement} shifts a physical quadrature by $\sqrt{2\pi}\epsilon$, a conventional quadrature variance $V$ corresponds to a displacement-coordinate variance $V/(2\pi)$.

For an error $\epsilon$, the continuous GKP stabilizer syndrome is
\begin{equation*}
s(\epsilon)=MJ_{2m}\epsilon\pmod 1.
\end{equation*}
Choose the centered representative $\widetilde s\in[-1/2,1/2)^{2m}$ and the pure error
\begin{equation*}
\eta(s)=(MJ_{2m})^{-1}\widetilde s.
\end{equation*}
Then every error compatible with syndrome $s$ lies in one and only one affine logical class
\begin{equation*}
\eta(s)+\xi_{\alpha}^{\perp}+L,
\qquad
\xi_{\alpha}^{\perp}\in L^\perp/L.
\end{equation*}
The exact maximum-likelihood probability of logical class $\alpha$, conditioned on $s$, is therefore
\begin{equation}
P(\alpha|s)
=\frac{1}{P(s)}
\sum_{\lambda\in L}
P_{\widetilde\sigma}
\!\left(\eta(s)+\xi_{\alpha}^{\perp}+\lambda\right).
\label{eq:local_gkp_mld}
\end{equation}
Equation~\eqref{eq:local_gkp_mld} is the exact logical-class maximum-likelihood target. The executable decoder used everywhere below is the closest-coset (CC), or max-log, approximation. For prior $\pi^0(\alpha)$ define
\begin{align}
D_{\alpha}^{\rm CC}(s)
&=\min_{\lambda\in L}
\frac{\left\|\eta(s)+\xi_{\alpha}^{\perp}+\lambda\right\|_2^2}
{\widetilde\sigma^2},\nonumber\\
\widehat P_{\rm CC}(\alpha\mid s)
&=\frac{\pi^0(\alpha)e^{-D_{\alpha}^{\rm CC}(s)/2}}
{\sum_{\beta}\pi^0(\beta)e^{-D_{\beta}^{\rm CC}(s)/2}},
\label{eq:local_gkp_closest_coset}
\end{align}
and $\widehat\alpha_{\rm CC}=\argmax_{\alpha}\widehat P_{\rm CC}(\alpha\mid s)$. Thus the full coset sum fixes the maximum-likelihood model, while the working decoder retains only the dominant lattice point in each logical coset. The two coincide when a single representative dominates but need not coincide at moderate noise. Consequently, every confidence used for erasure conversion is formed from the normalized CC weights and is calibrated against the declared channel. The syndrome-resolved graph decoder derived later is the partial-readout, signed generalization of this same closest-coset rule.

For the square code the formulas can be written without any abstract lattice notation.  If $\epsilon=(\epsilon_q,\epsilon_p)^T$, then
\begin{equation*}
s_{\square}(\epsilon)
=\sqrt2\begin{pmatrix}\epsilon_p\\-\epsilon_q\end{pmatrix}
\pmod1,
\qquad
\eta_{\square}(s)
=\frac{1}{\sqrt2}
\begin{pmatrix}-\widetilde s_2\\ \widetilde s_1\end{pmatrix}.
\end{equation*}
Writing the four logical classes as $(\mu,\nu)\in\F_2^2$ with shift $\mu e+\nu f$, Eq.~\eqref{eq:local_gkp_mld} becomes
\begin{align*}
d_{\mu\nu;k l}(s)
&:=\eta_{\square}(s)
+\frac{1}{\sqrt2}\begin{pmatrix}\mu\\\nu\end{pmatrix}
+\sqrt2\begin{pmatrix}k\\l\end{pmatrix},\\
P_{\square}(\mu,\nu\mid s)
&\propto
\sum_{k,l\in\Z}
\exp\!\left[-\frac{\|d_{\mu\nu;k l}(s)\|_2^2}
{2\widetilde\sigma^2}\right].\end{align*}
For the closest-coset implementation the square-code sector distance is
\begin{align*}
D_{\mu\nu}^{\square,\rm CC}(s)
&=\min_{k,l\in\mathbb Z}
\frac{\|d_{\mu\nu;kl}(s)\|_2^2}{\widetilde\sigma^2},\\
\widehat P_{\square,\rm CC}(\mu,\nu\mid s)
&\propto e^{-D_{\mu\nu}^{\square,\rm CC}(s)/2}.
\end{align*}
The hexagonal decoder is obtained by replacing $M_{\square}$, its dual representatives, and the Euclidean metric by those generated by $M_{\mathrm{hex}}$. This is the precise square-versus-hexagonal change exercised by the benchmark figures; the decoder rule itself is unchanged.

\label{subsec:loss_amplification_local_interface}

The physical model contains no independent hard-loss variable.  Standard Gaussian-channel composition gives, for pure loss of transmissivity $\eta_i$ followed by an ideal quantum-limited phase-insensitive amplifier of gain $1/\eta_i$ \cite{noh2019,hastrup2023,fukui2021},
\begin{equation}
\mathcal A_{1/\eta_i}\circ\mathcal L_{\eta_i}
=\mathcal G_{V_{\LA,i}},
\qquad
V_{\LA,i}=\frac{1-\eta_i}{\eta_i}I_{2m_i}
\label{eq:loss_amplification_channel_identity}
\end{equation}
in physical quadrature units.  Since Eq.~\eqref{eq:displacement} shifts a quadrature by $\sqrt{2\pi}\,\epsilon$, the displacement-coordinate covariance is
\begin{equation}
\Sigma_{\LA,i}(\eta_i)
=\frac{1-\eta_i}{2\pi\eta_i}I_{2m_i}.
\label{eq:loss_amplification_covariance}
\end{equation}
All displacement contributions are assembled once,
\begin{equation}
\Sigma_{\epsilon}
=\Sigma_{\rm prep}
+\bigoplus_{i=1}^{n}\Sigma_{\LA,i}(\eta_i)
+\Sigma_{\rm gate}+\Sigma_{\rm anc}+\Sigma_{\rm ff},
\label{eq:loss_amplification_total_covariance}
\end{equation}
with generated correlations retained.  Pre-amplification produces a different covariance, $(1-\eta_i)I/(2\pi)$, and is not the channel used in the reported simulations.

The two-parameter numerical benchmarks isolate preparation width and loss--amplification noise. Unless stated otherwise, their displacement covariance is
\begin{equation}
\Sigma_{\epsilon}^{\rm bench}
=\frac{\sigma_{\GKP}^{2}}{2\pi}I
+\bigoplus_i\frac{1-\eta_i}{2\pi\eta_i}I_{2m_i},
\label{eq:benchmark_covariance}
\end{equation}
while the additional gate, ancilla, feedforward, and detector contributions in Eq.~\eqref{eq:loss_amplification_total_covariance} are set to zero. Equation~\eqref{eq:loss_amplification_total_covariance}, rather than Eq.~\eqref{eq:benchmark_covariance}, remains the general model of the framework.

Before any destructive outer graph measurement, each block undergoes teleportation-based or nondestructive GKP recovery. Let $D_i^{\rm loc}$ denote the full continuous local record, $P_i^{\rm loc}$ the local observation map, and
\begin{equation*}
K_i=L_i\cap\ker P_i^{\rm loc}.
\end{equation*}
For residual local Pauli class $u_i\in\F_2^2$, the exact logical-class maximum-likelihood target is the wrapped local likelihood
\begin{equation}
W_{i,{\rm ML}}^{\rm loc}(u_i;D_i^{\rm loc})
=\sum_{[\lambda_i]\in(\Phi_i(u_i)+L_i)/K_i}
 g_i^{\rm loc}\!\left(D_i^{\rm loc}\mid P_i^{\rm loc}\lambda_i\right),
\label{eq:local_loss_amplification_weight}
\end{equation}
where the kernel uses the covariance in Eq.~\eqref{eq:loss_amplification_total_covariance}. The operational decoder retains only the dominant visible representative,
\begin{align}
D_{i,{\rm CC}}^{\rm loc}(u_i;D_i^{\rm loc})
&=\min_{[\lambda_i]\in(\Phi_i(u_i)+L_i)/K_i}
\mathcal Q_i^{\rm loc}(\lambda_i),\nonumber\\
\mathcal Q_i^{\rm loc}(\lambda_i)
&=-2\log g_i^{\rm loc}
\!\left(D_i^{\rm loc}\mid P_i^{\rm loc}\lambda_i\right),\nonumber\\
\widehat W_{i,{\rm CC}}^{\rm loc}(u_i;D_i^{\rm loc})
&=\exp[-D_{i,{\rm CC}}^{\rm loc}(u_i;D_i^{\rm loc})/2].
\label{eq:local_closest_coset_weight}
\end{align}
For a Gaussian kernel this is a covariance-weighted closest-vector problem. With prior $\pi_i^0$, the posterior distribution used by the graph controller is
\begin{equation}
\pi_i^{\rm loc}(u_i\mid D_i^{\rm loc})
=\frac{\pi_i^0(u_i)\widehat W_{i,{\rm CC}}^{\rm loc}(u_i;D_i^{\rm loc})}
{\sum_{v_i\in\F_2^2}\pi_i^0(v_i)\widehat W_{i,{\rm CC}}^{\rm loc}(v_i;D_i^{\rm loc})}.
\label{eq:local_loss_amplification_posterior}
\end{equation}
This normalized closest-coset distribution is retained as a posterior distribution even if its maximum-score correction is applied or tracked. The exact wrapped likelihood in Eq.~\eqref{eq:local_loss_amplification_weight} is retained only as the statistical target and for analytic validation.

A conservative basis-independent confidence is
\begin{align*}
\widehat u_i&=\argmax_{u_i}\pi_i^{\rm loc}(u_i\mid D_i^{\rm loc}),\\
\Gamma_i^{\rm blk}
&=\log\frac{\pi_i^{\rm loc}(\widehat u_i\mid D_i^{\rm loc})}
{1-\pi_i^{\rm loc}(\widehat u_i\mid D_i^{\rm loc})}.\end{align*}
For graph pathfinding, the measurement-specific confidence uses only the flip bit relevant to a proposed local Pauli $P\in\{X,Y,Z\}$:
\begin{align*}
p_i^P(c\mid D_i^{\rm loc})
&=\sum_{u_i:[P,u_i]_1=c}
\pi_i^{\rm loc}(u_i\mid D_i^{\rm loc}),\\
\widehat c_i^P&=\argmax_{c\in\F_2}p_i^P(c\mid D_i^{\rm loc}),\\
\Gamma_i^P&=\left|\log\frac{p_i^P(0\mid D_i^{\rm loc})}
{p_i^P(1\mid D_i^{\rm loc})}\right|.\end{align*}
A block may therefore be reliable for one measurement basis and unreliable for another. The conservative availability bit and the basis-resolved profile are
\begin{align}
a_i^{\rm blk}&=\mathbf1[\Gamma_i^{\rm blk}\geq\Gamma_{\rm loc}],\nonumber\\
a_i^P&=\mathbf1[\Gamma_i^P\geq\Gamma_{{\rm loc},P}].
\label{eq:basis_availability}
\end{align}

Operationally, the local recovery is a flagged instrument with accepted refreshed and deliberate-erasure outputs,
\begin{equation}
\left\{\mathcal J^{\rm acc}_{i,D},\mathcal J^{\rm era}_{i,D}\right\}_{D},
\qquad
\int dD\,[\mathcal J^{\rm acc}_{i,D}+\mathcal J^{\rm era}_{i,D}]
\ \text{trace preserving}.
\label{eq:local_flagged_instrument}
\end{equation}
The acceptance and rejection regions are
\begin{equation*}
\mathcal A_i=\{D:\Gamma_i(D)\geq\Gamma_{\rm loc}\},
\qquad
\mathcal E_i=\{D:\Gamma_i(D)<\Gamma_{\rm loc}\}.
\end{equation*}
For $D\in\mathcal A_i$, the refreshed output is admitted to the graph protocol. For $D\in\mathcal E_i$, the refreshed output is intentionally discarded and the classical flag $\locera$ is emitted. The calibrated local erasure and accepted-error probabilities are
\begin{align*}
\varepsilon_{\GKP,i}
&=\Pr[D_i^{\rm loc}\in\mathcal E_i],\\
&=\sum_{u_i}\pi_i^0(u_i)
\int_{\mathcal E_i}W_{i,{\rm ML}}^{\rm loc}(u_i;D)\,dD,\\
p_{i,{\rm err}}^{\rm acc}
&=\Pr[\widehat u_i\ne u_i,\,
D_i^{\rm loc}\in\mathcal A_i].
\end{align*}
Changing the threshold trades accepted Pauli error against a located erasure. This decision must occur before the destructive outer measurement; a low-confidence result obtained only afterward cannot generally be reinterpreted as a lost vertex because the measurement backaction and its unknown byproduct have already occurred.

\paragraph{Ideal-square validation of the closest-coset interface.}
The physical residual distribution can be evaluated without Monte Carlo for an ideal square GKP comb under the isotropic Gaussian loss--amplification channel. Let
\begin{equation*}
a_{\square}=\frac{1}{\sqrt2},\qquad
\widetilde\sigma^2(\eta)
=\widetilde\sigma_{\rm prep}^2
+\frac{1-\eta}{2\pi\eta},
\end{equation*}
and let $r=(r_1,r_2)$ be the centered pure-error coordinate in the dual fundamental cell
$\mathcal F_{\square}^{\perp}=[-a_{\square}/2,a_{\square}/2)^2$.
For $c\in\F_2$, the exact one-dimensional physical comb is
\begin{equation}
w_c(r_j;\widetilde\sigma)
=\sum_{k\in\Z}
\frac{\exp\![-(r_j+(2k+c)a_{\square})^2/(2\widetilde\sigma^2)]}
{\sqrt{2\pi}\widetilde\sigma}.
\label{eq:square_loss_amplification_axis_weight}
\end{equation}
The exact joint density of residual $r$ and residual Pauli class $u=(\mu,\nu)$ is
\begin{equation*}
p_{\square}(r,\mu,\nu)
=w_{\mu}(r_1;\widetilde\sigma)
 w_{\nu}(r_2;\widetilde\sigma).
\end{equation*}
It is normalized because the four logical shifts and the stabilizer translates partition $\R^2$. The closest-coset decision replaces each comb by its dominant peak,
\begin{align*}
\widehat w_{c,{\rm CC}}(r_j;\widetilde\sigma)
&=\exp\!\left[-\frac{1}{2\widetilde\sigma^2}
\min_{k\in\Z}(r_j+(2k+c)a_{\square})^2\right],\\
\widehat p_{\square,{\rm CC}}(r,\mu,\nu)
&=\widehat w_{\mu,{\rm CC}}(r_1;\widetilde\sigma)
\widehat w_{\nu,{\rm CC}}(r_2;\widetilde\sigma),\\
\widehat\pi_{\square,{\rm CC}}^{\rm loc}(\mu,\nu\mid r)
&=\frac{\widehat p_{\square,{\rm CC}}(r,\mu,\nu)}
{\sum_{a,b\in\F_2}\widehat p_{\square,{\rm CC}}(r,a,b)}.
\end{align*}
Let $\widehat u_{\rm CC}(r)=\argmax_u\widehat\pi_{\square,{\rm CC}}^{\rm loc}(u\mid r)$, let $\widehat p_{\max,{\rm CC}}(r)$ be the corresponding maximum normalized weight, and set $q_{\Gamma}=(1+e^{-\Gamma_{\rm loc}})^{-1}$. The CC acceptance region is
\begin{equation*}
\mathcal A_{\Gamma}^{\rm CC}
=\{r\in\mathcal F_{\square}^{\perp}:\widehat p_{\max,{\rm CC}}(r)\ge q_{\Gamma}\}.
\end{equation*}
The erasure and accepted-error probabilities combine the exact physical density with the CC decision rule,
\begin{align}
\varepsilon_{\GKP}^{\square}(\eta)
&=\int_{\mathcal F_{\square}^{\perp}\setminus\mathcal A_{\Gamma}^{\rm CC}}
\sum_{u\in\F_2^2}p_{\square}(r,u)\,d^2r,
\label{eq:square_loss_amplification_exact_erasure}\\
p_{{\rm err},\square}^{\rm acc}(\eta)
&=\int_{\mathcal A_{\Gamma}^{\rm CC}}
\sum_{u\in\F_2^2}p_{\square}(r,u)
\mathbf1[u\neq\widehat u_{\rm CC}(r)]\,d^2r.
\label{eq:square_loss_amplification_exact_error}
\end{align}
These expressions provide an analytic validation of the same closest-coset error--erasure interface used in the simulations. The hexagonal and multimode versions replace the square-lattice nearest-coset problem and physical comb by their corresponding lattice forms.

Algorithm~\ref{alg:local_interface} summarizes the local recovery and confidence-to-erasure interface.

\begin{prxalgorithm}{Local loss--amplification GKP interface}
\label{alg:local_interface}
\footnotesize
\begin{algorithmic}[1]
\Require Block transmissivity $\eta_i$; preparation, gate, ancilla, and detector covariances; local lattice and readout circuit; confidence thresholds.
\Ensure Local Pauli posterior and confidence profile; accepted refreshed block or $\locera$.
\State Apply $\mathcal A_{1/\eta_i}\circ\mathcal L_{\eta_i}$, or sample the equivalent displacement with covariance Eq.~\eqref{eq:loss_amplification_covariance}; do not sample a Bernoulli loss flag.
\State Perform teleportation-based or nondestructive GKP recovery and retain $D_i^{\rm loc}$.
\State Solve Eq.~\eqref{eq:local_closest_coset_weight} for all four local Pauli classes and normalize Eq.~\eqref{eq:local_loss_amplification_posterior}.
\State Compute $\Gamma_i^{\rm blk}$ and, when used, the basis-resolved profile $(\Gamma_i^X,\Gamma_i^Y,\Gamma_i^Z)$.
\If{the selected confidence exceeds threshold}
  \State Admit the refreshed output and pass $\pi_i^{\rm loc}$ upward.
\Else
  \State Deliberately discard the refreshed output, emit $\locera$, and retain the rejected syndrome or its integrated rejection likelihood.
\EndIf
\State \Return the complete local posterior distribution and the availability decision.
\end{algorithmic}
\end{prxalgorithm}

\paragraph{Modular homodyne records.}\label{subsec:modular_measurement_machinery}

Steane-type and teleportation-based GKP recovery circuits return a stabilizer syndrome modulo the lattice period together with a continuous residual in a chosen fundamental cell \cite{gkp2001,conrad2022}.  In the present normalization, a homodyne record $\widetilde q$ associated with $D(\xi)$ satisfies
\begin{equation*}
\widetilde q/\sqrt{2\pi}\pmod 1=s_{\xi}.
\end{equation*}
For a square-code logical-$\overline Z$ readout, the two logical hypotheses are the even and odd Gaussian combs in Eq.~\eqref{eq:square_loss_amplification_axis_weight} after conversion to physical quadrature units.  The retained analog statistic is their log-likelihood ratio, or its closest-coset approximation.  Fourier-dual, hexagonal, and multimode readouts are obtained by changing the measurement projection, lattice representatives, and covariance in Eqs.~\eqref{eq:local_gkp_mld} and \eqref{eq:local_closest_coset_weight}; no separate decoding principle is introduced.

\section{Graph codes as compiled GKP lattices}
\label{sec:graph_gkp_codes}

The graph-code layer is constructed inside the finite logical quotient of the product GKP code.  This is the point at which the binary graph formalism and the lattice formalism meet.  

\subsection{Outer graph code and Pauli embedding}

For a simple graph $G=(V,E)$, the associated graph state is the simultaneous $+1$ eigenstate of
\begin{equation}
K_v=X_v\prod_{w\in N(v)}Z_w,
\label{eq:graph_stabilizers}
\end{equation}
where $N(v)$ denotes the neighborhood of vertex $v$.  A phase-free $n$-qubit Pauli product is represented by
\begin{equation*}
u=(x\mid z)\in\F_2^{2n},
\end{equation*}
with binary symplectic pairing
\begin{equation}
[u,v]_{\DV}=x\!\cdot\!z'+z\!\cdot\!x'\pmod2.
\label{eq:symplectic_pairing}
\end{equation}

A graph-code progenitor contains one distinguished input vertex $\iota$ and $n$ code vertices.  Let $S_0\subset\F_2^{2(n+1)}$ be the stabilizer space of the progenitor, and let $\pi_{\rm in}$ and $\pi_c$ project respectively onto the input coordinate and the code coordinates.  For each one-qubit Pauli label $\alpha\in\F_2^2$, define the input fiber
\begin{equation}
F_\alpha=\{\widetilde s\in S_0:\pi_{\rm in}(\widetilde s)=\alpha\},
\qquad
\mathcal L_\alpha=\pi_c(F_\alpha).
\label{eq:fibers}
\end{equation}
When $\pi_{\rm in}|_{S_0}$ has rank two, the identity-input fiber produces the outer binary stabilizer code
\begin{equation*}
Q_G:=\mathcal L_0\subset\F_2^{2n},
\qquad \dim Q_G=n-1,
\end{equation*}
and each nonidentity fiber is an affine logical class
\begin{equation*}
\mathcal L_a=\ell_a+Q_G,
\qquad a\in\{X,Y,Z\}.
\end{equation*}
The space $Q_G$ is isotropic, every $\ell_a$ commutes with $Q_G$, and representatives may be chosen with $[\ell_X,\ell_Z]_{\DV}=1$.  These are the standard graph-code/stabilizer-fiber facts following from rank--nullity and isotropy of the progenitor stabilizer space \cite{gottesman1997,schlingemann2001,hein2006}; here they only fix the projection notation used below.

\paragraph{Product GKP quotient.}
\label{subsec:explicit_graph_gkp_compilation}

For graph vertex $i$, choose a one-qubit GKP lattice $L_i=L(M_i)\subset\R^{2m_i}$ and local logical displacement representatives $e_i,f_i\in L_i^\perp$ as in Eq.~\eqref{eq:local_e_f}.  Let $\iota_i:\R^{2m_i}\hookrightarrow V_{\rm loc}$ be the canonical block embedding and write
\begin{equation*}
E_i:=\iota_i(e_i),\qquad F_i:=\iota_i(f_i).
\end{equation*}
In blockwise quadrature order,
\begin{align*}
V_{\rm loc}&=\bigoplus_{i=1}^n\R^{2m_i},
&L_{\rm loc}&=\bigoplus_{i=1}^nL_i,\\
J_{\rm loc}&=\bigoplus_{i=1}^nJ_{2m_i},
&m_{\rm tot}&=\sum_i m_i.\end{align*}
The finite logical displacement module of the uncompiled product code is
\begin{equation*}
\mathsf H_{\rm loc}:=L_{\rm loc}^{\perp}/L_{\rm loc}
\simeq\bigoplus_{i=1}^n\F_2^2,
\end{equation*}
with quotient map
\begin{equation*}
\pi_{\rm loc}:L_{\rm loc}^{\perp}\longrightarrow\mathsf H_{\rm loc}.
\end{equation*}
The induced binary pairing is
\begin{equation*}
\bigl\langle[\xi],[\eta]\bigr\rangle_{\rm loc}
:=2\xi^TJ_{\rm loc}\eta\pmod2.
\end{equation*}
It is well defined because changing either representative by a vector of $L_{\rm loc}$ changes the symplectic product by an integer.

The canonical binary-to-logical map is the quotient isomorphism
\begin{equation*}
\overline\Phi:\F_2^{2n}\longrightarrow\mathsf H_{\rm loc},
\qquad
\overline\Phi(x\mid z)
=\left[\sum_i x_iE_i+\sum_i z_iF_i\right].
\end{equation*}
A concrete phase-space representative requires a section of $\pi_{\rm loc}$.  We fix the lift
\begin{equation*}
\Phi(u)=\sum_i x_iE_i+\sum_i z_iF_i\in L_{\rm loc}^{\perp}.
\end{equation*}
The quotient map $\overline\Phi$ is linear, whereas the chosen lift is linear only modulo the base lattice:
\begin{equation*}
\Phi(u\oplus v)-\Phi(u)-\Phi(v)\in L_{\rm loc}.
\end{equation*}
Moreover,
\begin{equation}
2\Phi(u)^TJ_{\rm loc}\Phi(v)
=[u,v]_{\DV}\pmod2.
\label{eq:Phi_symplectic}
\end{equation}
Thus $\overline\Phi$ identifies the binary Pauli module with the finite Heisenberg--Weyl quotient of the product GKP code; $\Phi$ is only a chosen lift used to write explicit displacement vectors and phases.

\subsection{Graph--GKP lattice compilation and two-block example}

The compiled graph--GKP lattice is the pullback of the outer binary code through the local logical quotient map:
\begin{equation}
L_G:=\pi_{\rm loc}^{-1}\!\left(\overline\Phi(Q_G)\right).
\label{eq:static_lattice_pullback}
\end{equation}
Equivalently,
\begin{equation*}
L_G=L_{\rm loc}+\spanZ\{\Phi(q):q\in Q_G\}.
\end{equation*}
Let $H_G$ be a binary generator matrix for $Q_G$, and let $R\in\R^{2n\times2m_{\rm tot}}$ have rows $(E_1^T,\ldots,E_n^T,F_1^T,\ldots,F_n^T)$ in the same binary-coordinate order.  An explicit overcomplete generator array is
\begin{equation}
\widetilde M_G=
\begin{pmatrix}
M_{\rm loc}\\ H_GR
\end{pmatrix},
\qquad
M_{\rm loc}=\bigoplus_iM_i.
\label{eq:overcomplete_MG}
\end{equation}
For square single-mode blocks, this becomes the standard binary-lifted square-lattice formula:
\begin{align*}
\widetilde M_G^{\square}
&=\frac{1}{\sqrt2}
\begin{pmatrix}2I_{2n}\\ H_G\end{pmatrix},\\
L_G^{\square}
&=\{v\in\R^{2n}:\sqrt2\,v\pmod2\in Q_G\}.\end{align*}
The hexagonal construction is obtained by the corresponding blockwise Gaussian code switch.  Hence the square/hexagonal comparison changes the local GKP basis and metric but leaves the outer graph code $Q_G$ unchanged.

A full-rank row basis $M_G$ is obtained by Smith/Hermite reduction of the integer module defined by Eq.~\eqref{eq:overcomplete_MG}, transporting the lattice phase sector with every basis change.  Since $Q_G$ is isotropic, $L_G$ is symplectically integral.  Faithfulness of $\overline\Phi$ gives
\begin{equation*}
L_G/L_{\rm loc}\simeq Q_G,
\qquad
[L_G:L_{\rm loc}]=2^{n-1}.
\end{equation*}
Because $|\det M_{\rm loc}|=2^n$ for $n$ local GKP qubits,
\begin{equation*}
|\det M_G|=\frac{2^n}{2^{n-1}}=2,
\end{equation*}
so the compiled lattice encodes one qubit.  Equivalently,
\begin{equation*}
L_G^\perp/L_G\simeq Q_G^\perp/Q_G\simeq\F_2^2.
\end{equation*}
Logical displacement representatives may be chosen as
\begin{equation*}
E_G:=\Phi(\ell_X),\qquad F_G:=\Phi(\ell_Z)
\end{equation*}
modulo $L_G$.  The ideal and finite-energy graph--GKP states are defined by
\begin{align*}
e^{i\phi_G(\lambda)}D(\lambda)|\mu_G\rangle
&=|\mu_G\rangle &&(\lambda\in L_G),\\
D(F_G)|\mu_G\rangle&=(-1)^\mu|\mu_G\rangle,\\
D(E_G)|\mu_G\rangle&=|\mu\oplus1\rangle_G,\\
|\widetilde\mu_{G,\beta}\rangle
&=\mathcal N_{G,\mu,\beta}e^{-\beta\hat N_{\rm tot}}|\mu_G\rangle.\end{align*}
The phase function $\phi_G$ extends the local lattice phase sector to the additional graph-code generators; Appendix~\ref{app:signed_compilation} gives the carry-aware construction.

\paragraph{Two-block compilation example.}
\label{subsec:worked_compilation}

Let the input vertex $\iota$ and code vertices $1,2$ form a triangle.  The progenitor stabilizers are
\begin{equation*}
K_\iota=X_\iota Z_1Z_2,\qquad
K_1=Z_\iota X_1Z_2,\qquad
K_2=Z_\iota Z_1X_2.
\end{equation*}
The phase-free projection of $K_1K_2$ generates the identity-input fiber, while $K_\iota$ and $K_1$ supply input-$X$ and input-$Z$ representatives.  Hence
\begin{equation*}
Q_G=\spanF\{Y_1Y_2\},
\qquad
\ell_X=Z_1Z_2,
\qquad
\ell_Z=X_1Z_2.
\end{equation*}
Replacing vertices $1,2$ by square GKP qubits and ordering coordinates as $(q_1,p_1,q_2,p_2)$ gives
\begin{align*}
M_{\rm loc}&=\sqrt2 I_4,\\
\Phi(Y_1Y_2)&=\frac{1}{\sqrt2}(1,1,1,1)^T.\end{align*}
Thus
\begin{equation*}
\widetilde M_G=
\begin{pmatrix}
\sqrt2&0&0&0\\
0&\sqrt2&0&0\\
0&0&\sqrt2&0\\
0&0&0&\sqrt2\\
1/\sqrt2&1/\sqrt2&1/\sqrt2&1/\sqrt2
\end{pmatrix}.
\end{equation*}
The additional row gives index two, so $|\det M_G|=2$.  The remaining logical displacement representatives are
\begin{align*}
E_G&=\Phi(\ell_X)=\frac{1}{\sqrt2}(0,1,0,1)^T,\\
F_G&=\Phi(\ell_Z)=\frac{1}{\sqrt2}(1,0,0,1)^T,\end{align*}
with $E_G^TJ_{\rm loc}F_G=-1/2\equiv1/2\pmod\Z$.  This example displays the same pullback construction used for every numerical graph.

\subsection{Decoder-conditioned branch restrictions as quotient submodules}
\label{subsec:branch_geometry}

The local GKP interface is completed before the destructive outer graph measurement. A realized graph branch is therefore specified by an actual outer-measurement word and a decoder-generated availability record,
\begin{equation*}
\mathbf q_b=(q_{b,1},\ldots,q_{b,n}),
\qquad
\mathbf a_b=(a_{b,1},\ldots,a_{b,n}).
\end{equation*}
For the conservative interface, $a_{b,i}\in\{0,1\}$. For the basis-resolved interface, the graph policy first consults $(a_i^X,a_i^Y,a_i^Z)$ from Eq.~\eqref{eq:basis_availability}; if the proposed basis is rejected, the refreshed block is discarded, $a_{b,i}$ is set to zero, and another representative is sought. The symbol $\locera$ denotes this local-decoder outcome, not direct observation that the oscillator was absent.

After the basis and availability decisions are fixed, define the local compatibility subspace
\begin{equation}
\mathsf C_i(q_i,a_i)=
\begin{cases}
\F_2^2,&a_i=1,\ q_i=I,\\
\spanF\{m(q_i)\},&a_i=1,\ q_i\in\{X,Y,Z\},\\
\{0\},&a_i=0,
\end{cases}
\label{eq:local_compatibility}
\end{equation}
where $m(X)=(1,0)$, $m(Y)=(1,1)$, and $m(Z)=(0,1)$. The global compatibility and measurement-supported submodules are
\begin{align*}
\mathsf C_b&:=\bigoplus_{i=1}^n\mathsf C_i(q_{b,i},a_{b,i}),\\
\mathsf R_b&:=\{u\in\mathsf C_b:q_{b,i}=I\Rightarrow u_i=0\}.\end{align*}
The compatible and measured outer-check spaces are
\begin{equation}
Q_b:=Q_G\cap\mathsf C_b,
\qquad
Q_b^{\meas}:=Q_G\cap\mathsf R_b,
\label{eq:branch_Q_spaces}
\end{equation}
and the logical fibers are the affine intersections
\begin{align*}
\mathcal L_{a,b}&:=(\ell_a+Q_G)\cap\mathsf C_b,\\
\mathcal L_{a,b}^{\meas}&:=(\ell_a+Q_G)\cap\mathsf R_b.\end{align*}

Because the intersection of an affine coset with a linear subspace is either empty or an affine coset of the intersected translation space, any nonempty branch fiber obeys
\begin{align}
\mathcal L_{a,b}&=q_{a,b}+Q_b,
& q_{a,b}&\in\mathcal L_{a,b},
\label{eq:branch_fiber_coset}\\
\mathcal L_{a,b}^{\meas}&=q_{a,b}^{\meas}+Q_b^{\meas},
& q_{a,b}^{\meas}&\in\mathcal L_{a,b}^{\meas}.
\label{eq:real_branch_fiber_coset}
\end{align}

A branch realizes an indirect Pauli measurement precisely when $\mathcal L_{a,b}^{\meas}$ is nonempty. Its reconstruction lattice is the same quotient pullback applied to the restricted outer code,
\begin{equation}
L_{G,b}:=\pi_{\rm loc}^{-1}(\overline\Phi(Q_b))
=L_{\rm loc}+\spanZ\{\Phi(q):q\in Q_b\}.
\label{eq:branch_lattice}
\end{equation}

\begin{theorem}[Branch-compatible graph--GKP compilation]
\label{thm:branch_compilation}
For $u,v\in\mathsf C_b$,
\begin{equation*}
u\oplus v\in Q_b
\quad\Longleftrightarrow\quad
\Phi(u)+L_{G,b}=\Phi(v)+L_{G,b}.
\end{equation*}
Consequently $\mathcal L_{a,b}$ is represented by the single displacement coset $\Phi(q_{a,b})+L_{G,b}$.
\end{theorem}
A proof, using only the quotient pullback and faithfulness of $\overline\Phi$, is given in Appendix~\ref{app:branch_compilation}.

The branch lattice is used only for observable reconstruction and signed evidence. Physical errors remain periodic under the static lattice $L_G$. The complete decoder record is
\begin{equation*}
D_b=\left(D_{1:n}^{\rm loc},\mathbf a_b,
 y_b^{\rm out},r_b^{\rm out}\right),
\end{equation*}
where all attempted local GKP recovery records are retained, including rejected records. The observation map has the causal block form
\begin{equation}
P_b^{\rm tot}=
\begin{pmatrix}
P^{\rm loc}\\
P_b^{\rm out}(\mathbf a_b)
\end{pmatrix}.
\label{eq:total_observation_map}
\end{equation}
The local rows are present before branch selection. A local erasure suppresses only the future outer-measurement rows on that block. Pure attenuation changes the covariance of the local record through Eq.~\eqref{eq:loss_amplification_covariance}; it does not delete an observation row.

The availability mask and branch are deterministic functions of the same analog records used for error inference,
\begin{equation}
b=\mathcal B_{\Gamma}(D_{1:n}^{\rm loc}).
\label{eq:branch_policy}
\end{equation}
Consequently their probability belongs to the physical likelihood. It is generally incorrect to multiply an independently sampled mask into the simulation or to assign an erased block a unit likelihood factor.

The three erasure events used later are now unambiguous:
\begin{equation*}
\begin{aligned}
\locera&:\ \text{local GKP abstention},\\
\accera&:\ \mathcal L_{a,b}^{\meas}=\varnothing,\\
\confera&:\ \Gamma_b^{\rm fr}<\Gamma_{\rm th}^{\rm fr}.
\end{aligned}
\end{equation*}
Figure~\ref{fig:overview} summarizes the static code, the local GKP interface, and the decoder-conditioned branch.

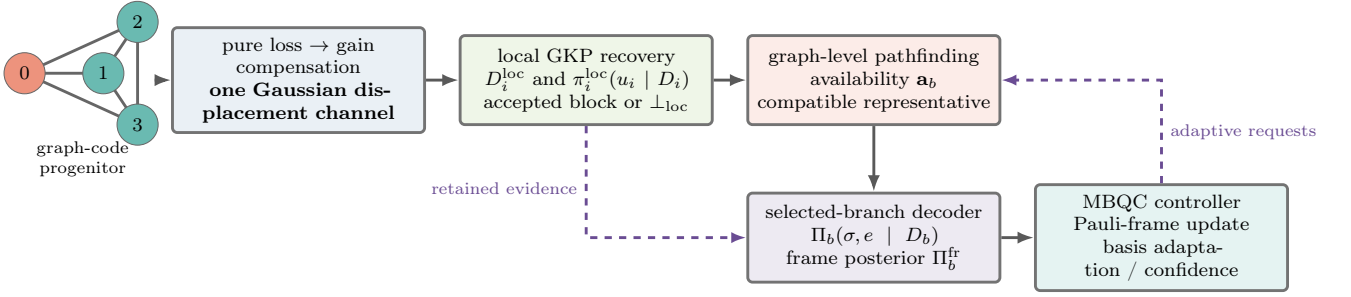
\begin{figure*}[!t]
\centering
\resizebox{0.98\textwidth}{!}{\begin{tikzpicture}[
  >=Latex,
  every node/.style={font=\footnotesize},
  box/.style={draw=black!55,rounded corners=2pt,very thick,align=center,minimum height=13mm,text width=34mm,inner sep=4pt},
  arrow/.style={-{Latex[length=2.4mm]},very thick,draw=black!65},
  soft/.style={-{Latex[length=2.2mm]},gkppurple,dashed,very thick}
]
\begin{scope}[xshift=0mm,yshift=25mm]
  \coordinate (g0) at (0,0.7); \coordinate (g1) at (1.15,0.7);
  \coordinate (g2) at (1.65,1.45); \coordinate (g3) at (1.65,-0.05);
  \draw[very thick,black!60] (g0)--(g1)--(g2)--(g3)--(g1) (g0)--(g2) (g0)--(g3);
  \node[circle,draw=black!60,fill=gkporange!75,minimum size=6mm,inner sep=0pt] at (g0) {0};
  \node[circle,draw=black!60,fill=gkpteal!70,minimum size=6mm,inner sep=0pt] at (g1) {1};
  \node[circle,draw=black!60,fill=gkpteal!70,minimum size=6mm,inner sep=0pt] at (g2) {2};
  \node[circle,draw=black!60,fill=gkpteal!70,minimum size=6mm,inner sep=0pt] at (g3) {3};
  \node[font=\scriptsize,align=center] at (0.85,-0.55) {graph-code\\progenitor};
\end{scope}
\node[box,fill=gkpblue!10] (noise) at (4.0,3.1) {pure loss $\rightarrow$ gain compensation\\\textbf{one Gaussian displacement channel}};
\node[box,fill=gkpgreen!12] (local) at (8.2,3.1) {local GKP recovery\\$D_i^{\rm loc}$ and $\pi_i^{\rm loc}(u_i\mid D_i)$\\accepted block or $\bot_{\rm loc}$};
\node[box,fill=gkporange!13] (branch) at (12.4,3.1) {graph-level pathfinding\\availability $\mathbf a_b$\\compatible representative};
\node[box,fill=gkppurple!12] (post) at (12.4,0.8) {selected-branch decoder\\$\Pi_b(\sigma,e\mid D_b)$\\frame posterior $\Pi_b^{\rm fr}$};
\node[box,fill=gkpteal!12] (ctrl) at (16.6,0.8) {MBQC controller\\Pauli-frame update\\basis adaptation / confidence};
\draw[arrow] (1.9,3.1)--(noise.west);
\draw[arrow] (noise)--(local);
\draw[arrow] (local)--(branch);
\draw[arrow] (branch)--(post);
\draw[arrow] (post)--(ctrl);
\draw[soft] (local.south) |- node[pos=0.28,left,font=\scriptsize]{retained evidence} (post.west);
\draw[soft] (ctrl.north) |- node[pos=0.25,right,font=\scriptsize]{adaptive requests} (branch.east);
\end{tikzpicture}}
\caption{Graph--GKP decoding pipeline. Pure loss followed by ideal gain compensation produces one additive Gaussian displacement channel. A local GKP recovery converts its continuous syndrome into a Pauli posterior and, at low confidence, a located decoder erasure $\locera$. The resulting availability record restricts the quotient only for graph representative selection. Physical error weights use the full static lattice and the joint selected-branch likelihood, which is then pushed to the MBQC Pauli frame.}
\label{fig:overview}
\end{figure*}

\section{Loss-tolerant logical measurements with graph--GKP codes}
\label{sec:logical_measurements}

After Algorithm~\ref{alg:local_interface}, the graph layer selects each destructive measurement from the basis-resolved availability profile while retaining every local recovery record in the global likelihood.

\subsection{Decoder-conditioned Pauli and equatorial measurements}

Let $a\in\{X,Y,Z\}$ be the requested logical Pauli observable. For a fixed accepted branch choose a reference
\begin{equation*}
q^0_{a,b}\in\mathcal L_{a,b}^{\meas}.
\end{equation*}
Any other measured representative has the form $q_{a,b}=q^0_{a,b}+s_{a,b}$ with $s_{a,b}\in Q_b^{\meas}$. The cocycle-aware product of local outcomes is transported to the fixed reference by subtracting the observed signed stabilizer character. Together with the tracked deterministic frame this gives the raw branch character
\begin{equation*}
\chi_b^{\rm raw}(a;q^0_{a,b})\in\F_2.
\end{equation*}
The final logical inference is independent of the operational representative and of the chosen reference section.

The adaptive policy differs from a hardware-loss decision tree in one essential respect. Before measuring a candidate local factor $P$, the controller queries the local confidence $\Gamma_i^P$. If that basis is accepted, the refreshed block is measured and the outcome is appended to the outer record. If it is rejected, the block is deliberately discarded, the availability entry is set to zero, and the policy searches for another representative. Since the rejected local syndrome was obtained before this decision, it remains part of $D_b$ even though no outer measurement is performed on that block.

This algorithm exposes the correct causal ordering: local GKP recovery and abstention precede graph pathfinding, whereas the global syndrome--frame posterior is evaluated only after the selected outer measurements have been performed.

\paragraph{Equatorial measurements.}

Universal MBQC requires a non-Clifford ingredient, commonly represented by
\begin{equation*}
A(\theta)=X\cos\theta+Y\sin\theta.
\end{equation*}
At the graph-code level, the encoded state is teleported to an output block and the output is measured through a non-Gaussian logical resource. For candidate output $o$, define
\begin{equation*}
\mathsf T_{b,o}
=\{u\in\mathsf C_b:u_i=0\ \text{for every unmeasured }i\neq o\}.
\end{equation*}
The stabilizer-pathfinding candidates are
\begin{align}
\mathcal P_{b,o}:=\bigl\{(q_X,q_Z):{}&q_X\in(\ell_X+Q_G)\cap\mathsf T_{b,o},\nonumber\\
&q_Z\in(\ell_Z+Q_G)\cap\mathsf T_{b,o},\nonumber\\
&[q_{X,o},q_{Z,o}]_1=1,\nonumber\\
&[q_{X,i},q_{Z,i}]_1=0\quad\forall i\neq o\bigr\}.
\label{eq:SPC_pair_set}
\end{align}
A branch exists when at least one output has $\mathcal P_{b,o}\neq\varnothing$ after the basis-specific local confidence tests.

A generic $A(\theta)$ measurement is not a displacement observable. We therefore isolate its non-Gaussianity in an encoded magic or Choi resource and use GKP Bell and modular-homodyne measurements for the remaining conditional circuit. For Pauli frame $X^aZ^b$,
\begin{equation*}
(X^aZ^b)^\dagger A(\theta)(X^aZ^b)
=(-1)^bA((-1)^a\theta).
\end{equation*}

The operator description below permits a noisy encoded resource and a noisy terminal readout. In the numerical panels labeled $A(\theta)$, however, the encoded non-Gaussian resource is treated as supplied ideally and no additional magic-state preparation, distillation, or injection error is assigned. Those panels therefore quantify the graph--GKP transport, branch selection, modular measurements, and syndrome-resolved Pauli-frame adaptation surrounding the terminal $A(\theta)$ interface; they are not end-to-end logical non-Clifford error rates.

Figure~\ref{fig:branch_magic} shows both the branch restriction produced by a local abstention and the magic-assisted equatorial interface. Thus an $X$-frame bit reverses the angle and a $Z$-frame bit flips the reported outcome. The uncertain physical contribution to this frame is obtained from the posterior in Sec.~\ref{subsec:outcome_likelihood}; it must not be replaced by a hard local decision before the teleportation branch is decoded.  The Pauli and equatorial cases share the same decoder-conditioned search and differ only in their terminal interface.

\begin{figure*}[!t]
\centering
\resizebox{0.98\textwidth}{!}{\begin{tikzpicture}[
  >=Latex,
  every node/.style={font=\footnotesize},
  panel/.style={draw=black!28,rounded corners=3pt,fill=black!1},
  qnode/.style={circle,draw=black!65,very thick,
                minimum size=7mm,inner sep=0pt},
  box/.style={draw=black!55,rounded corners=2pt,very thick,
              align=center,text width=27mm,
              minimum height=12mm,inner sep=3pt},
  arrow/.style={-{Latex[length=2.3mm]},
                very thick,draw=black!65}
]

\node[
  panel,
  minimum width=78mm,
  minimum height=45mm,
  anchor=south west
] at (0,0) {};

\node[font=\small\bfseries,align=center] at (3.9,4.05)
  {(a) Branch restriction after local abstention};

\coordinate (n0) at (1.1,2.3);
\coordinate (n1) at (2.7,3.25);
\coordinate (n2) at (4.45,2.8);
\coordinate (n3) at (5.15,1.3);
\coordinate (n4) at (2.9,1.0);

\draw[black!50,thick]
  (n0)--(n1)--(n2)--(n3)--(n4)--(n1)
  (n0)--(n4);

\draw[gkpblue,very thick]
  (n0)--(n1)--(n2);

\draw[gkporange,very thick,dashed]
  (n0)--(n4)--(n3);

\foreach \p/\lab in {n0/0,n1/1,n2/2,n3/3}{
  \node[qnode,fill=gkpteal!60] at (\p) {\lab};
}

\node[qnode,fill=gkporange!75] at (n4) {4};

\draw[gkporange,very thick]
  ($(n4)+(-0.22,-0.22)$)--($(n4)+(0.22,0.22)$);

\draw[gkporange,very thick]
  ($(n4)+(-0.22,0.22)$)--($(n4)+(0.22,-0.22)$);

\node[
  font=\scriptsize,
  align=center,
  text width=68mm
] at (3.9,0.25)
  {blue: selected branch; orange dashed: discarded candidate;
   the rejected syndrome remains in $D_b$.};

\node[
  panel,
  minimum width=120mm,
  minimum height=45mm,
  anchor=south west
] at (8.2,0) {};

\node[font=\small\bfseries,align=center] at (14.2,4.05)
  {(b) Magic-assisted equatorial interface};

\node[box,fill=gkpblue!10] (data) at (10.15,2.95)
  {data block\\$\rho_{\rm out}$};

\node[box,fill=gkporange!13] (magic) at (10.15,1.45)
  {encoded resource\\$\lvert A(\theta)\rangle$};

\node[box,fill=gkpgreen!13] (bell) at (14.15,2.20)
  {GKP Bell\\measurement};

\node[
  box,
  fill=gkppurple!12,
  text width=30mm
] (frame) at (17.75,2.20)
  {frame posterior\\
   $\Pi_b^{\rm fr}(f\mid D_b)$\\
   $\widehat f_b=(\widehat a,\widehat b)$,
   $\Gamma_b^{\rm fr}$};

\draw[arrow] (data.east)--(bell.west);
\draw[arrow] (magic.east)--(bell.west);
\draw[arrow] (bell.east)--(frame.west);

\node[
  box,
  fill=gkpteal!12,
  text width=42mm,
  minimum height=14mm
] (out) at (17.75,0.65)
  {feedforward $X^{\widehat a}Z^{\widehat b}$\\
   adapted angle
   $\theta_{\rm phys}=(-1)^{\widehat a}\theta$\\
   corrected bit
   $m_{\rm corr}=m_{\rm raw}\oplus\widehat b$};

\draw[arrow] (frame.south)--(out.north);

\end{tikzpicture}}
\caption{Decoder-conditioned branch restriction and non-Pauli terminal interface. (a) A locally ambiguous block is converted to $\locera$ before the outer graph measurements, so candidates requiring that block are removed, while its rejected local syndrome remains in the global record. (b) An encoded non-Gaussian resource and a GKP Bell measurement implement the equatorial interface. The posterior frame controls the future angle and reported bit. The numerical $A(\theta)$ benchmarks treat this terminal resource ideally and evaluate the surrounding transport and frame layer.}
\label{fig:branch_magic}
\end{figure*}

Algorithm~\ref{alg:logical_branch} gives one branch-selection procedure for both Pauli and equatorial logical measurements.

\begin{prxalgorithm}{Decoder-conditioned logical measurement branch}
\label{alg:logical_branch}
\footnotesize
\begin{algorithmic}[1]
\Require Requested task $T\in\{a,A(\theta)\}$; graph-code data; local GKP posteriors and basis-resolved confidences; precompiled pathfinding policy; calibrated outer readout model; declared terminal non-Gaussian resource model (ideal in the reported $A(\theta)$ benchmarks).
\Ensure $\accera$, or accepted branch data $b$, a logical representative or teleportation pair, and the retained measurement record.
\State Initialize the availability and basis record from the accepted local GKP outputs; retain every local recovery record, including rejected ones.
\If{$T=a\in\{X,Y,Z\}$}
  \State Set the viable set to $(\ell_a+Q_G)\cap\mathsf C_b$ and search for a representative in $\mathsf R_b$.
\Else
  \State For each candidate output $o$, construct $\mathcal P_{b,o}$ from Eq.~\eqref{eq:SPC_pair_set} and search for a valid teleportation pair.
\EndIf
\While{no valid representative or teleportation pair survives}
  \State Query the next required local Pauli basis and its confidence from the precompiled policy.
  \If{the required basis is rejected}
    \State Discard the refreshed block, set its availability to zero, record $\locera$, retain its syndrome, and update the viable sets.
  \Else
    \State Perform the selected outer modular measurement and append its signed discrete and analog record.
  \EndIf
  \If{all viable representatives or outputs are exhausted} \State \Return $\accera$ with the retained records. \EndIf
\EndWhile
\If{$T=A(\theta)$}
  \State Perform the accepted non-Gaussian output interface and retain the conditional operators $E_{m,\theta}^{(\sigma,e,b)}$.
\Else
  \State Choose a fixed reference $q^0_{a,b}$ and construct the reference-normalized raw character $\chi_b^{\rm raw}$.
\EndIf
\State Assemble $D_b$, the selected observation map $P_b^{\rm tot}$, and the branch-transfer data.
\State \Return the accepted branch and complete retained record.
\end{algorithmic}
\end{prxalgorithm}

\paragraph{Graph-level observation map.}\label{subsec:modular_homodyne_graph}

The local recovery rows $P^{\rm loc}$ are present for every attempted block.  If $S_b$ is the net Gaussian transformation before the selected outer detections, $N_b^{\rm out}$ contains the candidate outer displacement measurements, and $\Pi_b^{\rm out}(\mathbf a_b)$ keeps only those actually performed, then
\begin{equation*}
P_b^{\rm out}(\mathbf a_b)
=\Pi_b^{\rm out}(\mathbf a_b)N_b^{\rm out}J_{\rm loc}S_b.
\end{equation*}
Together with Eq.~\eqref{eq:total_observation_map}, this gives the retained residual $r_b=\operatorname{red}_{\mathcal F_b}(P_b^{\rm tot}\epsilon+n_b)$ and covariance
\begin{equation*}
\Sigma_{b,\obs}
=P_b^{\rm tot}\Sigma_\epsilon(P_b^{\rm tot})^T
+\Sigma_{b,\det}.
\end{equation*}
Attenuation changes $\Sigma_\epsilon$ but removes no row.  A decoder erasure suppresses only a future outer row; an unretained record is marginalized over its calibrated acceptance or rejection region.

\section{Measurement-based analog error correction from the GKP quotient}
\label{sec:analog_qec}

The graph front end identifies an accessible logical representative and its signed raw value.  The remaining task is the standard error-correction decomposition of the outer binary code, followed by the GKP logical-coset likelihood and a branch-dependent pushforward.  

\subsection{Outer and observed syndrome characters}
\label{subsec:syndrome_sectors}

The coordinate-free syndrome of a physical Pauli label $u\in\F_2^{2n}$ is the character
\begin{equation*}
\begin{aligned}
\syn_G(u)&\in Q_G^*:=\operatorname{Hom}(Q_G,\F_2),\\
\bigl(\syn_G(u)\bigr)(q)&=[q,u]_{\DV}.
\end{aligned}
\end{equation*}
Choose an ordered basis of $Q_G$ and stack it into a full-row-rank matrix $H_G$.  With
\begin{equation*}
J_{\DV}=\begin{pmatrix}0&I_n\\I_n&0\end{pmatrix}
\quad\text{over }\F_2,
\end{equation*}
the coordinate vector of the same character is
\begin{equation*}
\sigma=\syn_G(u)=H_GJ_{\DV}u^T\in\F_2^{n-1}.
\end{equation*}
A binary pure-error section is a map
\begin{equation*}
r_G:Q_G^*\longrightarrow\F_2^{2n},
\qquad
\syn_G(r_G(\sigma))=\sigma,
\qquad r_G(0)=0.
\end{equation*}
Its phase-space lift is
\begin{equation*}
\eta_G(\sigma):=\Phi(r_G(\sigma))\in L_{\rm loc}^{\perp}.
\end{equation*}
This notation separates the measured syndrome from a chosen pure error.  The binary section need not be linear; changing it merely relabels residual logical classes.

Choose representatives $\ell_e$ of $Q_G^\perp/Q_G$, with $e\in\F_2^2\simeq\{I,X,Z,Y\}$.  The standard stabilizer decomposition is
\begin{equation}
\F_2^{2n}
=\bigsqcup_{\sigma\in Q_G^*}
 \bigsqcup_{e\in\F_2^2}
 \bigl(r_G(\sigma)+\ell_e+Q_G\bigr).
\label{eq:syndrome_logical_partition}
\end{equation}
Thus every physical error is specified by an outer syndrome character and a residual logical class.  Equation~\eqref{eq:syndrome_logical_partition} is the standard stabilizer syndrome/normalizer decomposition \cite{gottesman1997}.  Changing the pure-error section only relabels $e$ within each fixed $\sigma$ and therefore leaves every physical coset and its controller pushforward unchanged.

\paragraph{Observed-syndrome restriction.}

A destructive branch generally reveals only the restriction of the full syndrome to the checks that were actually measured, that is, to the measured-check subspace $Q_b^{\meas}=Q_G\cap\mathsf R_b$ of Eq.~\eqref{eq:branch_Q_spaces}.  Let $\iota_b:Q_b^{\meas}\hookrightarrow Q_G$ be the inclusion.  The abstract observed-syndrome map is the dual restriction
\begin{equation*}
\operatorname{res}_b:=\iota_b^*:Q_G^*\longrightarrow(Q_b^{\meas})^*.
\end{equation*}
The signed measurement record produces an observed character
\begin{equation*}
z_b\in(Q_b^{\meas})^*.
\end{equation*}
After choosing a basis $q_{b,k}$ of $Q_b^{\meas}$ and writing $q_{b,k}=c_{b,k}H_G$, the matrix with rows $c_{b,k}$ is denoted by $C_b$.  It is only the coordinate matrix of $\operatorname{res}_b$.

Linearity of the stabilizer character immediately gives, for an ideal signed record generated by physical Pauli error $u$,
\begin{equation}
z_b=\operatorname{res}_b\!\left(\syn_G(u)\right)
=C_b\syn_G(u),
\label{eq:observed_syndrome_restriction}
\end{equation}
and hence
\begin{equation}
\Sigma_b
=\{\sigma\in Q_G^*: \operatorname{res}_b(\sigma)=z_b\}
=\{\sigma:C_b\sigma=z_b\}.
\label{eq:Sigma_b}
\end{equation}

The signs in $z_b$ are accumulated with the Hermitian Pauli convention of Appendix~\ref{app:signed_compilation}, including local $Y$ phases and the tracked Pauli frame.  When cell bits are noisy, the hard restriction is replaced by a calibrated likelihood on $(Q_b^{\meas})^*$; it is not multiplied independently of the analog data unless that conditional independence is part of the readout model.

Figure~\ref{fig:observable_error_layers} emphasizes the central separation.  The submodules $\mathsf C_b$ and $\mathsf R_b$ constrain only the observable and check representatives reconstructed from the branch.  A physical conjugate error need not lie in either submodule.

\begin{figure*}[!t]
\centering
\resizebox{0.96\textwidth}{!}{\begin{tikzpicture}[
  >=Latex,
  every node/.style={font=\footnotesize},
  box/.style={draw=black!55,rounded corners=2pt,very thick,align=center,text width=61mm,minimum height=11mm,inner sep=4pt},
  arrow/.style={-{Latex[length=2.4mm]},very thick,draw=black!65},
  cross/.style={-{Latex[length=2.2mm]},gkppurple,dashed,very thick},
  paneltitle/.style={font=\bfseries\small,align=center},
  question/.style={font=\scriptsize,align=center,text width=68mm}
]
\node[paneltitle] at (3.6,5.1) {(a) Observable compatibility};
\node[question] at (3.6,4.55) {Which representatives remain measurable after the availability restriction?};
\node[box,fill=gkpblue!11] (c) at (3.6,3.65) {branch compatibility module $\mathsf C_b$};
\node[box,fill=gkpgreen!12] (r) at (3.6,2.25) {measurement-supported module $\mathsf R_b$\\and restricted check character};
\node[box,fill=gkppurple!12] (p) at (3.6,0.85) {observable reconstruction and output map $P_b$};
\draw[arrow] (c)--(r); \draw[arrow] (r)--(p);
\node[paneltitle] at (12.0,5.1) {(b) Physical-error compatibility};
\node[question] at (12.0,4.55) {Which global sectors remain plausible for the retained analog record?};
\node[box,fill=gkpblue!11] (s) at (12.0,3.65) {static outer-code sectors $r_G(\sigma)+\ell_e+Q_G$};
\node[box,fill=gkpgreen!12] (l) at (12.0,2.25) {compiled lattice cosets $\mathcal C_{\sigma,e}=\tau_{\sigma,e}+L_G$};
\node[box,fill=gkporange!13] (w) at (12.0,0.85) {selected-branch likelihood and posterior weights};
\draw[arrow] (s)--(l); \draw[arrow] (l)--(w);
\draw[cross] (w.south) -- ++(0,-0.55) -| (p.south);
\node[font=\scriptsize,fill=white,inner sep=1pt] at (7.8,-0.15) {push forward only after sector scoring};
\end{tikzpicture}}
\caption{Observable compatibility is not physical-error compatibility.  The branch word $b$ restricts the quotient modules used to reconstruct the observable and the measured syndrome character.  Physical errors are classified by the full outer-code partition $r_G(\sigma)+\ell_e+Q_G$, compiled into static $L_G$ cosets, and only then pushed forward through $P_b$.}
\label{fig:observable_error_layers}
\end{figure*}
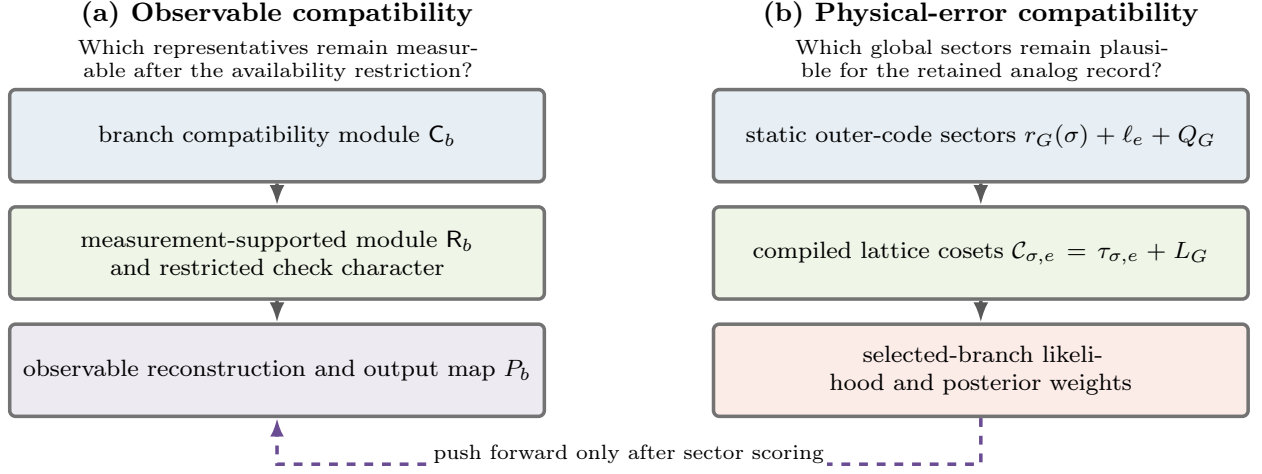

\subsection{Syndrome-resolved sectors and selected-branch likelihood}
\label{subsec:branch_likelihood}

For each syndrome--logical pair define the binary affine sector
\begin{equation*}
U_{\sigma,e}:=r_G(\sigma)+\ell_e+Q_G
\end{equation*}
and its phase-space shift
\begin{equation*}
\tau_{\sigma,e}:=\eta_G(\sigma)+\Phi(\ell_e).
\end{equation*}
The physical displacement sector is the static coset
\begin{equation*}
\mathcal C_{\sigma,e}=\tau_{\sigma,e}+L_G.
\end{equation*}
Using $\eta_G(\sigma)+\Phi(\ell_e)$ rather than $\Phi(r_G(\sigma)+\ell_e)$ makes the role of the chosen lift explicit; the two expressions differ at most by a base-lattice carry and therefore define the same $L_G$ coset.

\paragraph{Complete-syndrome specialization.}
If every compiled stabilizer syndrome is available, the exact and closest-coset rules are obtained directly from Eqs.~\eqref{eq:local_gkp_mld} and \eqref{eq:local_gkp_closest_coset} by replacing $L$ with $L_G$ and using a full-rank generator $M_G$.  Writing
\begin{align*}
s_G(\epsilon)&=M_GJ_{\rm loc}\epsilon\pmod1,\\
\eta_{M_G}(s_G)&=(M_GJ_{\rm loc})^{-1}\widetilde s_G,
\end{align*}
the implemented sector distance is
\begin{align*}
\Delta_{G,e}(\lambda)
&:=\eta_{M_G}(s_G)+\xi_e^\perp+\lambda,\\
D_G^{\rm CC}(e\mid s_G)
&=\min_{\lambda\in L_G}
\Delta_{G,e}(\lambda)^T\Sigma_G^{-1}\Delta_{G,e}(\lambda).
\end{align*}
A destructive graph branch reveals only a projection of this comb, which leads to the quotient problem below.

\paragraph{Selected-branch likelihood.}
\label{subsec:selected_branch_likelihood}

By Eq.~\eqref{eq:branch_policy}, the availability mask and branch are deterministic functions of the retained local recovery records. The exact selected-event channel is therefore
\begin{equation}
p_b(D_b,b\mid\sigma,e)
=\mathbf1[b=\mathcal B_{\Gamma}(D_{1:n}^{\rm loc})]
\,p_{\LA}(D_b\mid\sigma,e),
\label{eq:selected_branch_channel}
\end{equation}
where $p_{\LA}$ is induced by Eq.~\eqref{eq:loss_amplification_total_covariance}, the local instruments, and the selected outer measurements. The indicator displays the selected-event conditioning explicitly; once $D_{1:n}^{\rm loc}$ is stored, it is fixed rather than sampled independently.

For a factorized local readout, define the acceptance and rejection factors
\begin{align}
H_i^{\rm acc}(u_i;D_i)
&=\mathbf1[D_i\in\mathcal A_i]W_i^{\rm loc}(u_i;D_i),\nonumber\\
H_i^{\rm era}(u_i;D_i)
&=\mathbf1[D_i\in\mathcal E_i]W_i^{\rm loc}(u_i;D_i),
\label{eq:Hera_retained}
\end{align}
when a rejected syndrome is retained. If only the erasure flag is retained, the correct factor is instead
\begin{equation}
H_i^{\rm era}(u_i)
=\int_{\mathcal E_i}W_i^{\rm loc}(u_i;D)\,dD.
\label{eq:Hera_integrated}
\end{equation}
Setting an erased-site factor to one is valid only if this sector-dependent rejection probability has already been absorbed into a branch prior.

Let $H_G$ have $r=n-1$ rows and define
\begin{equation*}
u_b(\sigma,e,s)
=r_G(\sigma)\oplus\ell_e\oplus sH_G,
\qquad s\in\F_2^{r}.
\end{equation*}
For accepted block $i$, let $\Theta_{i,b}^{\rm out}$ denote the wrapped likelihood of the outer modular measurement actually performed, and abbreviate $u_{b,i}=[u_b(\sigma,e,s)]_i$. The exact factorized binary-lifted score is
\begin{align}
\widetilde p_b(D_b\mid\sigma,e)
={}&\mathbf1[C_b\sigma=z_b]\nonumber\\
&\times\sum_{s\in\F_2^{r}}
\left[\prod_{i:a_{b,i}=1}
H_i^{\rm acc}(u_{b,i};D_i^{\rm loc})\right]
\nonumber\\
&\times\left[\prod_{i:a_{b,i}=0}
H_i^{\rm era}(u_{b,i};D_i^{\rm loc})\right]
\nonumber\\
&\times\left[\prod_{i:a_{b,i}=1}
\Theta_{i,b}^{\rm out}(u_{b,i};D_i^{\rm out})\right].
\label{eq:loss_amplification_factorized_score}
\end{align}
There is no outer factor on a locally erased block because no destructive outer measurement is performed there. Its local rejection factor remains as evidence.

For correlated multimode noise, the products in Eq.~\eqref{eq:loss_amplification_factorized_score} are replaced by a single joint wrapped likelihood. The static physical sector is
\begin{equation*}
\mathcal C_{\sigma,e}=\eta_G(\sigma)+\Phi(\ell_e)+L_G.
\end{equation*}
For compact notation set $P_b:=P_b^{\rm tot}$ and
\begin{equation*}
K_b^{\rm lat}=L_G\cap\ker P_b.
\end{equation*}
A calibrated joint readout kernel $g_b(D_b,b\mid\mu)$ includes the loss--amplification covariance, local acceptance/rejection regions, the branch policy, cell labels, and outer residuals. The quotient wrapped channel is
\begin{align}
\widetilde p_b(D_b,b\mid\sigma,e)
&=\sum_{[\lambda]\in\mathcal C_{\sigma,e}/K_b^{\rm lat}}
 g_b(D_b,b\mid P_b\lambda),
\label{eq:unnormalized_joint_channel}\\
Z_{b,\sigma,e}
&=\sum_b\int \widetilde p_b(D,b\mid\sigma,e)\,dD,
\label{eq:sector_normalizer}\\
p_b(D_b,b\mid\sigma,e)
&=\widetilde p_b(D_b,b\mid\sigma,e)/Z_{b,\sigma,e}.
\label{eq:joint_channel}
\end{align}
A translation-covariant readout has a common sector normalizer; the explicit expression is retained for a general calibrated instrument.

For a Gaussian residual conditioned on the discrete local and outer records,
\begin{align*}
\widetilde p_b(D_b,b\mid\sigma,e)
={}&\sum_{[\lambda]\in\mathcal C_{\sigma,e}/K_b^{\rm lat}}
 p_b^{\rm disc}(D_b^{\rm disc},b\mid P_b\lambda)\\
&\times\varphi_{\Sigma_{b,\obs}}(r_b-P_b\lambda),\\
Q_b(\lambda)
={}&(r_b-P_b\lambda)^T\Sigma_{b,\obs}^{-1}(r_b-P_b\lambda),\end{align*}
with one term per visible quotient class. The local erasure event is contained in $p_b^{\rm disc}$ and is not multiplied a second time.

The closest-coset decoder converts each term to an additive energy,
\begin{align}
\mathcal E_b(\lambda;D_b,b)
&=Q_b(\lambda)-2\log p_b^{\rm disc}
(D_b^{\rm disc},b\mid P_b\lambda),\nonumber\\
D_b^{\rm CC}(\sigma,e)
&=\min_{[\lambda]\in\mathcal C_{\sigma,e}/K_b^{\rm lat}}
\mathcal E_b(\lambda;D_b,b),\nonumber\\
\widehat{\widetilde p}_{b}^{\rm CC}(D_b,b\mid\sigma,e)
&=\exp[-D_b^{\rm CC}(\sigma,e)/2].
\label{eq:closest_coset_branch_score}
\end{align}

The approximation discarded by the closest-coset rule can be isolated exactly: the corresponding wrapped score before normalization factorizes as
\begin{align}
\widetilde p_b(D_b,b\mid\sigma,e)
&=e^{-D_b^{\rm CC}(\sigma,e)/2}
R_b(\sigma,e;D_b),\nonumber\\
R_b(\sigma,e;D_b)
&=\sum_{[\lambda]\in\mathcal C_{\sigma,e}/K_b^{\rm lat}}
e^{-\delta\mathcal E_b(\lambda;\sigma,e)/2}\geq1,\nonumber\\
\delta\mathcal E_b(\lambda;\sigma,e)
&:=\mathcal E_b(\lambda;D_b,b)-D_b^{\rm CC}(\sigma,e)\geq0.
\label{eq:maxlog_correction_factor}
\end{align}
Thus the closest-coset and exact wrapped posteriors coincide whenever $R_b$ is independent of $(\sigma,e)$, and they differ only through the sector dependence of this omitted multiplicity factor. All posterior odds and confidence values reported below are therefore probabilities within the declared closest-coset model; they are not asserted to equal exact maximum-likelihood posterior probabilities.

We use the convention $-\log0=+\infty$. When sector normalizers are common, they cancel from all posterior ratios. For a non-translation-covariant calibrated instrument, a sector-dependent CC normalizer $\widehat Z_{b,\sigma,e}^{\rm CC}$ is retained and $\widehat p_b^{\rm CC}=\widehat{\widetilde p}_b^{\rm CC}/\widehat Z_{b,\sigma,e}^{\rm CC}$. Equation~\eqref{eq:closest_coset_branch_score} is the sole sector-score approximation used by the algorithms and simulations.

The quotient sum in Eq.~\eqref{eq:joint_channel} is ordinary marginalization over periodic directions invisible to $P_b$.  Equivalently, one may assign variance $\tau I$ to such a coordinate, normalize its periodized Gaussian, and take $\tau\to\infty$; the standard Poisson-summation limit for lattice theta functions leaves one term per class of $\mathcal C_{\sigma,e}/K_b^{\rm lat}$ \cite{conrad2022}.  This concerns a genuinely unrecorded coordinate, such as an omitted future outer measurement or an unavailable Bell quadrature. It is not the physical model for attenuation. Under the adopted loss--amplification channel, attenuation changes the finite covariance in Eq.~\eqref{eq:loss_amplification_total_covariance}; a locally rejected output suppresses a future outer row, while the earlier local recovery record remains in the likelihood.

\subsection{Closest-coset syndrome and Pauli-frame decoding}
\label{subsec:outcome_likelihood}

Let $D_b=(D_{1:n}^{\rm loc},\mathbf a_b,y_b^{\rm out},r_b^{\rm out})$ denote the complete record retained by the selected branch. The sector engine supplies $\widehat p_b^{\rm CC}(D_b,b\mid\sigma,e)$ for every outer syndrome $\sigma$ and residual logical class $e$, including the probability of the local acceptance and rejection events that generated $b$. The primary decoder object is the normalized closest-coset distribution
\begin{equation}
\Pi_b(\sigma,e\mid D_b)
=
\frac{\pi_b(\sigma,e)\widehat p_b^{\rm CC}(D_b,b\mid\sigma,e)}
{\displaystyle\sum_{\sigma',e'}
\pi_b(\sigma',e')\widehat p_b^{\rm CC}(D_b,b\mid\sigma',e')}
\label{eq:sector_posterior}
\end{equation}
for a declared prior $\pi_b$. Throughout the remainder, $\Pi_b$ denotes this normalized CC distribution; replacing $\widehat p_b^{\rm CC}$ by the exact channel $p_b$ recovers the exact maximum-likelihood posterior. For an exact signed record, the distribution has support only on $\sigma\in\Sigma_b$; a fully resolved syndrome collapses the syndrome marginal. The maximum-score sector under the closest-coset likelihood is
\begin{equation*}
(\widehat\sigma_b,\widehat e_b)
:=\argmax_{\sigma,e}\Pi_b(\sigma,e\mid D_b),
\end{equation*}
with a fixed tie-breaking convention. This pair is useful for recovery diagnostics and for passing syndrome information to a higher decoder, but an MBQC controller ultimately needs the Pauli action induced on the surviving logical wire.

For a completed branch with one logical output, the stabilizer/Clifford branch compiler determines a binary Pauli-transfer map
\begin{equation*}
\tau_b:\F_2^{2n}/Q_G\longrightarrow\F_2^2,
\end{equation*}
which is constant on every physical sector $r_G(\sigma)+\ell_e+Q_G$. Operationally, $\tau_b$ is obtained by propagating a basis of physical Pauli labels through the branch tableau and reading the induced logical Pauli on the output block. Define the error-induced outgoing frame
\begin{equation*}
f_b(\sigma,e)
:=\tau_b\!\left(r_G(\sigma)\oplus\ell_e+Q_G\right)
\in\F_2^2.
\end{equation*}
For a terminal Pauli measurement with no surviving output wire, the corresponding transfer map is scalar and equals the flip action $[q^0_{a,b},r_G(\sigma)\oplus\ell_e]_{\DV}$.

The controller-facing frame posterior is the pushforward
\begin{equation}
\Pi_b^{\rm fr}(f\mid D_b)
=
\sum_{\substack{\sigma,e:\\ f_b(\sigma,e)=f}}
\Pi_b(\sigma,e\mid D_b),
\qquad f\in\F_2^2.
\label{eq:frame_posterior}
\end{equation}
Its maximum-a-posteriori estimate and log-odds confidence are
\begin{align*}
\widehat f_b
&:=\argmax_f\Pi_b^{\rm fr}(f\mid D_b),\\
\Gamma_b^{\rm fr}
&:=\log\frac{\Pi_b^{\rm fr}(\widehat f_b\mid D_b)}
{1-\Pi_b^{\rm fr}(\widehat f_b\mid D_b)}.\end{align*}
Within the normalized closest-coset model,
\begin{equation}
\Pr_{\rm CC}[f_b\neq\widehat f_b\mid D_b]
=\frac{1}{1+e^{\Gamma_b^{\rm fr}}}.
\label{eq:frame_error_probability}
\end{equation}
The corresponding frequentist accepted-frame error is measured in Monte Carlo and used to calibrate the confidence threshold; it is not assumed to equal the raw max-log probability without validation.
The joint-sector confidence $\Gamma_b^{\sigma e}$ is defined analogously from $\Pi_b(\widehat\sigma_b,\widehat e_b\mid D_b)$. The two confidences need not agree because several physical sectors may induce the same logical frame. The preferred decoder output is the full posterior $\Pi_b(\sigma,e\mid D_b)$ together with its frame pushforward; a compressed decision record is
\begin{equation*}
\mathcal O_b^{\rm hard}
=(\widehat\sigma_b,\widehat e_b,\widehat f_b,
\Gamma_b^{\rm fr}).
\end{equation*}
This contains the requested $(\widehat e_b,\widehat\sigma_b,\Gamma_b)$ interface while making explicit the controller-facing frame $\widehat f_b$. For a terminal Pauli measurement, define the pure-error-corrected branch offset
\begin{equation*}
\chi_b^{r_G}(a,\sigma;q^0)
:=\chi_b^{\rm raw}(a;q^0)
\oplus[q^0_{a,b},r_G(\sigma)]_{\DV}.
\end{equation*}

Define the terminal Pauli-outcome map
\begin{align}
\mu_{a,b}(\sigma,e)
&:=\chi_b^{\rm raw}(a;q^0)
 \oplus[q^0_{a,b},r_G(\sigma)\oplus\ell_e]_{\DV}
\nonumber\\
&=\chi_b^{r_G}(a,\sigma;q^0)\oplus[a,e]_1.
\label{eq:outcome_likelihood_residual}
\end{align}

The terminal logical-measurement distribution is the deterministic pushforward of the normalized sector distribution,
\begin{equation}
\Pr_b(m\mid D_b,a)
=\sum_{\sigma,e}\Pi_b(\sigma,e\mid D_b)
\mathbf1[m=\mu_{a,b}(\sigma,e)].
\label{eq:outcome_likelihood}
\end{equation}
For a branch with a surviving logical output, Eq.~\eqref{eq:frame_posterior} is the corresponding Pauli-frame pushforward.  These are standard applications of Bayes' rule and the law of total probability.  The distributions are independent of the operational representative and of the chosen measurement-supported reference: changing either multiplies the raw record and its stabilizer correction by the same signed character.  Changing the pure-error section only permutes the labels $(\sigma,e)$ of fixed physical sectors, so the induced outcome and frame distributions are unchanged.

For an equatorial measurement, the terminal bit posterior is instead
\begin{align}
\Pr_b(m\mid D_b,\theta,\varrho)
={}&\sum_{\sigma,e}\Pi_b(\sigma,e\mid D_b)
\Tr[\varrho E^{(\sigma,e,b)}_{m,\theta}],
\label{eq:nonpauli_povm}
\end{align}
because an uncertain $X$ frame changes the angle rather than merely flipping a bit. The frame posterior must therefore be retained before selecting the next physical basis.

Let $\mathfrak f_t\in\F_2^2$ be the logical Pauli frame tracked by the controller before branch $b$. Let $T_b^{\rm log}\in\operatorname{Sp}_2(\F_2)$ be the Clifford action of the ideal branch and let $d_b(y_b)\in\F_2^2$ be the deterministic measurement byproduct. For true sector $(\sigma,e)$, the true outgoing frame is
\begin{equation}
\mathfrak f_{t+1}^{\rm true}
=T_b^{\rm log}\mathfrak f_t^{\rm true}
\oplus d_b(y_b)\oplus f_b(\sigma,e),
\label{eq:true_frame_update}
\end{equation}
whereas the hard controller uses
\begin{equation}
\widehat{\mathfrak f}_{t+1}
=T_b^{\rm log}\widehat{\mathfrak f}_t
\oplus d_b(y_b)\oplus\widehat f_b.
\label{eq:controller_frame_update}
\end{equation}
If the incoming frame was tracked correctly, the untracked residual is
\begin{equation}
\delta\mathfrak f_{t+1}
=f_b(\sigma,e)\oplus\widehat f_b.
\label{eq:residual_frame_fault}
\end{equation}
Thus a correctly identified frame is removed completely by classical feedforward, while an incorrect frame estimate leaves a nontrivial logical Pauli fault. For the next equatorial measurement,
\begin{equation}
\theta_{t+1}^{\rm phys}
=(-1)^{(\widehat{\mathfrak f}_{t+1})_X}\theta_{t+1},
\qquad
m_{t+1}^{\rm corr}
=m_{t+1}^{\rm raw}\oplus(\widehat{\mathfrak f}_{t+1})_Z.
\label{eq:future_basis_update}
\end{equation}

Confidence can be propagated rather than thresholded immediately. If the incoming controller frame has distribution $\Pi_t^{\rm ctrl}$ and is conditionally independent of the current branch error, then
\begin{align}
\Pi_{t+1}^{\rm ctrl}(f')
={}&\sum_{f,g}
\Pi_t^{\rm ctrl}(f)\Pi_b^{\rm fr}(g\mid D_b)\nonumber\\
&\times\mathbf1\!\left[
 f'=T_b^{\rm log}f\oplus d_b(y_b)\oplus g
\right].
\label{eq:posterior_frame_propagation}
\end{align}
Correlated branches require the corresponding joint posterior instead of the product in Eq.~\eqref{eq:posterior_frame_propagation}. A confidence erasure is declared only when the controller chooses to compress the posterior distribution and $\Gamma_b^{\rm fr}<\Gamma_{\rm th}^{\rm fr}$; a terminal measurement may analogously use the binary outcome log odds.

The graph-module failure probability is
\begin{equation}
P_{\rm fail}
=P_{\accera}+P_{\confera}+P_{\rm frame}^{\rm err}
\label{eq:loss_amplification_failure_metric}
\end{equation}
where
\begin{align*}
P_{\accera}&=\Pr[\mathcal L_{a,b}^{\meas}=\varnothing],\\
P_{\confera}&=\Pr[\mathcal L_{a,b}^{\meas}\neq\varnothing,
\ \Gamma_b^{\rm fr}<\Gamma_{\rm th}^{\rm fr}],\\
P_{\rm frame}^{\rm err}
&=\Pr[\widehat f_b\neq f_b(\sigma,e),\ \text{module accepted}].\end{align*}
The local erasure probability $\varepsilon_{\GKP}$ and accepted local Pauli error $p_{\rm err}^{\rm acc}$ are reported separately. Adding $\varepsilon_{\GKP}$ directly to Eq.~\eqref{eq:loss_amplification_failure_metric} would double count locally erased blocks that the graph code routes around successfully.

Algorithm~\ref{alg:syndrome_frame_decoder} combines selected-branch closest-coset scoring, syndrome/frame inference, and the MBQC controller update.

\begin{prxalgorithm}{Syndrome-resolved closest-coset and Pauli-frame decoder}
\label{alg:syndrome_frame_decoder}
\footnotesize
\begin{algorithmic}[1]
\Require Complete selected-branch record $(D_b,b)$; priors $\pi_b(\sigma,e)$; compiled lattice and invisible sublattice; pure-error section; branch transfer $\tau_b$; incoming controller-frame posterior; confidence threshold.
\Ensure Syndrome and frame posteriors, accepted frame update and adapted measurement instruction, or $\accera/\confera$.
\If{the requested branch operation is inaccessible} \State \Return $\accera$. \EndIf
\For{each retained syndrome $\sigma$ and logical class $e$}
  \State Form $\mathcal C_{\sigma,e}/K_b^{\rm lat}$ and include the branch event and all accepted, rejected, and readout likelihood factors.
  \State Solve the covariance-weighted closest-vector problem and set $w_{\sigma,e}=\pi_b(\sigma,e)e^{-D_b^{\rm CC}(\sigma,e)/2}$, including any calibrated sector normalizer.
\EndFor
\State Normalize $w_{\sigma,e}$ to obtain $\Pi_b(\sigma,e\mid D_b)$ and its MAP sector.
\State Push $\Pi_b$ through $\tau_b$ to obtain $\Pi_b^{\rm fr}$, the MAP frame $\widehat f_b$, and $\Gamma_b^{\rm fr}$; form any terminal Pauli or equatorial outcome posterior.
\State Compose the incoming controller-frame posterior with $\Pi_b^{\rm fr}$ using Eq.~\eqref{eq:posterior_frame_propagation}, or use the supplied joint posterior for correlated branches.
\If{the relevant frame or outcome confidence is below threshold}
  \State \Return $\confera$ while retaining the complete posterior distributions for higher-level decoding.
\EndIf
\State Apply the accepted frame update, adapt subsequent equatorial bases and reported bits using Eq.~\eqref{eq:future_basis_update}, and return the updated controller record.
\end{algorithmic}
\end{prxalgorithm}

\paragraph{Closest-coset implementation.}\label{subsec:closest_coset_decoder}

For Gaussian observations, whitening by $\Sigma_{b,\obs}^{-1/2}$ converts the continuous part of Eq.~\eqref{eq:closest_coset_branch_score} to a Euclidean closest-vector problem; discrete acceptance, rejection, cell, and branch factors remain additive negative-log costs.  The solver searches $\mathcal C_{\sigma,e}/K_b^{\rm lat}$, returns a global minimizer under a deterministic tie rule, and assigns score $e^{-D_b^{\rm CC}/2}$.  All numerical panels, recursive posterior distributions, and fusion results use this same max-log rule.  Because nonminimal representatives are discarded, raw odds are calibrated against simulated accepted-frame errors before thresholding.

\subsection{End-to-end specialization to the cube and decorated pentagon}
\label{subsec:worked_modules}

The preceding constructions are now evaluated for the two modules used in the simulations.  All binary sums are over $\F_2$, Pauli products use the Hermitian convention of Appendix~\ref{app:signed_compilation}, and every executable local and graph-level factor uses the closest-coset scores in Eqs.~\eqref{eq:local_closest_coset_weight} and \eqref{eq:closest_coset_branch_score}.  The exact wrapped sums remain reference models only.

\paragraph{Cube progenitor: explicit code and GKP lattice.}

Order the seven code vertices as $(1,2,3,4,5,6,7)$.  With input vertex $0$, the neighbors of the input are $\{1,3,7\}$.  The identity-input fiber is generated by
\begin{align*}
h_1&=\pi_c(K_2)=Z_1X_2Z_3Z_5,\\
h_2&=\pi_c(K_4)=Z_1X_4Z_5Z_7,\\
h_3&=\pi_c(K_5)=Z_2Z_4X_5Z_6,\\
h_4&=\pi_c(K_6)=Z_3Z_5X_6Z_7,\\
h_5&=\pi_c(K_1K_3)=X_1X_3Z_4Z_6,\\
h_6&=\pi_c(K_1K_7)=X_1X_7Z_2Z_6.\end{align*}
The input-$X$ generator $K_0$ and the input-$Z$ generator $K_1$ give
\begin{equation*}
\ell_X=Z_1Z_3Z_7,
\qquad
\ell_Z=X_1Z_2Z_4,
\qquad
[\ell_X,\ell_Z]_{\DV}=1.
\end{equation*}
In the binary order $(x_1,\ldots,x_7\mid z_1,\ldots,z_7)$, the resulting outer-code generator is
\begin{widetext}
\begin{equation}
H_{\rm cube}=
\left(
\begin{array}{ccccccc|ccccccc}
0&1&0&0&0&0&0&1&0&1&0&1&0&0\\
0&0&0&1&0&0&0&1&0&0&0&1&0&1\\
0&0&0&0&1&0&0&0&1&0&1&0&1&0\\
0&0&0&0&0&1&0&0&0&1&0&1&0&1\\
1&0&1&0&0&0&0&0&0&0&1&0&1&0\\
1&0&0&0&0&0&1&0&1&0&0&0&1&0
\end{array}
\right).
\label{eq:H_cube_explicit}
\end{equation}
\end{widetext}
Thus the square-block oscillator code is the explicit binary-lifted lattice.  Here $\widetilde M_{\rm cube}^{\square}$ is an overcomplete generating array, whereas $M_{\rm cube}$ denotes any square full-rank lattice basis obtained from it:
\begin{align*}
\widetilde M_{\rm cube}^{\square}
&=\frac1{\sqrt2}
\begin{pmatrix}2I_{14}\\ H_{\rm cube}\end{pmatrix},\\
\operatorname{covol}(L_{\rm cube})
&=|\det M_{\rm cube}|=\frac{2^7}{2^6}=2.\end{align*}
The hexagonal version is obtained by the blockwise symplectic switch $S_{\rm hex}^{\oplus7}$; the binary matrix and the branch algebra do not change.

An exact graph-level check is obtained by enumerating the $64$ representatives of each affine class $\ell_a+Q_G$.  Let $N_r^{(a)}$ be the number of $r$-block local-decoder-erasure sets for which at least one representative of logical $a$ has identity on every unavailable block.  For each $a\in\{X,Y,Z\}$ the enumeration gives
\begin{equation*}
(N_0,N_1,N_2,N_3,N_4,N_5,N_6,N_7)
=(1,7,21,28,7,0,0,0).
\end{equation*}
Hence the heralded branch-accessibility polynomial is
\begin{align}
P_{\rm cube}^{\rm acc}(\varepsilon)
&=\sum_{r=0}^7N_r\varepsilon^r(1-\varepsilon)^{7-r}\nonumber\\
&=1-7\varepsilon^3+21\varepsilon^5-21\varepsilon^6+6\varepsilon^7.
\label{eq:cube_accessibility_polynomial}
\end{align}
It follows directly that every two-block local decoder erasure is routable at the representative-existence level and that $P_{\rm cube}^{\rm acc}(1/2)=1/2$.  Equation~\eqref{eq:cube_accessibility_polynomial} is the algebraic transfer function of the local GKP erasure rate under independent identical block decisions. The online policy still imposes the causal measurement order of Algorithm~\ref{alg:logical_branch}.

\paragraph{One cube branch conditioned on a local GKP erasure.}
Consider a terminal logical-$X$ measurement with
\begin{equation*}
\mathbf q_{\rm c}=(I,X,Z,X,Z,X,Z),\qquad\mathbf a_{\rm c}=(0,1,1,1,1,1,1).
\end{equation*}
Although the reference $\ell_X$ acts on block $1$, whose refreshed output was discarded by the local GKP decoder, multiplication by $h_1$ gives the accessible representative
\begin{equation*}
q_{X,\rm c}=\ell_X\oplus h_1=X_2Z_5Z_7.
\end{equation*}
The measured-check subspace has rank two and may be generated by
\begin{align*}
q_{\rm c,1}&=h_4=Z_3Z_5X_6Z_7,\\
q_{\rm c,2}&=h_1\oplus h_2=X_2Z_3X_4Z_7.\end{align*}
Relative to the rows of Eq.~\eqref{eq:H_cube_explicit}, the syndrome-restriction matrix is therefore
\begin{equation*}
C_{\rm c}=
\begin{pmatrix}
0&0&0&1&0&0\\
1&1&0&0&0&0
\end{pmatrix}.
\end{equation*}
Let $\mu_i^P\in\F_2$ denote the eigenvalue bit of the modular measurement of local Pauli $P$ on block $i$.  For zero incoming frame and the displayed positive stabilizer representatives, the signed record is
\begin{align*}
\chi_{\rm c}^{\rm raw}(X)
&=\mu_2^X\oplus\mu_5^Z\oplus\mu_7^Z,\\
z_{{\rm c},1}
&=\mu_3^Z\oplus\mu_5^Z\oplus\mu_6^X\oplus\mu_7^Z,\\
z_{{\rm c},2}
&=\mu_2^X\oplus\mu_3^Z\oplus\mu_4^X\oplus\mu_7^Z.\end{align*}
Known lattice-phase and incoming-frame offsets are added exactly as in Appendix~\ref{app:signed_compilation}.  The exact discrete constraints are
\begin{equation*}
\sigma_4=z_{{\rm c},1},
\qquad
\sigma_1\oplus\sigma_2=z_{{\rm c},2}.
\end{equation*}
Thus $\operatorname{rank}C_{\rm c}=2$ and $16$ of the $64$ complete syndromes remain possible before the analog record is used.

A convenient pure-error section is
\begin{equation*}
r_{\rm cube}(\sigma)
=Z_2^{\sigma_1}Z_4^{\sigma_2}Z_5^{\sigma_3}
 Z_6^{\sigma_4}Z_3^{\sigma_5}Z_7^{\sigma_6},
\end{equation*}
for which $\syn_G(r_{\rm cube}(\sigma))=\sigma$.  Direct symplectic products give
\begin{equation*}
[q_{X,\rm c},r_{\rm cube}(\sigma)]_{\DV}=\sigma_1,
\qquad
[q_{X,\rm c},\ell_e]_{\DV}=e_Z.
\end{equation*}
The outcome pushforward in Eq.~\eqref{eq:outcome_likelihood} therefore becomes the closed formula
\begin{equation*}
 m_X
=\mu_2^X\oplus\mu_5^Z\oplus\mu_7^Z
  \oplus\sigma_1\oplus e_Z
\end{equation*}
For independent local readouts, set
\begin{equation*}
u_{\rm c}(\sigma,e,a)
:=r_{\rm cube}(\sigma)\oplus\ell_e\oplus aH_{\rm cube}.
\end{equation*}
Writing $u_{{\rm c},i}=[u_{\rm c}(\sigma,e,a)]_i$, define the $64$ candidate negative-log costs
\begin{align*}
\mathcal E_{\rm c}(\sigma,e,a)
=-2\log\Bigg[&H_1^{\rm era}(u_{{\rm c},1};D_1^{\rm loc})
\prod_{i=2}^{7}H_i^{\rm acc}(u_{{\rm c},i};D_i^{\rm loc})\\
&\times\prod_{i=2}^{7}
\Theta_{i,\rm c}^{\rm out}(u_{{\rm c},i};D_i^{\rm out})\Bigg],
\end{align*}
with $-\log0=+\infty$. The explicit closest-coset score is
\begin{align}
D_{\rm c}^{\rm CC}(\sigma,e)
&=\min_{a\in\F_2^6}\mathcal E_{\rm c}(\sigma,e,a),\nonumber\\
\widehat p_{\rm c}^{\rm CC}(D_{\rm c}\mid\sigma,e)
&=\mathbf1[\sigma_4=z_{{\rm c},1}]\nonumber\\
&\quad\times\mathbf1[\sigma_1\oplus\sigma_2=z_{{\rm c},2}]
 e^{-D_{\rm c}^{\rm CC}(\sigma,e)/2}.
\label{eq:cube_explicit_sector_score}
\end{align}
Normalizing $\pi_{\rm c}(\sigma,e)\widehat p_{\rm c}^{\rm CC}$ gives $\Pi_{\rm c}$, and the final binary distribution is
\begin{equation}
\Pr_{c}(m\mid D_{\rm c})
=\sum_{\sigma,e}
\Pi_{\rm c}(\sigma,e\mid D_{\rm c})
\mathbf1\!\left[
m=\chi_{\rm c}^{\rm raw}\oplus\sigma_1\oplus e_Z
\right].
\label{eq:cube_explicit_posterior_pushforward}
\end{equation}
Equations~\eqref{eq:cube_explicit_sector_score} and \eqref{eq:cube_explicit_posterior_pushforward} provide the fully expanded closest-coset-and-pushforward example used to instantiate Algorithm~\ref{alg:syndrome_frame_decoder}.

\paragraph{Decorated pentagon: teleportation and frame maps.}

Order the five code vertices as $(1,2,3,4,5)$.  The identity-input fiber is generated by
\begin{align*}
g_1&=\pi_c(K_2)=Z_1X_2Z_3,\\
g_2&=\pi_c(K_3)=Z_2X_3Z_4,\\
g_3&=\pi_c(K_5)=Z_1X_5,\\
g_4&=\pi_c(K_1K_4)=X_1X_4Z_2Z_3Z_5,\end{align*}
and the logical classes may be represented by
\begin{equation*}
\ell_X=Z_1Z_4,
\qquad
\ell_Z=X_1Z_2Z_5.
\end{equation*}
Thus
\begin{equation}
H_{\rm dec}=
\left(
\begin{array}{ccccc|ccccc}
0&1&0&0&0&1&0&1&0&0\\
0&0&1&0&0&0&1&0&1&0\\
0&0&0&0&1&1&0&0&0&0\\
1&0&0&1&0&0&1&1&0&1
\end{array}
\right),
\label{eq:H_decorated_explicit}
\end{equation}
and the square-block binary-lifted generator is the following overcomplete array; $M_{\rm dec}$ denotes any square full-rank basis of the same lattice:
\begin{align*}
\widetilde M_{\rm dec}^{\square}
&=\frac1{\sqrt2}
\begin{pmatrix}2I_{10}\\H_{\rm dec}\end{pmatrix},\\
\operatorname{covol}(L_{\rm dec})
&=|\det M_{\rm dec}|=\frac{2^5}{2^4}=2.\end{align*}

For a candidate output $o$, the stabilizer-pathfinding test is the explicit search for $(q_X,q_Z)\in\mathcal P_{b,o}$ of Eq.~\eqref{eq:SPC_pair_set}.  Enumerating the $16$ representatives in each logical class gives $148$ labeled pairs before loss pruning.  If $N_r^{\rm SPC}$ is the number of $r$-block local-erasure sets admitting at least one such pair, then
\begin{align*}
&(N_0^{\rm SPC},N_1^{\rm SPC},N_2^{\rm SPC},
  N_3^{\rm SPC},N_4^{\rm SPC},N_5^{\rm SPC})\\
&\hspace{5em}=(1,5,5,0,0,0).\end{align*}
The five routable two-local-erasure sets are
$\{1,2\}$, $\{1,3\}$, $\{1,5\}$, $\{2,5\}$, and $\{3,5\}$, and the heralded SPC-accessibility polynomial is
\begin{align}
P_{\rm dec}^{\rm SPC}(\varepsilon)
&=(1-\varepsilon)^5+5\varepsilon(1-\varepsilon)^4
 +5\varepsilon^2(1-\varepsilon)^3\nonumber\\
&=1-5\varepsilon^2+5\varepsilon^3-\varepsilon^5.
\label{eq:decorated_accessibility_polynomial}
\end{align}
In particular every single local decoder erasure has an explicit teleportation certificate.  Table~\ref{tab:decorated_single_loss} gives identity-oriented pairs, meaning that the local factors of $(q_X,q_Z)$ on the output are $(X_o,Z_o)$.  For such a pair, if the non-output measurement parities are $d_X$ and $d_Z$, the ideal output state is $X_o^{d_X}Z_o^{d_Z}|\psi\rangle$.

\begin{table*}[t]
\caption{Explicit one-local-erasure stabilizer-pathfinding certificates for the decorated-pentagon-plus-leaf code.  The bits $\mu_i^P$ are local measurement eigenvalue bits.  Each row avoids the indicated locally erased block, and $q_X$ and $q_Z$ anticommute only on the listed output.}
\label{tab:decorated_single_loss}
\small
\begin{tabular*}{\textwidth}{@{\extracolsep{\fill}}cllll@{}}
\toprule
Locally erased block & Output & $q_X$ & $q_Z$ & Deterministic frame $(d_X,d_Z)$\\
\midrule
$1$ & $3$ & $Z_2X_3X_5$ & $Z_3X_4$ & $(\mu_4^X,\,\mu_2^Z\oplus\mu_5^X)$\\
$2$ & $5$ & $Z_4X_5$ & $X_1X_3Z_4Z_5$ & $(\mu_1^X\oplus\mu_3^X\oplus\mu_4^Z,\,\mu_4^Z)$\\
$3$ & $5$ & $Z_4X_5$ & $X_1Z_2Z_5$ & $(\mu_1^X\oplus\mu_2^Z,\,\mu_4^Z)$\\
$4$ & $5$ & $Z_2X_3X_5$ & $X_1Z_2Z_5$ & $(\mu_1^X\oplus\mu_2^Z,\,\mu_2^Z\oplus\mu_3^X)$\\
$5$ & $3$ & $Z_1Z_2X_3$ & $Z_3X_4$ & $(\mu_4^X,\,\mu_1^Z\oplus\mu_2^Z)$\\
\bottomrule
\end{tabular*}
\end{table*}

The Pauli-frame transfer is also explicit.  For any physical phase-free Pauli label $u$, an identity-oriented SPC pair induces
\begin{equation}
\tau_{b,o}(u)
=\left([q_Z,u]_{\DV},[q_X,u]_{\DV}\right)
=:(f_X,f_Z).
\label{eq:SPC_frame_transfer_explicit}
\end{equation}
Indeed, a sign flip of the logical-$Z$ relation is an $X_o$ byproduct and a sign flip of the logical-$X$ relation is a $Z_o$ byproduct.

Choose the pure-error section
\begin{equation*}
r_{\rm dec}(\sigma)
=Z_2^{\sigma_1}X_4^{\sigma_2}
 Z_5^{\sigma_3}X_5^{\sigma_4},
\qquad
\syn_G(r_{\rm dec}(\sigma))=\sigma.
\end{equation*}
For the last row of Table~\ref{tab:decorated_single_loss}, with output $3$ and block $5$ locally erased,
\begin{equation*}
q_X^{(3)}=Z_1Z_2X_3,
\qquad
q_Z^{(3)}=Z_3X_4,
\end{equation*}
and direct pairing gives
\begin{equation*}
f_{b_3}(\sigma,e)=(e_X,e_Z).
\end{equation*}
The frame posterior therefore simplifies to
\begin{equation}
\Pi_{b_3}^{\rm fr}(f\mid D)
=\sum_{\sigma}\Pi_{b_3}(\sigma,f\mid D).
\label{eq:decorated_frame_posterior_output3}
\end{equation}
By contrast, for the third row of the table, with block $3$ locally erased and output $5$,
\begin{equation*}
q_X^{(5)}=Z_4X_5,
\qquad
q_Z^{(5)}=X_1Z_2Z_5,
\end{equation*}
and Eq.~\eqref{eq:SPC_frame_transfer_explicit} gives the nontrivial syndrome-dependent map
\begin{equation*}
f_{b_5}(\sigma,e)
=\left(e_X\oplus\sigma_4,\,
        e_Z\oplus\sigma_2\oplus\sigma_3\right)
\end{equation*}
Accordingly,
\begin{equation}
\Pi_{b_5}^{\rm fr}(f\mid D)
=\sum_{\substack{\sigma,e:\\
f_X=e_X\oplus\sigma_4,\\
f_Z=e_Z\oplus\sigma_2\oplus\sigma_3}}
\Pi_{b_5}(\sigma,e\mid D).
\label{eq:decorated_frame_posterior_output5}
\end{equation}
This is a concrete case in which marginalizing over $\sigma$ before computing the output frame is incorrect.  For both minimal branches, $Q_b^{\meas}=\{0\}$: no complete outer check is supported solely on the Pauli-measured non-output blocks.  The $16$ syndrome values must therefore be weighted by their analog GKP likelihood rather than replaced by $\sigma=0$.

For the output-$5$ branch, the deterministic frame from Table~\ref{tab:decorated_single_loss} is
\begin{equation*}
d_{b_5}(y)=
(\mu_1^X\oplus\mu_2^Z,\,\mu_4^Z).
\end{equation*}
After the posterior estimate $\widehat f=(\widehat f_X,\widehat f_Z)$, the MBQC controller implements
\begin{equation}
\theta_{\rm phys}
=(-1)^{\mu_1^X\oplus\mu_2^Z\oplus\widehat f_X}\theta,
\qquad
m_{\rm corr}
=m_{\rm raw}\oplus\mu_4^Z\oplus\widehat f_Z
\label{eq:decorated_explicit_controller_update}
\end{equation}
A correct $\widehat f$ removes the physical sector completely by feedforward; an incorrect estimate leaves the residual logical Pauli $f\oplus\widehat f$; and low log-odds produce the confidence erasure of Algorithm~\ref{alg:syndrome_frame_decoder}.

For independent local readouts, substituting $H_{\rm dec}$ and $r_{\rm dec}$ into the factorized likelihood produces $2^4=16$ candidate negative-log costs. The decorated-pentagon closest-coset score is the exponential of minus one half of their minimum. Unlike the cube branch, both minimal branches have $Q_b^{\meas}=\{0\}$, so no measured-check indicator precedes this minimization. The normalized CC distribution is then pushed through Eq.~\eqref{eq:decorated_frame_posterior_output3} or Eq.~\eqref{eq:decorated_frame_posterior_output5}, and a terminal non-Pauli measurement additionally uses the error-conditioned POVM of Eq.~\eqref{eq:nonpauli_povm}.

\subsection{Worked-module transfer and finite-size benchmarks}
\label{subsec:loss_amplification_exact_results}

The local decoder produces the IID abstention probability
\begin{equation*}
\varepsilon
=\varepsilon_{\GKP}(\eta,\Sigma_{\rm prep},\Gamma_{\rm loc}).
\end{equation*}
The exact cube and decorated-pentagon accessibility probabilities are therefore obtained by evaluating Eqs.~\eqref{eq:cube_accessibility_polynomial} and \eqref{eq:decorated_accessibility_polynomial} at this calibrated $\varepsilon$.  This composition is the only conversion from optical attenuation to graph unavailability: the graph polynomials act on decoder flags, not directly on $1-\eta$.

For the decorated pentagon, the nontrivial accessibility crossover with a single accepted refreshed block is
\begin{equation*}
\varepsilon_{\rm dec}^{\star}=0.2738905549\ldots,
\qquad
P_{\rm dec}^{\rm SPC}(\varepsilon_{\rm dec}^{\star})
=1-\varepsilon_{\rm dec}^{\star}.
\end{equation*}
As for the cube crossover stated after Eq.~\eqref{eq:cube_accessibility_polynomial}, this is an erasure-only accessibility fixed point of the reliability polynomial.  Total module failure additionally contains accepted frame errors and $\confera$, and the decoder-generated erasure probability is itself a function of the analog channel and the confidence rule.

The local quantities $\varepsilon_{\GKP}$ and $p_{\rm err}^{\rm acc}$ follow from Eqs.~\eqref{eq:square_loss_amplification_exact_erasure} and \eqref{eq:square_loss_amplification_exact_error} for the ideal-square benchmark, or from the corresponding calibrated lattice instrument.  Raising $\Gamma_{\rm loc}$ transfers probability from accepted local errors to $\locera$, after which the graph routes only those abstention patterns admitted by its accessibility polynomial.

For independent but nonidentical local interfaces, the correct accessibility expression is the multivariate reliability polynomial
\begin{equation*}
P_G^{\rm acc}
=\sum_{\bm a\in\F_2^n}
\prod_{i=1}^{n}(1-\varepsilon_i)^{a_i}
\varepsilon_i^{1-a_i}
\,\mathbf1[\mathcal L^{\meas}_{a,\bm a}\neq\varnothing].
\end{equation*}
If the local records or erasure flags are correlated, the product distribution must be replaced by the exact mask distribution generated by the joint Gaussian channel and local policy,
\begin{equation*}
P_G^{\rm acc}
=\sum_{\bm a}
\Pr(\bm a\mid\bm\eta,\Sigma_{\epsilon},\Gamma_{\rm loc})
\,\mathbf1[\mathcal L^{\meas}_{a,\bm a}\neq\varnothing].
\end{equation*}
Thus the scalar substitution $\ell\mapsto\varepsilon_{\GKP}$ is valid only in the IID case.

\paragraph{Benchmark conventions and interpretation.}\label{subsec:fixed_family_results}

Every numerical result in this section uses the closest-coset decoder of Sec.~\ref{subsec:outcome_likelihood}; no alternative decoder is mixed into any panel. The simulations use the attenuation coordinate
\begin{equation*}
\ell=1-\eta,
\end{equation*}
and do not draw an independent Bernoulli deletion. The parameter $\ell$ changes the displacement covariance through the loss--amplification channel, while the availability mask is generated only after local GKP decoding and confidence filtering. To isolate the additional degradation caused by propagation relative to the nonzero finite-squeezing baseline, the logarithmic panels report
\begin{equation}
P_{\rm exc}^{a}(\ell,\sigma_{\GKP})
=P_{\rm fail}^{a}(\ell,\sigma_{\GKP})
-P_{\rm fail}^{a}(0,\sigma_{\GKP}),
\label{eq:excess_failure_definition}
\end{equation}
with a numerical plotting floor applied only on logarithmic axes. The total failure probability is always recovered as
\begin{equation*}
P_{\rm fail}^{a}(\ell,\sigma_{\GKP})
=P_{\rm fail}^{a}(0,\sigma_{\GKP})+P_{\rm exc}^{a}(\ell,\sigma_{\GKP}),
\end{equation*}
and remains the absolute operational quantity.

We use three distinct benchmark markers. First, the fixed points of Eqs.~\eqref{eq:cube_accessibility_polynomial} and \eqref{eq:decorated_accessibility_polynomial} characterize only the IID erasure-accessibility subproblem. Second, an attenuation-reference crossing $\ell_{\rm ref}$ is defined by
\begin{equation}
P_{\rm exc}^{a}(\ell_{\rm ref},\sigma_{\GKP})=\ell_{\rm ref}.
\label{eq:attenuation_reference_crossing}
\end{equation}
Because $\ell$ is a transmissivity deficit rather than a Pauli-error probability, Eq.~\eqref{eq:attenuation_reference_crossing} is a  reference crossing, not a fault-tolerance threshold. Third, the ratio-one contour defines an excess-noise break-even point
\begin{equation}
P_{\rm exc}^{a,\rm enc}(\ell_{\rm br},\sigma_{\GKP})
=P_{\rm exc}^{a,\rm 1GKP}(\ell_{\rm br},\sigma_{\GKP}).
\label{eq:excess_break_even}
\end{equation}
It compares only the propagation-induced increments; it does not by itself establish an advantage in total failure probability.

For experimental interpretation we quote both attenuation and GKP squeezing in decibels. The plotted $\sigma_{\GKP}$ is taken in the conventional quadrature normalization with vacuum variance $1/2$; in the displacement coordinates of Eq.~\eqref{eq:displacement}, its preparation-noise contribution is $\Sigma_{\rm prep}=\sigma_{\GKP}^{2}I/(2\pi)$. We use
\begin{align}
\mathcal L_{\rm ch}^{\rm dB}(\ell)
&=-10\log_{10}\eta=-10\log_{10}(1-\ell),\nonumber\\
s_{\GKP}^{\rm dB}
&=-10\log_{10}\!\left(2\sigma_{\GKP}^{2}\right).
\label{eq:db_conversions}
\end{align}
Thus $\sigma_{\GKP}=0.10,0.15,0.20$ correspond respectively to $17.0$, $13.5$, and $11.0\,\mathrm{dB}$ of GKP squeezing. 

Figure~\ref{fig:result1_square_excess} specializes the closest-coset decoder to the cube and decorated-pentagon progenitors. The cube implements logical Pauli measurements, whereas the decorated pentagon implements the protected transport and frame-management layer of the $A(\theta)$ branch with an ideal supplied terminal resource. The line panels compare the encoded excess failure with a matched single-block GKP reference; the heat maps resolve the joint dependence on $\ell$ and $\sigma_{\GKP}$. A ratio below one means a smaller propagation-induced excess, not necessarily a smaller total logical failure. The branch diagrams fix the exact labeled progenitors used by the simulations.

\begin{figure*}[!t]
\centering
\includegraphics[width=0.98\textwidth]{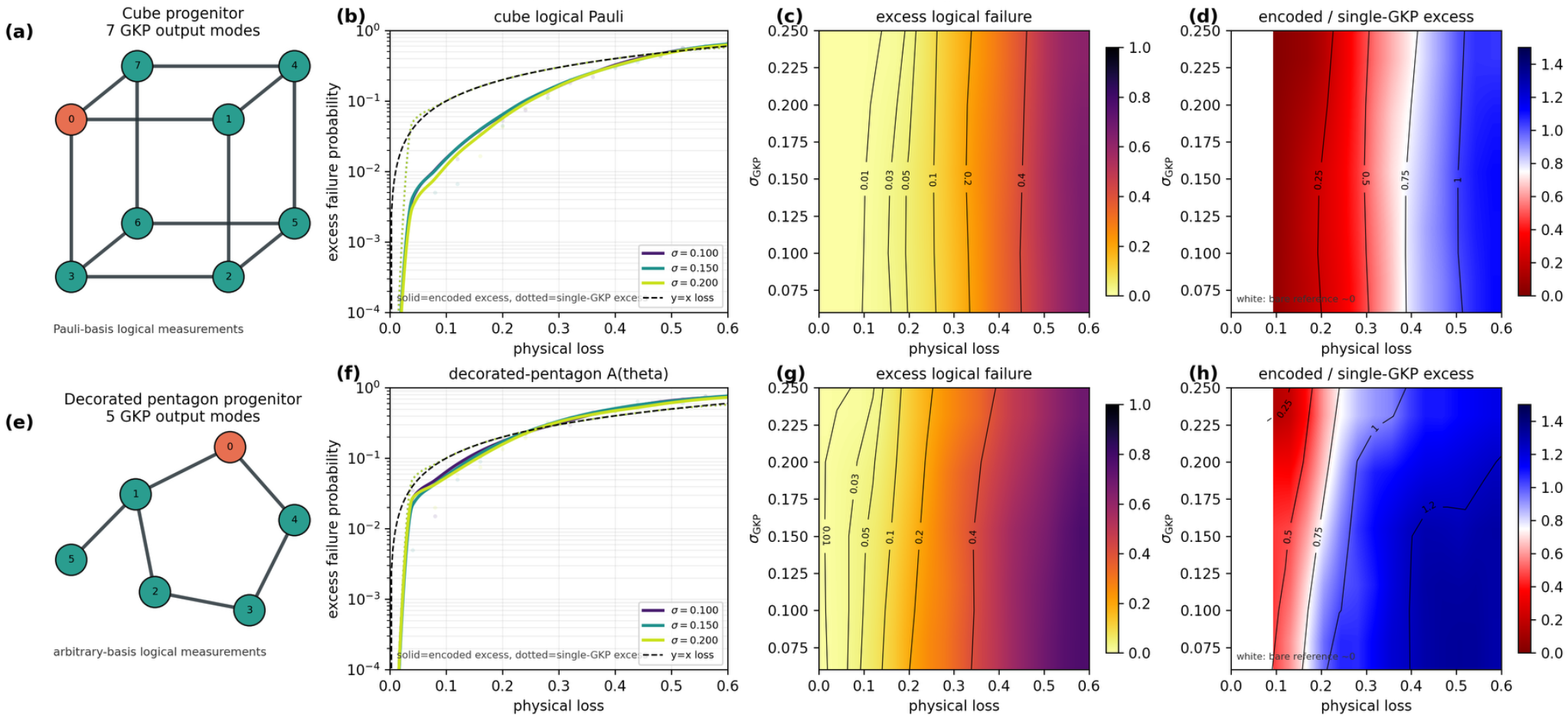}
\caption{Fixed-progenitor closest-coset benchmarks for square GKP blocks. Panels (a)--(d) use the seven-block cube code for logical Pauli measurements; panels (e)--(h) use the five-block decorated-pentagon code for the protected routing and frame layer surrounding an ideal supplied $A(\theta)$ terminal resource. Solid curves in (b) and (f) are encoded excess failure, dotted curves are the matched single-block GKP excess, and the dash-dotted line is the attenuation-coordinate reference. Panels (c) and (g) show encoded excess failure across $(\ell,\sigma_{\GKP})$, while (d) and (h) show the ratio of encoded to single-block excess. The coordinate $\ell=1-\eta$ enters through the Gaussian loss--amplification covariance; graph unavailability is generated by local decoding.}
\label{fig:result1_square_excess}
\end{figure*}

\paragraph{Fixed-progenitor interpretation.}
Reading the attenuation-reference intersections in Fig.~\ref{fig:result1_square_excess}, the cube Pauli module has $\ell_{\rm ref}\simeq0.49$--$0.51$, equivalent through Eq.~\eqref{eq:db_conversions} to approximately $2.9$--$3.1\,\mathrm{dB}$ of channel attenuation. The decorated-pentagon transport/frame benchmark has $\ell_{\rm ref}\simeq0.24$--$0.28$, or $1.2$--$1.4\,\mathrm{dB}$. The ratio-one contours give the distinct excess-noise pseudothreshold ranges: approximately $\ell_{\rm br}\simeq0.43$--$0.50$ ($2.4$--$3.0\,\mathrm{dB}$) for the cube and $\ell_{\rm br}\simeq0.22$--$0.35$ ($1.1$--$1.9\,\mathrm{dB}$) for the $A(\theta)$ transport/frame layer. Reducing the available squeezing from $17.0$ to $11.0\,\mathrm{dB}$ raises the excess-failure surface and narrows the region in which the encoded propagation-induced increment is smaller. These statements concern the closest-coset finite-size benchmark and do not include noise in the terminal non-Gaussian resource.

Figures~\ref{fig:graph_family_catalog}, \ref{fig:result2_hex_family}, and \ref{fig:result2_square_family} compare the selected $n=5,\ldots,11$ graph families for logical-Pauli measurements and for the protected $A(\theta)$ transport/frame task. Figure~\ref{fig:graph_family_catalog} shows the specific progenitor graphs used in the family sweeps; the orange vertex marks the input node $0$ from which the logical-measurement pathfinding starts. Graph size is not monotonically ordered once analog frame errors and decoder-generated erasures are included, and the displayed families were chosen only as representative benchmarks. They are therefore comparisons of the declared graph catalog, with mode-constrained optimization reserved for future work.

\begin{figure*}[!t]
\centering
\resizebox{0.98\linewidth}{!}{%
\begin{tikzpicture}[
  x=1cm,
  y=1cm,
  v/.style={
    circle,
    draw=black,
    line width=0.8pt,
    fill=gkpteal!95,
    minimum size=6mm,
    inner sep=0pt
  },
  vin/.style={
    circle,
    draw=black,
    line width=0.8pt,
    fill=gkporange!95,
    minimum size=6mm,
    inner sep=0pt
  },
  edge/.style={
    line width=0.9pt,
    draw=black!70,
    line cap=round,
    line join=round
  },
  nlabel/.style={font=\Large}
]

\path[use as bounding box] (0.4,-0.1) rectangle (32.2,5.7);

\node[font=\bfseries\LARGE] at (16.4,5.25)
  {Graph-code families selected for measurement simulations};

\node[v] at (12.1,4.20) {};
\node[anchor=west,font=\Large] at (12.5,4.20)
  {remaining graph vertex};

\node[vin] at (19.1,4.20) {};
\node[anchor=west,font=\Large] at (19.5,4.20)
  {input vertex \(0\)};

\begin{scope}[shift={(3.0,1.45)},name prefix=g5-]
\node[nlabel] at (0,2.05) {$n=5$};

\coordinate (a1) at (-0.9,0.0);
\coordinate (a2) at (0.0,1.0);
\coordinate (a3) at (1.1,0.6);
\coordinate (a4) at (1.1,-0.8);
\coordinate (a5) at (0.0,-1.2);

\draw[edge]
  (a1)--(a2)--(a3)--(a4)--(a5)--(a1);

\node[v]   at (a1) {};
\node[v]   at (a2) {};
\node[vin] at (a3) {};
\node[v]   at (a4) {};
\node[v]   at (a5) {};
\end{scope}

\begin{scope}[shift={(7.3,1.45)},name prefix=g6-]
\node[nlabel] at (0,2.05) {$n=6$};

\coordinate (a1) at (-1.4,0.3);
\coordinate (a2) at (-0.4,0.1);
\coordinate (a3) at (0.1,1.0);
\coordinate (a4) at (1.1,0.3);
\coordinate (a5) at (0.9,-0.8);
\coordinate (a6) at (-0.2,-0.9);

\draw[edge]
  (a2)--(a3)--(a4)--(a5)--(a6)--(a2)
  (a1)--(a2);

\node[v]   at (a1) {};
\node[v]   at (a2) {};
\node[v]   at (a3) {};
\node[vin] at (a4) {};
\node[v]   at (a5) {};
\node[v]   at (a6) {};
\end{scope}

\begin{scope}[shift={(11.7,1.45)},name prefix=g7-]
\node[nlabel] at (0,2.05) {$n=7$};

\coordinate (a1) at (-1.2,0.7);
\coordinate (a2) at (-1.5,-0.4);
\coordinate (a3) at (-0.3,-1.0);
\coordinate (a4) at (0.9,-0.6);
\coordinate (a5) at (1.2,0.5);
\coordinate (a6) at (0.0,0.9);
\coordinate (a7) at (-0.2,0.0);

\draw[edge]
  (a1)--(a2)--(a3)--(a4)--(a5)--(a6)--(a1)
  (a6)--(a7)--(a3);

\node[v]   at (a1) {};
\node[v]   at (a2) {};
\node[v]   at (a3) {};
\node[vin] at (a4) {};
\node[v]   at (a5) {};
\node[v]   at (a6) {};
\node[v]   at (a7) {};
\end{scope}

\begin{scope}[shift={(16.0,1.45)},name prefix=g8-]
\node[nlabel] at (0,2.05) {$n=8$};

\coordinate (a1) at (-1.4,0.4);
\coordinate (a2) at (-0.8,1.3);
\coordinate (a3) at (0.3,0.9);
\coordinate (a4) at (1.2,0.4);
\coordinate (a5) at (0.9,-0.7);
\coordinate (a6) at (-0.2,-1.0);
\coordinate (a7) at (-1.2,-0.4);
\coordinate (a8) at (-0.8,0.0);

\draw[edge]
  (a1)--(a2)--(a3)--(a4)--(a5)--(a6)--(a7)--(a1)
  (a2)--(a6)
  (a8)--(a2)
  (a8)--(a5)
  (a1)--(a8);

\node[v]   at (a1) {};
\node[v]   at (a2) {};
\node[vin] at (a3) {};
\node[v]   at (a4) {};
\node[v]   at (a5) {};
\node[v]   at (a6) {};
\node[v]   at (a7) {};
\node[v]   at (a8) {};
\end{scope}

\begin{scope}[shift={(20.5,1.45)},name prefix=g9-]
\node[nlabel] at (0,2.05) {$n=9$};

\coordinate (a1) at (-1.4,0.5);
\coordinate (a2) at (-0.8,1.0);
\coordinate (a3) at (-0.5,-0.1);
\coordinate (a4) at (0.3,0.3);
\coordinate (a5) at (1.0,0.8);
\coordinate (a6) at (1.8,1.1);
\coordinate (a7) at (0.7,-0.6);
\coordinate (a8) at (-0.6,-0.9);
\coordinate (a9) at (-1.6,-0.4);

\draw[edge]
  (a1)--(a2)--(a4)--(a5)--(a7)--(a3)--(a1)
  (a3)--(a4)
  (a3)--(a8)
  (a8)--(a7)
  (a9)--(a1)
  (a9)--(a3)
  (a5)--(a6);

\node[v]   at (a1) {};
\node[v]   at (a2) {};
\node[v]   at (a3) {};
\node[vin] at (a4) {};
\node[v]   at (a5) {};
\node[v]   at (a6) {};
\node[v]   at (a7) {};
\node[v]   at (a8) {};
\node[v]   at (a9) {};
\end{scope}

\begin{scope}[shift={(24.8,1.45)},name prefix=g10-]
\node[nlabel] at (0,2.05) {$n=10$};

\coordinate (a1)  at (-1.4,0.6);
\coordinate (a2)  at (-0.9,1.1);
\coordinate (a3)  at (-0.2,0.8);
\coordinate (a4)  at (0.6,0.9);
\coordinate (a5)  at (0.9,0.1);
\coordinate (a6)  at (0.3,-0.6);
\coordinate (a7)  at (-0.6,-0.7);
\coordinate (a8)  at (-1.2,-0.2);
\coordinate (a9)  at (1.7,-0.1);
\coordinate (a10) at (2.5,-0.3);

\draw[edge]
  (a1)--(a2)--(a3)--(a4)--(a5)--(a6)--(a7)--(a8)--(a1)
  (a3)--(a6)
  (a2)--(a7)
  (a8)--(a3)
  (a5)--(a9)--(a10);

\node[v]   at (a1) {};
\node[vin] at (a2) {};
\node[v]   at (a3) {};
\node[v]   at (a4) {};
\node[v]   at (a5) {};
\node[v]   at (a6) {};
\node[v]   at (a7) {};
\node[v]   at (a8) {};
\node[v]   at (a9) {};
\node[v]   at (a10) {};
\end{scope}

\begin{scope}[shift={(30.0,1.45)},name prefix=g11-]
\node[nlabel] at (0,2.05) {$n=11$};

\coordinate (a1)  at (-1.5,0.2);
\coordinate (a2)  at (-0.9,0.9);
\coordinate (a3)  at (-0.3,1.2);
\coordinate (a4)  at (0.4,0.9);
\coordinate (a5)  at (0.9,0.1);
\coordinate (a6)  at (0.3,-0.7);
\coordinate (a7)  at (-0.5,-0.9);
\coordinate (a8)  at (-1.2,-0.4);
\coordinate (a9)  at (1.7,0.5);
\coordinate (a10) at (1.7,-0.4);
\coordinate (a11) at (0.9,0.6);

\draw[edge]
  (a1)--(a2)--(a3)--(a4)--(a5)--(a6)--(a7)--(a8)--(a1)
  (a2)--(a7)
  (a1)--(a11)
  (a3)--(a11)
  (a4)--(a11)
  (a5)--(a11)
  (a9)--(a11)
  (a10)--(a5);

\node[v]   at (a1) {};
\node[v]   at (a2) {};
\node[v]   at (a3) {};
\node[v]   at (a4) {};
\node[v]   at (a5) {};
\node[v]   at (a6) {};
\node[v]   at (a7) {};
\node[v]   at (a8) {};
\node[v]   at (a9) {};
\node[v]   at (a10) {};
\node[vin] at (a11) {};
\end{scope}

\end{tikzpicture}%
}
\caption{
Representative graph-code families used in the graph-family simulations
of Figs.~\ref{fig:result2_hex_family}
and~\ref{fig:result2_square_family}.
The same fixed catalog is used for both logical-Pauli measurements and
the protected \(A(\theta)\) transport/Pauli-frame task.
Orange marks the input vertex \(0\), while teal marks the remaining
graph-code vertices.
These graphs form the benchmark catalog used in the numerical sweeps
and are not claimed to be optimized at any fixed number of modes.
}
\label{fig:graph_family_catalog}
\end{figure*}
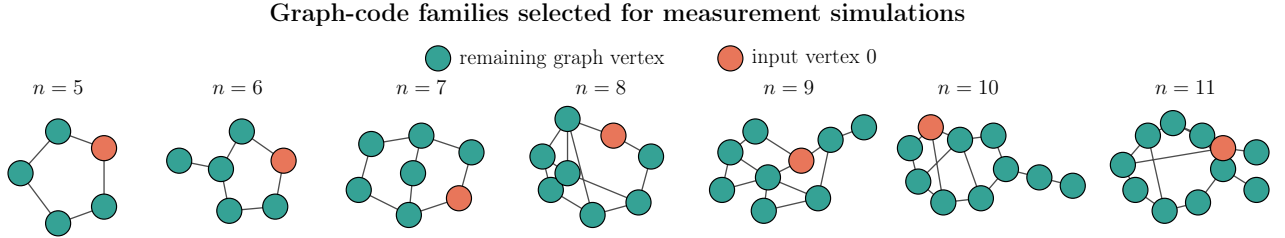

\begin{figure*}[!t]
\centering
\includegraphics[
  width=0.98\linewidth
]{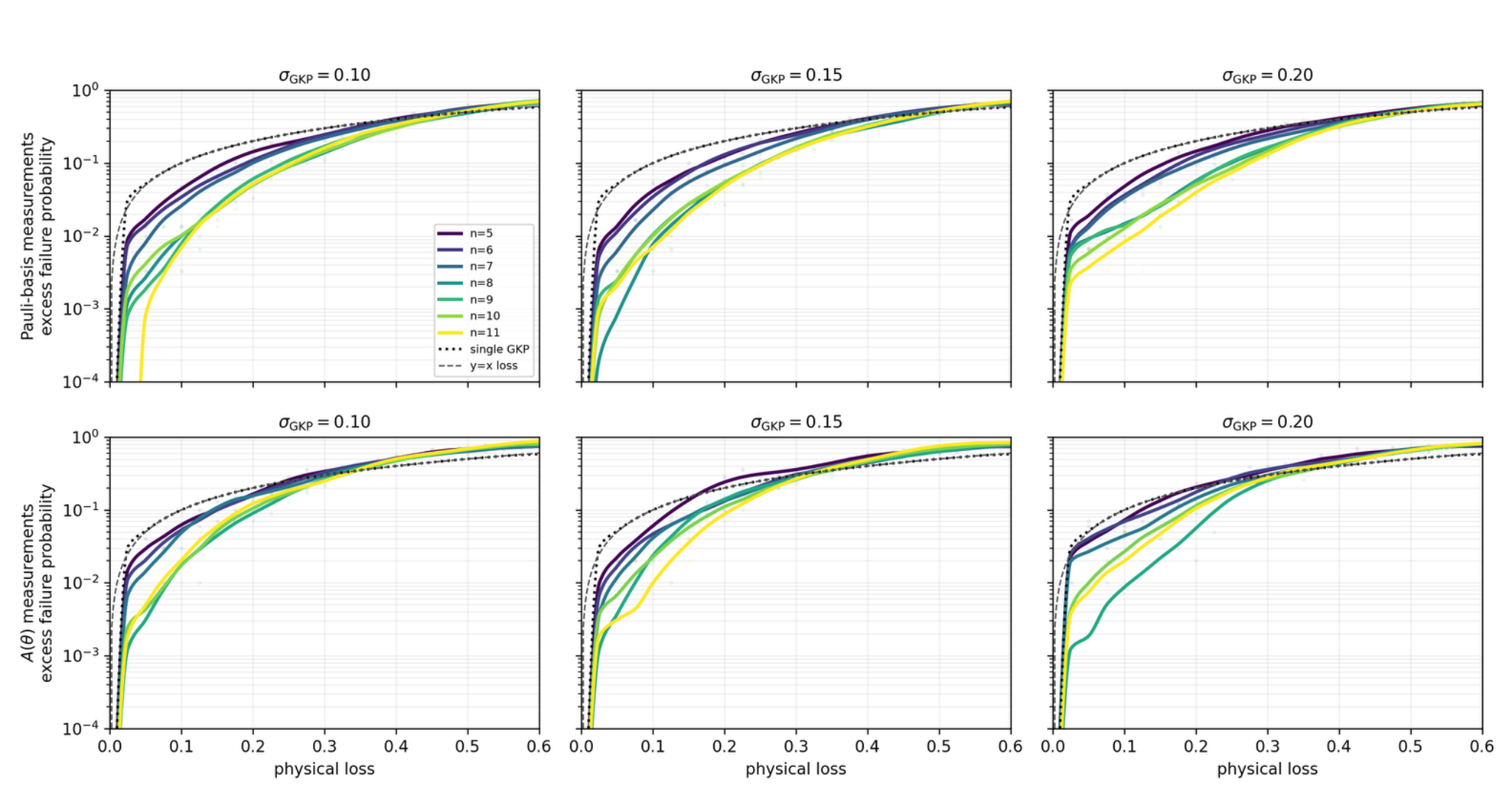}
\caption{
Graph-family excess-failure benchmarks for hexagonal-lattice GKP blocks
decoded with the closest-coset decoder.
The columns correspond to
\(\sigma_{\mathrm{GKP}}=0.10\), \(0.15\), and \(0.20\).
The upper row reports logical-Pauli measurements, while the lower row
reports the protected \(A(\theta)\) transport/Pauli-frame layer, for
graph sizes \(n=5,\ldots,11\).
The dotted black curve gives the matched single-block GKP excess-failure
probability, and the dash-dotted curve gives the attenuation-coordinate
reference.
Crossings with the relevant reference curves define finite-size
pseudothresholds for this fixed, unoptimized graph catalog.
}
\label{fig:result2_hex_family}
\end{figure*}

\paragraph{Hexagonal-family interpretation.}
For the best member of the displayed hexagonal catalog, the Pauli-measurement attenuation-reference crossing lies in the interval $\ell_{\rm ref}\simeq0.48$--$0.51$, corresponding to $2.8$--$3.1\,\mathrm{dB}$ of attenuation. The best protected $A(\theta)$ transport/frame crossing is $\ell_{\rm ref}\simeq0.34$--$0.37$, or $1.8$--$2.0\,\mathrm{dB}$. Across all sizes $n=5,\ldots,11$, the topology-dependent spread is larger than the shift produced by changing $\sigma_{\GKP}$ from $0.10$ ($17.0\,\mathrm{dB}$) to $0.20$ ($11.0\,\mathrm{dB}$): graph choice controls the displayed reference crossing, whereas finite squeezing mainly changes the vertical failure level. The nonmonotonic ordering with $n$ demonstrates topology dependence within the declared catalog, not an optimized size law.

\begin{figure*}[!t]
\centering
\includegraphics[width=0.98\textwidth]{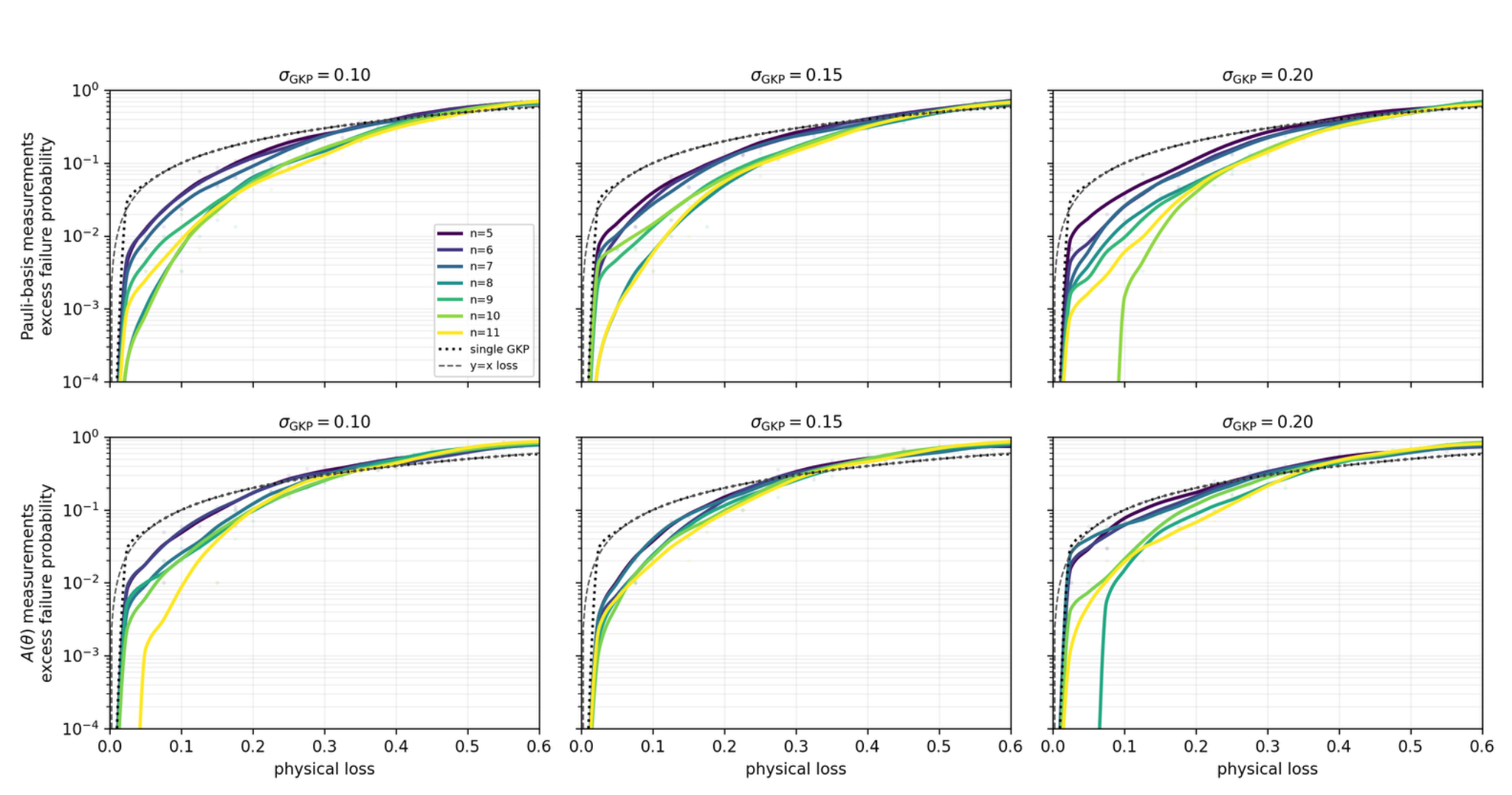}
\caption{Square-lattice counterpart of Fig.~\ref{fig:result2_hex_family}; all graph, decoder, ideal-terminal-resource, and reference-crossing conventions are unchanged.}
\label{fig:result2_square_family}
\end{figure*}

\paragraph{Square-family interpretation.}
The square-lattice catalog in Fig.~\ref{fig:result2_square_family} gives best Pauli attenuation-pseudothresholds $\ell_{\rm ref}\simeq0.47$--$0.50$ ($2.8$--$3.0\,\mathrm{dB}$) and best protected $A(\theta)$ transport/frame pseudothresholds $\ell_{\rm ref}\simeq0.34$--$0.35$ ($1.8$--$1.9\,\mathrm{dB}$). These values are close to the hexagonal results, but the hexagonal lattice generally lowers the excess-failure probability at fixed $(\ell,s_{\GKP}^{\rm dB})$ because of its larger shortest logical displacement. Consequently, the lattice effect appears more clearly in the vertical separation of the finite-size curves than in a large displacement of the attenuation-reference crossing.

\section{Extending to larger graphs by modularization}
\label{sec:modularization}
The static lattice substitution is unchanged by the adopted channel, but the effective channel passed between layers is not a scalar physical-loss probability.  Every child exports a Pauli-frame posterior and an availability decision generated by its decoder.  The internal origin of an unavailable child, $\accera$ or $\confera$, is retained for accounting, while the parent sees a located unavailable interface.

\subsection{Recursive graph--GKP modules and finite-depth benchmarks}

A level-one graph--GKP module is specified by
\begin{equation*}
\mathcal M_G=
\bigl(M_G,\phi_G,L_G,J_G,
\{e_{G,\alpha},f_{G,\alpha}\},\mathcal D_G,
\Gamma_{\rm loc},\Gamma_{\rm module}\bigr),
\end{equation*}
where the first five entries define the compiled oscillator code and $\mathcal D_G=\mathcal D_G^{\rm CC}$ is the fixed selected-branch closest-coset engine of Sec.~\ref{subsec:outcome_likelihood}. For physical parameters
$\lambda_{\rm phys}=(\bm\eta,\Sigma_{\rm prep},\Sigma_{\rm gate},\Sigma_{\rm hom})$, its syndrome-resolved channel is
\begin{equation*}
\mathcal N_G^{\rm post}(\lambda_{\rm phys}):
D_b\longmapsto
\bigl(\Pi_b(\sigma,e\mid D_b),\Pi_b^{\rm fr},
\widehat f_b,\Gamma_b^{\rm fr},a_b^{\rm int}\bigr),
\end{equation*}
where $a_b^{\rm int}=0$ if the module emits $\accera$ or $\confera$ at the declared interface.  The posterior remains part of the stored decoder record even when the interface rejects the module.  This prevents a parent decoder from treating a data-dependent erasure as an independent Bernoulli event.

\paragraph{Recursive lattice substitution.}

Suppose each physical qubit of an outer graph code is replaced by a level-$k$ graph--GKP module with lattice $L_G^{[k]}$ and logical representatives $e_i^{[k]},f_i^{[k]}$.  The child product code has quotient
\begin{equation*}
\mathsf H_{\rm child}^{[k]}
=\left(\bigoplus_i (L_G^{[k]})^\perp\right)
 \bigg/\left(\bigoplus_i L_G^{[k]}\right),
\end{equation*}
and the outer binary code is embedded by $\overline\Phi_G^{[k+1]}$.  A chosen lift is
\begin{equation*}
\Phi_G^{[k+1]}(x\mid z)
=\sum_{i=1}^{n_G}x_ie_i^{[k]}
+\sum_{i=1}^{n_G}z_if_i^{[k]}.
\end{equation*}
The level-$(k+1)$ lattice is the same quotient pullback as at level one,
\begin{align}
L_G^{[k+1]}
&=(\pi_{\rm child}^{[k]})^{-1}
 \!\left(\overline\Phi_G^{[k+1]}(Q_G(G))\right)\nonumber\\
&=\left(\bigoplus_{i=1}^{n_G}L_G^{[k]}\right)
+\spanZ\{\Phi_G^{[k+1]}(q):q\in Q_G(G)\}.
\label{eq:recursive_lattice}
\end{align}
Thus the loss--amplification channel changes the probabilistic interface but not the compiled lattice or logical quotient.

For conditionally independent children, the posterior-channel recursion is
\begin{equation}
\mathcal N_{k+1}^{\rm post}
=\mathcal R_G^{\rm post}
\bigl(\mathcal N_k^{\rm post}\bigr),
\qquad
\mathcal N_1^{\rm post}=\mathcal N_G^{\rm post}(\lambda_{\rm phys}),
\label{eq:channel_recursion}
\end{equation}
where $\mathcal R_G^{\rm post}$ multiplies the selected child likelihoods only after including each child selection/erasure event and then performs the outer syndrome and frame pushforward.  Correlated child records require their joint channel rather than the product in Eq.~\eqref{eq:channel_recursion}.  Algorithm~\ref{alg:recursive_decoder} applies this recursion at depth $k$.

\begin{prxalgorithm}{Recursive graph--GKP module decoding}
\label{alg:recursive_decoder}
\footnotesize
\begin{algorithmic}[1]
\Require Outer graph module $G$; depth $k$; physical loss--amplification parameters; local and module confidence rules; child joint-noise model.
\Ensure Level-$k$ syndrome-resolved frame channel and optional hard interface flag.
\If{$k=1$}
  \State Run Algorithm~\ref{alg:local_interface} on every block, construct the decoder-conditioned branch with Algorithm~\ref{alg:logical_branch}, and decode it with Algorithm~\ref{alg:syndrome_frame_decoder}; \Return $\mathcal N_1^{\rm post}$.
\EndIf
\State Decode every level-$(k-1)$ child and retain its complete frame posterior, hard interface flag, and the likelihood of the accepted/rejected event.
\State Construct $L_G^{[k]}$ by Eq.~\eqref{eq:recursive_lattice} and translate the child interface flags into the parent availability record.
\State Run the parent pathfinding policy before any destructive parent-level measurement.
\If{no parent representative survives}
  \State Return $\accera$ while retaining the complete child evidence.
\EndIf
\State Form the selected parent likelihood from the product of child posterior distributions, or from their declared joint channel when correlated.
\State Normalize the parent sector posterior, push it to the parent frame posterior, and threshold only at the declared parent interface.
\State \Return the level-$k$ posterior distribution and the disjoint $\accera/\confera$/accepted-frame label.
\end{algorithmic}
\end{prxalgorithm}

\paragraph{Accessibility recursion.}

If every child independently emits an unavailable flag with probability $\varepsilon_k$ and all accepted-frame errors are ignored, the two worked modules reduce to scalar accessibility recursions
\begin{align*}
\varepsilon_{k+1}^{\rm cube}
&=7\varepsilon_k^3-21\varepsilon_k^5
 +21\varepsilon_k^6-6\varepsilon_k^7,\\
\varepsilon_{k+1}^{\rm dec}
&=5\varepsilon_k^2-5\varepsilon_k^3+\varepsilon_k^5.\end{align*}
These maps are exact only for the availability subproblem under IID child flags.  They are useful algebraic checks of an implementation but are not the full recursive loss--amplification channel.  A physically meaningful recursion must additionally propagate accepted Pauli-frame errors, confidence distributions, branch-selection likelihoods, and any inter-block correlations.

The finite-squeezing simulations propagate complete closest-coset child posterior distributions rather than the scalar unit-test maps. Figures~\ref{fig:result3_hex_concat} and \ref{fig:result3_square_concat} compare depths $L=1,2,3,4$ for the Pauli and protected $A(\theta)$ transport/frame tasks. For continuity with the fixed-size panels we quote intersections with the attenuation-coordinate reference.  A level-to-level pseudothreshold, when a crossing exists, would instead satisfy
\begin{equation}
P_{\rm fail}^{(L+1)}(\ell_{\rm pseudo}^{(L)},\sigma_{\GKP})
=P_{\rm fail}^{(L)}(\ell_{\rm pseudo}^{(L)},\sigma_{\GKP}).
\label{eq:finite_depth_pseudothreshold}
\end{equation}

\begin{figure*}[!t]
\centering
\includegraphics[width=0.98\textwidth]{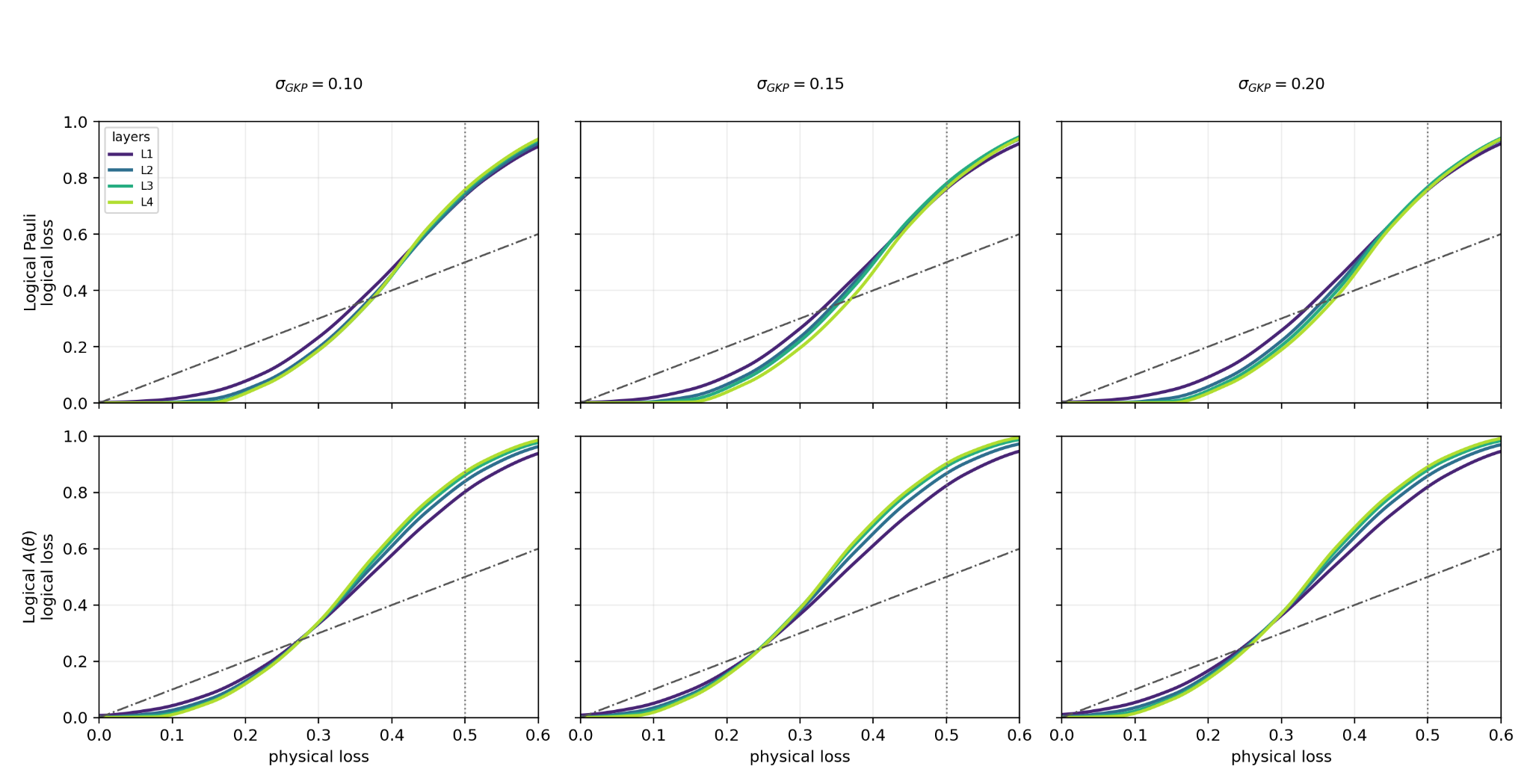}
\caption{Finite-depth recursive closest-coset loss sweeps for hexagonal GKP blocks. Columns use $\sigma_{\GKP}=0.10$, $0.15$, and $0.20$; the upper row is the cube logical-Pauli task and the lower row is the decorated-pentagon protected $A(\theta)$ transport/frame task. Curves L1--L4 correspond to one through four self-concatenation layers. The dash-dotted line is the attenuation-coordinate reference and the vertical dotted line marks $\ell=0.5$. }
\label{fig:result3_hex_concat}
\end{figure*}

\paragraph{Hexagonal-recursion interpretation.}
The finite-depth Pauli attenuation-reference pseudothresholds in Fig.~\ref{fig:result3_hex_concat} span $\ell_{\rm ref}\simeq0.33$--$0.37$, or $1.7$--$2.0\,\mathrm{dB}$, over $L=1,\ldots,4$ and the three squeezing values. The protected $A(\theta)$ transport/frame speudothresholds span $\ell_{\rm ref}\simeq0.24$--$0.27$, equivalent to $1.2$--$1.4\,\mathrm{dB}$. Additional layers lower the displayed logical failure in part of the low-attenuation regime and shift the reference intersections with depth. The degradation from $17.0$ to $11.0\,\mathrm{dB}$ squeezing appears as a higher logical-failure curve. Four finite layers do not establish asymptotic suppression.

\begin{figure*}[!t]
\centering
\includegraphics[width=0.98\textwidth]{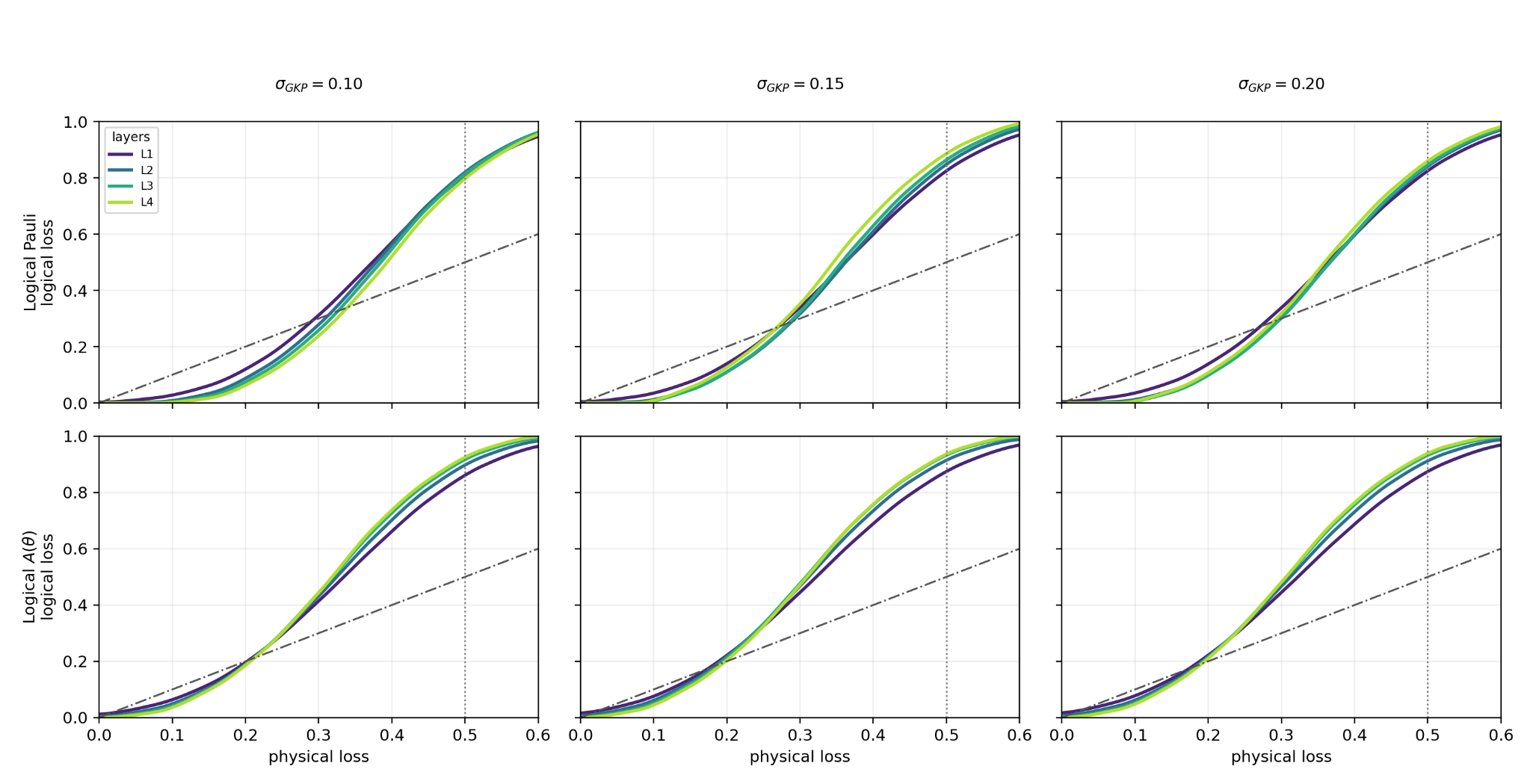}
\caption{Square-lattice counterpart of the finite-depth recursion study in Fig.~\ref{fig:result3_hex_concat}; the same ideal-terminal-resource and reference-crossing conventions apply.}
\label{fig:result3_square_concat}
\end{figure*}

\paragraph{Square-recursion interpretation.}
For square GKP blocks, the Pauli attenuation-reference pseudothresholds range from approximately $\ell_{\rm ref}=0.25$ to $0.34$, corresponding to $1.25$--$1.8\,\mathrm{dB}$, while the protected $A(\theta)$ transport/frame pseudothresholds range from $\ell_{\rm ref}=0.17$ to $0.21$, or $0.81$--$1.0\,\mathrm{dB}$. Relative to the hexagonal recursion, the smaller local distance reduces both displayed crossover windows. Increasing the depth lowers logical failure over part of the low-attenuation range.

\section{Logical graph-state fusions with GKP blocks}
\label{sec:fusion}

Fusion measurements are the second central primitive of photonic measurement-based architectures. A logical fusion should infer two commuting parity bits while distinguishing successful fusion, partial erasure, full erasure, and an accepted but wrong parity.

\subsection{Bell sectors and syndrome-resolved parity decoding}

Consider two independently prepared graph--GKP modules $A$ and $B$ with product lattice
\begin{equation*}
L_{AB}=L_G^{(A)}\oplus L_G^{(B)},
\qquad
J_{AB}=J_A\oplus J_B.
\end{equation*}
For boundary logical sites $a$ and $b$, the commuting parity observables are represented by
\begin{equation*}
f_X=e_a^{(A)}+e_b^{(B)},
\qquad
f_Z=f_a^{(A)}+f_b^{(B)}.
\end{equation*}
They commute because the two local half-integer symplectic pairings add to an integer. The four parity sectors are not generated by shifting with $f_X$ and $f_Z$ themselves, because a measured parity commutes with its own subgroup and does not change its eigenvalue. Instead choose dual shifts $g_X,g_Z$ satisfying
\begin{align}
2f_X^TJ_{AB}g_X&=1,&2f_Z^TJ_{AB}g_Z&=1,\nonumber\\
2f_X^TJ_{AB}g_Z&=0,&2f_Z^TJ_{AB}g_X&=0
\qquad(\bmod 2).
\label{eq:dual_shifts}
\end{align}
Let $\mathcal Q_{AB}\simeq\F_2^4$ denote the two-module logical Pauli quotient and define
\begin{equation*}
\varpi(v)=\bigl(2f_X^TJ_{AB}v,2f_Z^TJ_{AB}v\bigr)\pmod2.
\end{equation*}
The common periodicity after adjoining the measured parity subgroup is
\begin{equation*}
L_{\Bell}=L_{AB}+\spanZ\{f_X,f_Z\}.
\end{equation*}
For a parity-action label $E=(e_X,e_Z)\in\F_2^2$, define the affine sector
\begin{equation*}
\mathcal C^{E}_{\fus}=e_Xg_X+e_Zg_Z+L_{\Bell}.
\end{equation*}

\begin{theorem}[CV/GKP Bell parity-sector reduction]
\label{thm:fusion_partition}
The map $\varpi:\mathcal Q_{AB}\to\F_2^2$ is surjective and
\begin{equation*}
\ker\varpi=H_{\Bell}:=\spanF\{f_X,f_Z\}.
\end{equation*}
The dual shifts in Eq.~\eqref{eq:dual_shifts} define a section of $\varpi$. Consequently every two-logical-qubit Pauli class has a unique decomposition
\begin{equation}
v=e_Xg_X+e_Zg_Z+h,
\qquad (e_X,e_Z)\in\F_2^2,
\quad h\in H_{\Bell}.
\label{eq:fusion_unique_decomposition}
\end{equation}
Under the faithful displacement embedding, adjoining $H_{\Bell}$ to the product-lattice periodicity gives $L_{\Bell}$, and the four affine sets $\mathcal C^E_{\fus}$ are precisely the four parity-action sectors: they are disjoint modulo $L_{\Bell}$ and together contain all sixteen two-module logical Pauli classes.
\end{theorem}

Theorem~\ref{thm:fusion_partition} is a direct consequence of rank--nullity for the parity map together with the faithful displacement embedding. Figure~\ref{fig:fusion_logic} illustrates the resulting sector logic.

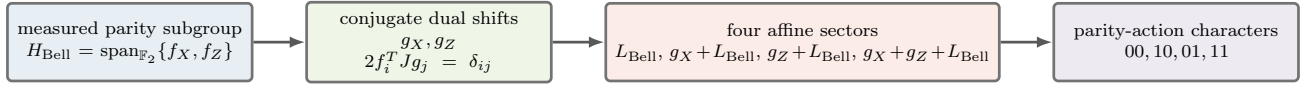
\begin{figure*}[!t]
\centering
\resizebox{0.96\textwidth}{!}{\begin{tikzpicture}[
  >=Latex,
  every node/.style={font=\footnotesize},
  box/.style={draw=black!55,rounded corners=2pt,very thick,align=center,text width=35mm,minimum height=12mm,inner sep=4pt},
  arrow/.style={-{Latex[length=2.4mm]},very thick,draw=black!65}
]
\node[box,fill=gkpblue!11] (parity) at (0,0) {measured parity subgroup\\$H_{\rm Bell}=\mathrm{span}_{\F_2}\{f_X,f_Z\}$};
\node[box,fill=gkpgreen!12,right=8mm of parity] (dual) {conjugate dual shifts\\$g_X,g_Z$\\$2f_i^TJg_j=\delta_{ij}$};
\node[box,fill=gkporange!13,right=8mm of dual,text width=58mm] (sectors) {four affine sectors\\$L_{\rm Bell}$, $g_X+L_{\rm Bell}$, $g_Z+L_{\rm Bell}$, $g_X+g_Z+L_{\rm Bell}$};
\node[box,fill=gkppurple!12,right=8mm of sectors] (chars) {parity-action characters\\$00,10,01,11$};
\draw[arrow] (parity)--(dual); \draw[arrow] (dual)--(sectors); \draw[arrow] (sectors)--(chars);
\end{tikzpicture}}
\caption{Fusion-sector logic. The commuting measured parity generators define the Bell parity subgroup. Outcome sectors are labeled by characters of this subgroup, equivalently by conjugate dual shifts satisfying Eq.~\eqref{eq:dual_shifts}. Shifting by a measured parity generator itself would leave its eigenvalue unchanged.}
\label{fig:fusion_logic}
\end{figure*}

\paragraph{Parity likelihood and erasure classes.}

For single-mode boundary blocks, a CV Bell measurement may be implemented with a $50{:}50$ beam splitter followed by homodyne measurements of
\begin{equation*}
q_-=(q_a-q_b)/\sqrt 2,
\qquad
p_+=(p_a+p_b)/\sqrt 2.
\end{equation*}
The four sectors $\mathcal C^{E}_{\fus}$ are disjoint modulo $L_{\Bell}$ and exhaust the sixteen two-logical-qubit Pauli classes after grouping them by their action on the measured parities.

Before a Bell attempt, both interface blocks undergo Algorithm~\ref{alg:local_interface}.  The complete fusion record is
\begin{equation}
D_{\fus}=\bigl(D_A^{\rm loc},D_B^{\rm loc},a_A,a_B,
D_{\Bell}^{\rm out}\bigr).
\label{eq:loss_amplification_fusion_record}
\end{equation}
If either local decoder abstains, the corresponding physical Bell attempt is omitted and the rejected local syndrome remains evidence.  If both are accepted, $D_{\Bell}^{\rm out}$ contains the two Bell homodyne records.  The selected fusion likelihood therefore has the causal form
\begin{equation*}
p_{\fus}(D_{\fus},b\mid E)
=\mathbf1[b=\mathcal B_{\fus}(D_A^{\rm loc},D_B^{\rm loc})]
 p_{\LA}(D_{\fus}\mid E),
\end{equation*}
with local $H^{\rm acc}$ or $H^{\rm era}$ factors exactly as in Eq.~\eqref{eq:loss_amplification_factorized_score}.  There is no independent photon-loss draw in this channel description.  A probabilistic linear-optical Bell analyzer may be composed as a separate declared hardware channel, but it is not generated by $1-\eta$ in the adopted channel.

Let $P_{\fus}$ be the Bell observation map, $r_{\fus}$ the visible residual, and $\Sigma_{\fus,\obs}$ the covariance. With $K^{\rm lat}_{\fus}=L_{\Bell}\cap\ker P_{\fus}$, the exact translation-covariant maximum-likelihood target for parity sector $E$ is
\begin{align*}
W_{\fus,{\rm ML}}(E\mid r_{\fus})
&=\sum_{[\lambda]\in\mathcal C^E_{\fus}/K^{\rm lat}_{\fus}}
\exp[-Q_{\fus}(\lambda)/2],\\
Q_{\fus}(\lambda)
&=(r_{\fus}-P_{\fus}\lambda)^T
\Sigma_{\fus,\obs}^{-1}
(r_{\fus}-P_{\fus}\lambda).
\end{align*}
The implemented fusion decoder uses the closest-coset score
\begin{align*}
D_{\fus}^{\rm CC}(E)
&=\min_{[\lambda]\in\mathcal C^E_{\fus}/K^{\rm lat}_{\fus}}
Q_{\fus}(\lambda),\\
\widehat W_{\fus,{\rm CC}}(E)
&=e^{-D_{\fus}^{\rm CC}(E)/2},
\end{align*}
with the discrete local and Bell factors included as additive negative-log costs when present. If the Bell readout is not translation covariant, each CC sector is assigned its calibrated normalizer. The primary fusion output is the normalized closest-coset distribution over the parity-frame error class,
\begin{equation}
\Pi_{\fus}(E\mid D_{\fus})
=
\frac{\pi(E)\widehat p_{\fus}^{\rm CC}(D_{\fus}\mid E)}
{\sum_{E'}\pi(E')\widehat p_{\fus}^{\rm CC}(D_{\fus}\mid E')},
\label{eq:fusion_posterior}
\end{equation}
where $D_{\fus}$ contains the discrete Bell data and analog residual. In the translation-covariant Gaussian model, $\widehat p_{\fus}^{\rm CC}$ is proportional to $\widehat W_{\fus,{\rm CC}}$. The maximum-score parity-frame estimate and confidence are
\begin{align*}
\widehat E
&=\argmax_E\Pi_{\fus}(E\mid D_{\fus}),\\
\Gamma_{\fus}
&=\log\frac{\Pi_{\fus}(\widehat E\mid D_{\fus})}
{1-\Pi_{\fus}(\widehat E\mid D_{\fus})}.\end{align*}
The estimated $\widehat E$ is the Pauli-frame update passed to the outer fusion-network controller. The reported parity values are only the pushforward
\begin{align*}
\Pr_{\fus}(s_X,s_Z\mid D_{\fus})
={}&\sum_E\Pi_{\fus}(E\mid D_{\fus})\\
&\times\mathbf1[(s_X,s_Z)=\chi_{\fus}\oplus E].\end{align*}
If module syndromes are unresolved, Eq.~\eqref{eq:fusion_posterior} is first formed on $(\sigma_A,e_A,\sigma_B,e_B)$ and then pushed forward through the parity-action map $\varpi$.

The two marginal parity confidences may be thresholded independently when the architecture permits a one-parity erasure. Both accepted parities give a successful logical fusion unless the estimated parity frame is wrong; one accepted parity gives a partial erasure; neither gives a full erasure. If one Bell quadrature is absent, the associated precision vanishes and the corresponding posterior becomes uninformative. A separate $p_{\rm fail}=1/2$ linear-optical Bell-failure layer, when relevant to a particular hardware encoding, must be composed explicitly and is not intrinsic to the CV/GKP likelihood. An ``unflagged fusion fault'' is the event that the accepted frame estimate differs from the simulated true class; it is not observable by the decoder during operation.

Algorithm~\ref{alg:fusion_decoder} gives the corresponding logical parity-fusion decoder.

\begin{prxalgorithm}{Syndrome-resolved logical parity-fusion decoder}
\label{alg:fusion_decoder}
\footnotesize
\begin{algorithmic}[1]
\Require Two interface modules under the adopted loss--amplification channel; local records and confidence rules; Bell readout circuit; module priors; parity-confidence thresholds.
\Ensure Posterior over parity-frame classes, MAP frame/confidence, parity marginals, and disjoint erasure labels.
\State Run Algorithm~\ref{alg:local_interface} on both interface blocks and retain both local posteriors and decisions.
\If{either interface is locally erased}
  \State Omit the Bell attempt, retain the corresponding rejection likelihood, and expose the declared located fusion-interface erasure.
\Else
  \State Perform the Bell measurement and append $D_{\Bell}^{\rm out}$ to Eq.~\eqref{eq:loss_amplification_fusion_record}.
\EndIf
\State Construct $f_X,f_Z$, verify dual shifts $g_X,g_Z$, and form $L_{\Bell}=L_{AB}+\spanZ\{f_X,f_Z\}$.
\For{each retained $(\sigma_A,e_A,\sigma_B,e_B)$}
  \State Compute $E=\varpi(e_A,e_B)$ and evaluate the selected-branch closest-coset score, including local accept/reject factors and any intrinsic Bell-analyzer outcome.
\EndFor
\State Normalize the joint weights and push them forward to $\Pi_{\fus}(E\mid D_{\fus})$.
\State Compute $(\widehat E,\Gamma_{\fus})$ and the two parity marginals.
\State Apply the declared marginal or joint confidence rule: one rejected parity is a partial erasure, two are a full erasure, and an accepted mismatch $\widehat E\neq E$ is a fusion-frame fault.
\State \Return the complete parity-frame posterior and hard interface label.
\end{algorithmic}
\end{prxalgorithm}

Figure~\ref{fig:fusion_schematic} summarizes the transversal and adaptive fusion patterns used below.

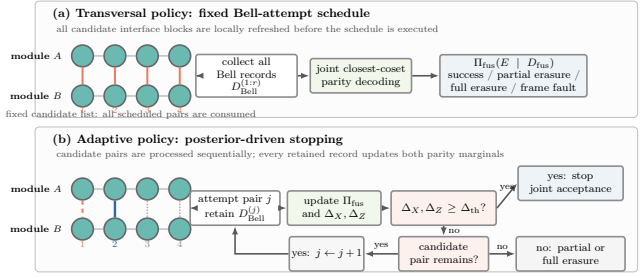
\begin{figure}[!t]
\centering
\resizebox{0.98\columnwidth}{!}{\begin{tikzpicture}[
  x=1cm,y=1cm,>=Latex,
  every node/.style={font=\footnotesize},
  panel/.style={draw=black!25,rounded corners=3pt,fill=black!1},
  q/.style={circle,draw=black!60,very thick,minimum size=5.7mm,inner sep=0pt},
  pair/.style={line width=1.5pt,draw=gkporange!90},
  chosen/.style={line width=1.8pt,draw=gkpblue},
  pending/.style={line width=1.2pt,draw=black!30,densely dotted},
  flow/.style={-{Latex[length=2.2mm]},very thick,draw=black!65},
  box/.style={draw=black!45,rounded corners=2pt,very thick,align=center,
              fill=white,minimum height=9mm,inner sep=2.2pt},
  decoder/.style={box,fill=gkpgreen!10},
  decision/.style={box,fill=gkporange!10},
  output/.style={box,fill=gkpblue!8},
  note/.style={font=\scriptsize,align=center,text=black!70}
]

\node[panel,minimum width=15.8cm,minimum height=3.05cm,anchor=south west]
  at (0,4.15) {};
\node[anchor=west,font=\small\bfseries] at (0.35,6.88)
  {(a) Transversal policy: fixed Bell-attempt schedule};
\node[note,anchor=west] at (0.45,6.47)
  {all candidate interface blocks are locally refreshed before the schedule is executed};

\node[anchor=east,font=\scriptsize\bfseries] at (0.82,5.77) {module $A$};
\node[anchor=east,font=\scriptsize\bfseries] at (0.82,4.72) {module $B$};

\foreach \x/\lab in {1.25/1,2.10/2,2.95/3,3.80/4}{
  \node[q,fill=gkpteal!70] (TA\lab) at (\x,5.77) {};
  \node[q,fill=gkpteal!70] (TB\lab) at (\x,4.72) {};
}
\draw[black!35,thick] (TA1)--(TA2)--(TA3)--(TA4);
\draw[black!35,thick] (TB1)--(TB2)--(TB3)--(TB4);
\foreach \i in {1,2,3,4}{
  \draw[pair] (TA\i)--(TB\i);
  \node[font=\scriptsize,text=gkporange!90] at ($(TB\i)+(0,-0.38)$) {$\i$};
}
\node[note] at (2.52,4.23) {fixed candidate list: all scheduled pairs are consumed};

\node[box,text width=25mm] (recordsT) at (5.55,5.25)
  {collect all Bell records\\$D_{\rm Bell}^{(1:r)}$};
\node[decoder,text width=25mm] (decodeT) at (8.55,5.25)
  {joint closest-coset\\parity decoding};
\node[output,text width=37mm] (outT) at (12.55,5.25)
  {$\Pi_{\rm fus}(E\mid D_{\rm fus})$\\success / partial erasure /\\full erasure / frame fault};

\draw[flow] (4.25,5.25)--(recordsT.west);
\draw[flow] (recordsT.east)--(decodeT.west);
\draw[flow] (decodeT.east)--(outT.west);

\node[panel,minimum width=15.8cm,minimum height=3.80cm,anchor=south west]
  at (0,0.10) {};
\node[anchor=west,font=\small\bfseries] at (0.35,3.57)
  {(b) Adaptive policy: posterior-driven stopping};
\node[note,anchor=west] at (0.45,3.17)
  {candidate pairs are processed sequentially; every retained record updates both parity marginals};

\node[anchor=east,font=\scriptsize\bfseries] at (0.82,2.28) {module $A$};
\node[anchor=east,font=\scriptsize\bfseries] at (0.82,1.23) {module $B$};

\foreach \x/\lab in {1.25/1,2.10/2,2.95/3,3.80/4}{
  \node[q,fill=gkpteal!70] (AA\lab) at (\x,2.28) {};
  \node[q,fill=gkpteal!70] (AB\lab) at (\x,1.23) {};
}
\draw[black!35,thick] (AA1)--(AA2)--(AA3)--(AA4);
\draw[black!35,thick] (AB1)--(AB2)--(AB3)--(AB4);
\draw[pair,dashed] (AA1)--(AB1);
\draw[chosen] (AA2)--(AB2);
\draw[pending] (AA3)--(AB3);
\draw[pending] (AA4)--(AB4);
\node[font=\scriptsize,text=gkporange!90] at ($(AB1)+(0,-0.38)$) {$1$};
\node[font=\scriptsize,text=gkpblue] at ($(AB2)+(0,-0.38)$) {$2$};
\node[font=\scriptsize,text=black!45] at ($(AB3)+(0,-0.38)$) {$3$};
\node[font=\scriptsize,text=black!45] at ($(AB4)+(0,-0.38)$) {$4$};
\node[box,text width=22mm] (attemptA) at (5.25,1.78)
  {attempt pair $j$\\retain $D_{\rm Bell}^{(j)}$};
\node[decoder,text width=23mm] (updateA) at (7.85,1.78)
  {update $\Pi_{\rm fus}$\\and $\Delta_X,\Delta_Z$};
\node[decision,text width=26mm] (testA) at (10.72,1.78)
  {$\Delta_X,\Delta_Z\geq\Delta_{\rm th}$?};
\node[output,text width=26mm] (acceptA) at (14.05,2.43)
  {yes: stop\\joint acceptance};
\node[decision,text width=22mm] (remainA) at (10.72,0.58)
  {candidate pair remains?};
\node[box,fill=black!3,text width=18mm] (nextA) at (7.65,0.58)
  {yes: $j\leftarrow j+1$};
\node[box,fill=black!4,text width=27mm] (eraseA) at (14.05,0.58)
  {no: partial or full erasure};

\draw[flow] (4.25,1.78)--(attemptA.west);
\draw[flow] (attemptA.east)--(updateA.west);
\draw[flow] (updateA.east)--(testA.west);
\draw[flow] (testA.east)--(acceptA.west) node[midway,above,font=\scriptsize] {yes};
\draw[flow] (testA.south)--(remainA.north) node[midway,right,font=\scriptsize] {no};
\draw[flow] (remainA.west)--(nextA.east) node[midway,above,font=\scriptsize] {yes};
\draw[flow] (nextA.west)-|(attemptA.south);
\draw[flow] (remainA.east)--(eraseA.west) node[midway,above,font=\scriptsize] {no};

\end{tikzpicture}}
\caption{Transversal and adaptive logical-fusion policies. All candidate
interface blocks are locally refreshed before any Bell attempt.
(a) The transversal policy executes a fixed set of pairwise GKP Bell
measurements, collects the complete Bell record, and performs joint
closest-coset parity decoding. (b) The adaptive policy processes candidate
pairs sequentially and updates the fusion posterior
$\Pi_{\fus}(E\mid D_{\fus})$ and parity margins
$\Delta_X,\Delta_Z$ after every retained Bell record. It stops only when
both parity margins satisfy the declared joint acceptance rule; otherwise
it advances to the next candidate and returns a partial or full erasure
if the candidate list is exhausted. In the adaptive panel, the orange
dashed link denotes an inconclusive record, the blue link the first record
satisfying the joint rule, and the gray dotted links candidates that are
not attempted. The diagram specifies causal ordering and stopping logic,
not a resource-normalized comparison.}
\label{fig:fusion_schematic}
\end{figure}

\subsection{Fusion accounting and finite-squeezing benchmarks}
\label{subsec:fusion_benchmark}

Algorithm~\ref{alg:fusion_decoder} induces the outer-interface channel
\begin{equation}
\bigl(P_{\fus}^{\rm succ},P_{\fus}^{\rm part},
P_{\fus}^{\rm full},P_{\fus}^{\rm err}\bigr),
\label{eq:loss_amplification_fusion_channel_vector}
\end{equation}
where the last component is an accepted parity-frame mismatch. Local interface erasures are already included through the adaptive fusion policy and are not added again when routing succeeds.

Figures~\ref{fig:result4_hex_fusion} and \ref{fig:result4_square_fusion} compare adaptive and transversal fusion at $\Delta_{\rm th}=0.50$ for $n=4,\ldots,9$, using the same closest-coset parity decoder and identical physical and confidence parameters. Success requires both logical parities to be accepted and correct.

\begin{figure*}[!t]
\centering
\includegraphics[width=0.98\textwidth]{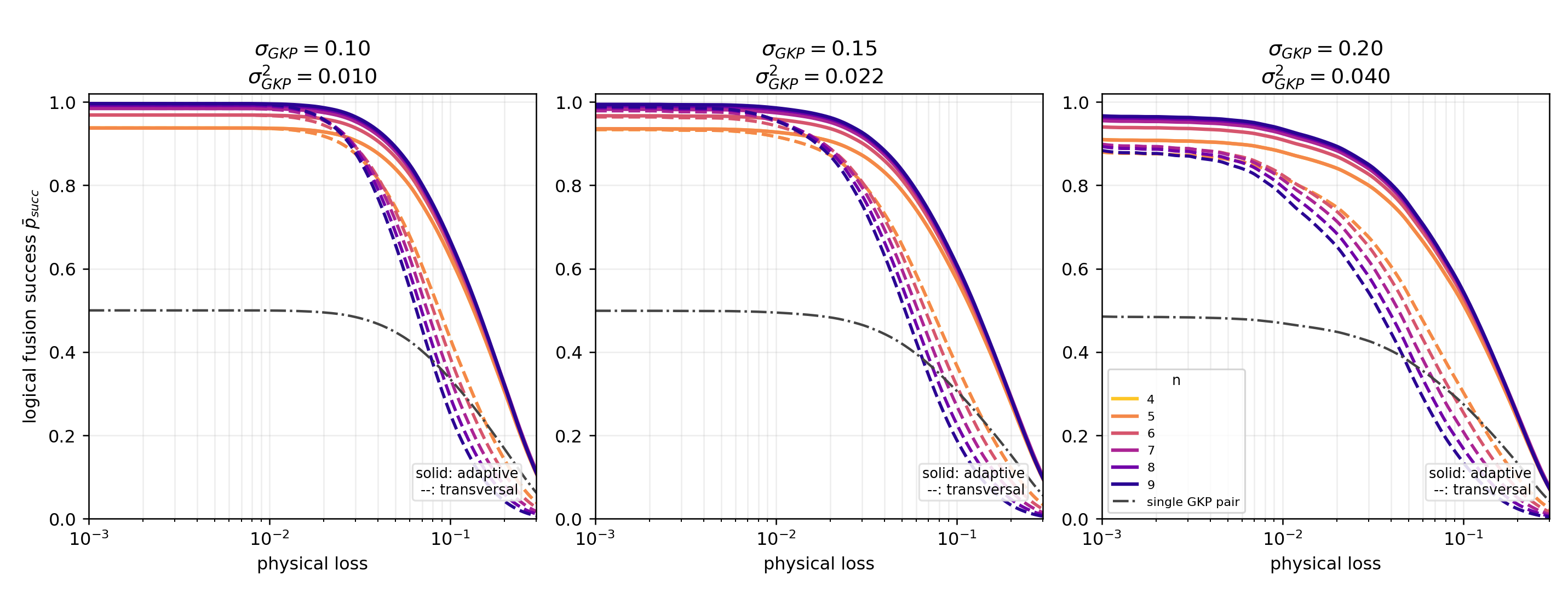}
\caption{Logical CV/GKP fusion success with closest-coset parity decoding for hexagonal GKP blocks at parity-confidence threshold $\Delta_{\rm th}=0.50$.  The three panels use $\sigma_{\GKP}=0.10$, $0.15$, and $0.20$ and graph sizes $n=4,\ldots,9$.  Solid curves are adaptive fusion and dashed curves are transversal fusion.  The dotted curve is the boosted physical-fusion reference and the dash-dotted curve is the single-GKP-pair reference.  The physical-loss coordinate is $\ell=1-\eta$.}
\label{fig:result4_hex_fusion}
\end{figure*}

\paragraph{Hexagonal-fusion interpretation.}
A useful operational threshold for Fig.~\ref{fig:result4_hex_fusion} is the half-success point $\ell_{1/2}$ defined by $p_{\rm succ}=1/2$.  For adaptive fusion, the displayed graph sizes give $\ell_{1/2}\simeq0.066$--$0.071$ at $17.0\,\mathrm{dB}$ squeezing, $0.060$--$0.064$ at $13.5\,\mathrm{dB}$, and $0.053$--$0.056$ at $11.0\,\mathrm{dB}$.  These correspond to only $0.30$--$0.32$, $0.27$--$0.29$, and $0.24$--$0.25\,\mathrm{dB}$ of channel attenuation.  Transversal fusion reaches half success earlier, at approximately $0.038$--$0.048$, $0.032$--$0.041$, and $0.022$--$0.033$ physical loss, respectively ($0.17$--$0.21$, $0.14$--$0.18$, and $0.10$--$0.15\,\mathrm{dB}$).  The adaptive advantage therefore survives finite squeezing, although the useful loss window contracts continuously as the squeezing is reduced.

\begin{figure*}[!t]
\centering
\includegraphics[width=0.98\textwidth]{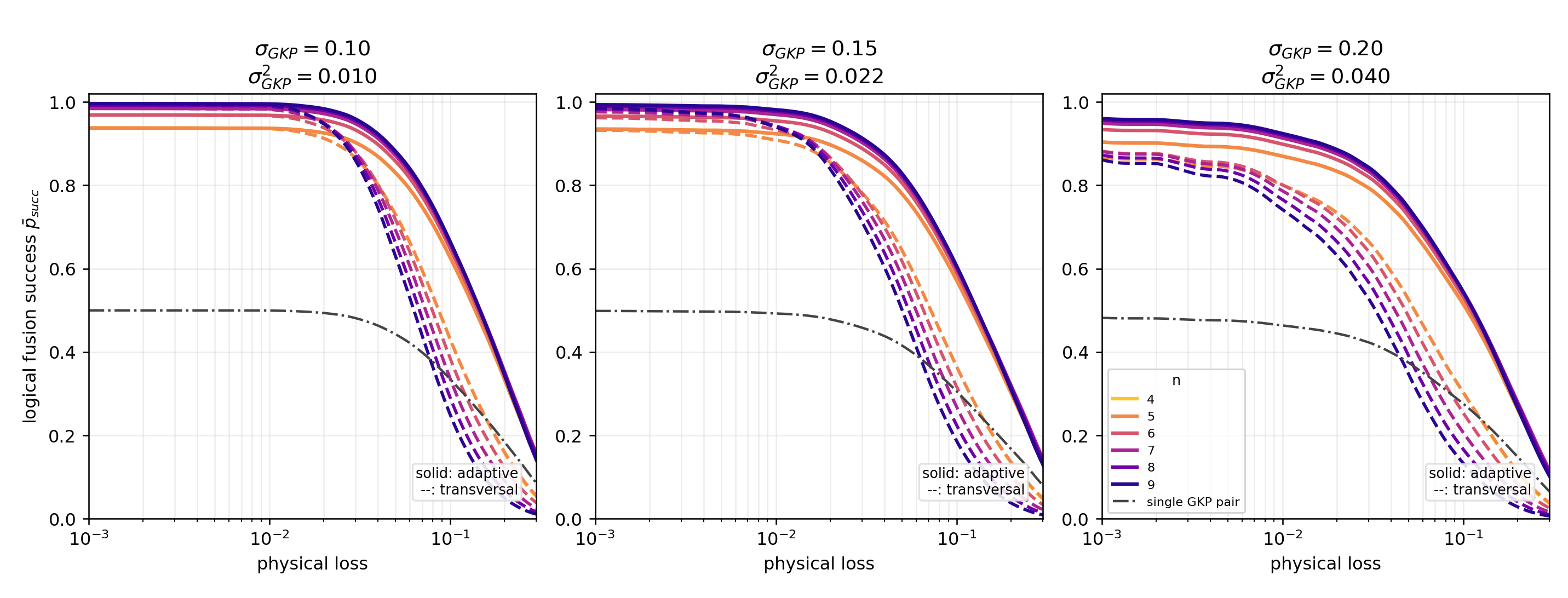}
\caption{Square-lattice counterpart of the logical-fusion benchmark in Fig.~\ref{fig:result4_hex_fusion}.}
\label{fig:result4_square_fusion}
\end{figure*}

\paragraph{Square-fusion interpretation.}
The square-lattice half-success points are very close to the hexagonal ones.  At $17.0$, $13.5$, and $11.0\,\mathrm{dB}$ squeezing, adaptive fusion gives $\ell_{1/2}\simeq0.067$--$0.071$, $0.063$--$0.068$, and $0.055$--$0.059$, respectively.  These ranges correspond to $0.30$--$0.32$, $0.28$--$0.31$, and $0.25$--$0.26\,\mathrm{dB}$ attenuation.  The transversal ranges are $0.037$--$0.047$, $0.031$--$0.043$, and $0.021$--$0.034$; in decibels they are $0.16$--$0.21$, $0.14$--$0.19$, and $0.09$--$0.15\,\mathrm{dB}$.  Across the simulated parameter range, lattice geometry changes the detailed success curves more strongly than it changes the half-success attenuation itself.
\section{Discussion and outlook}
\label{sec:discussion}

The central contribution of this work is a causal statistical interface between bosonic recovery and graph-level measurement control. In the proposed architecture, local GKP recovery does more than output a hard Pauli decision. It produces a continuous syndrome record, a posterior distribution over residual Pauli classes, and a confidence-dependent availability decision. When the confidence is insufficient, the refreshed block is deliberately converted into a located erasure before the destructive outer graph measurement. Importantly, this decision does not remove the corresponding record from the likelihood. The rejected syndrome remains informative about the physical displacement sector and therefore contributes to the subsequent inference of the outer syndrome and outgoing logical Pauli frame. This selected-event conditioning is the main distinction between the present decoder and a simple serial composition of a local GKP decoder with a graph-code loss decoder.

This distinction also clarifies the separation between several operational failure mechanisms. A local abstention, denoted by $\perp_{\rm loc}$, removes one candidate block from the future measurement pattern but need not cause logical failure because another representative may remain accessible. An accessibility failure $\perp_{\rm acc}$ occurs only when no compatible representative of the requested logical observable survives. A confidence failure $\perp_{\rm conf}$ occurs when a branch remains accessible but its inferred outgoing frame is not sufficiently reliable. Finally, an accepted frame error corresponds to an incorrect logical Pauli update despite branch acceptance. Treating these events separately avoids double counting and exposes the tradeoff controlled by the local confidence threshold: increasing the threshold converts more uncertain accepted Pauli errors into located erasures, which can then be handled by graph redundancy when an alternative representative exists.

The controller-facing object is consequently not a single decoded bit, but a posterior over the outer syndrome and the induced logical Pauli frame. This is particularly important in MBQC, where the frame determines both the interpretation of previous measurement outcomes and the physical basis of subsequent measurements. A wrong $Z$-frame estimate flips a reported logical outcome, whereas a wrong $X$-frame estimate changes the sign of a future equatorial measurement angle. Retaining the complete posterior therefore permits soft propagation of frame uncertainty through the computation rather than forcing an irreversible hard decision after every module. The update rule in Eq.~\eqref{eq:posterior_frame_propagation} provides the natural classical interface for this purpose: each graph--GKP module exports a posterior channel that can be composed with the controller state and used to adapt later measurement bases.

The finite-size results demonstrate how this interface behaves for explicit graph modules, but they should be interpreted according to the benchmark definitions introduced in Sec.~\ref{subsec:loss_amplification_exact_results}. The cube and decorated-pentagon reliability polynomials possess nontrivial fixed points for the idealized independent-erasure accessibility problem. These fixed points quantify the graph combinatorics acting on decoder-generated availability flags; they are not direct thresholds in the optical attenuation variable. In the full loss--amplification model, attenuation first modifies the Gaussian displacement covariance, which changes both the local abstention probability and the accepted Pauli-error distribution. The reported attenuation-reference crossings and encoded-to-single-GKP excess break-even regions therefore characterize finite-size operating boundaries of the complete decoder rather than asymptotic fault-tolerance thresholds.

Within these qualifications, the numerical comparisons reveal several useful trends. First, the outer graph topology has a substantial effect on the loss-tolerance boundary, and increasing the number of modes does not produce a monotonic improvement for the unoptimized graph catalog. This reflects the competing effects of additional logical representatives, larger measurement support, accumulated frame uncertainty, and the probability that at least one required block is locally rejected. Second, the hexagonal GKP lattice generally reduces the logical failure probability relative to the square lattice because of its larger shortest nontrivial displacement. The lattice advantage appears more clearly in the vertical separation of the failure curves than in a large shift of the attenuation-reference crossing. This suggests that graph topology and local lattice geometry play complementary roles: the graph controls accessibility and rerouting, while the local lattice primarily controls the reliability of the accepted analog decisions.

The non-Pauli results should likewise be interpreted as benchmarks of the protected transport and Pauli-frame layer surrounding the $A(\theta)$ interface. The encoded non-Gaussian terminal resource is supplied ideally in the reported simulations, so the curves do not include magic-state preparation, distillation, injection, or terminal-readout errors. Nevertheless, they quantify a necessary component of a fault-tolerant non-Clifford MBQC primitive: the ability to route the encoded state to a viable output, infer the syndrome-dependent outgoing frame, and adapt the physical angle and reported bit accordingly. An end-to-end non-Pauli analysis should incorporate a noisy encoded resource, a specific injection or teleportation circuit, and correlations between resource preparation, GKP recovery, and graph measurements. Such an extension would determine whether the routing and frame-management advantage survives after the complete non-Gaussian resource cost is included.

Recursive concatenation provides a route from the finite modules studied here to larger architectures, but the relevant recursion is a recursion of posterior channels rather than a scalar recursion of erasure probabilities. Each child module exports an availability flag together with a complete distribution over its residual frame. The parent must include the likelihood of the child acceptance or rejection event before applying graph pathfinding and syndrome inference. The scalar accessibility maps derived for the cube and decorated pentagon are therefore useful implementation checks, but they omit accepted frame errors, confidence distributions, and correlations between child modules. A future threshold analysis should propagate the full joint channel to substantially larger depths, extract genuine level-to-level pseudothresholds, and test whether the logical error decreases systematically below a stable critical region. Efficient density-evolution, tensor-network, or population-dynamics approximations may be useful for this purpose when exact posterior propagation becomes prohibitively expensive.

The fusion construction shows that the same decoder principle applies to logical two-module operations. The commuting logical parities define a Bell subgroup, while the observable parity sectors are labeled by conjugate dual shifts rather than by the measured parity generators themselves. The closest-coset decoder then produces a posterior over the parity-frame class, from which successful fusion, partial erasure, full erasure, and accepted parity-frame error are obtained as distinct outcomes. This formulation is directly compatible with fusion-based quantum computation \cite{bartolucci2023fusion,song2024encoded}, where a higher-level controller must know not only whether a fusion was attempted, but also which parity components were accepted and how reliable the inferred frame is.

The adaptive and transversal fusion curves demonstrate a policy-level advantage for sequential posterior updating. Under the declared attempt rules, the adaptive policy reaches the half-success marker at larger attenuation because it updates the parity posterior after each retained Bell record and stops when the joint acceptance criterion is satisfied. The present comparison is intentionally not resource normalized. Adaptive and transversal strategies can consume different expected numbers of Bell attempts, interface blocks, homodyne detections, and graph vertices. A complete architecture comparison should therefore optimize a joint objective involving logical fusion success, accepted parity-frame error, expected Bell-attempt count, optical depth, and resource-state consumption. It should also include any intrinsic probabilistic Bell-analyzer failure as a separate hardware channel rather than identifying it with the attenuation parameter of the loss--amplification model.

The posterior interface is also relevant to all-photonic quantum repeaters and quantum-network architectures \cite{fukui2021all,rozpkedek2023all,swain2026loss,koudia2024quantum}. In a repeater chain, the quality of a received graph fragment need not be represented by a binary success flag. The local GKP records can instead be converted into module-level availability and frame distributions that inform whether the fragment should be consumed, fused, rerouted, stored, or replaced. This provides a natural interface between the physical bosonic layer and link- or network-layer scheduling. The same principle extends to multimode bosonic communication, where several spatial or modal channels provide correlated reliability information \cite{Koudia2025CVComm,Koudia2025MIMO}. In such settings, syndrome-resolved posteriors could support path diversity, mode selection, decoder allocation, and confidence-aware routing.

Several limitations identify concrete directions for further work. The numerical benchmarks use the loss--amplification covariance together with a controlled finite-squeezing contribution, while additional gate, ancilla, detector, and feedforward noise terms in Eq.~\eqref{eq:loss_amplification_total_covariance} are set to zero. Incorporating these terms will be necessary for comparison with a specific optical implementation. All executable results use the closest-coset approximation rather than the exact wrapped-coset likelihood; small-module comparisons between the two would quantify the effect of the omitted multiplicity factors. The graph families are finite and unoptimized, the recursive study is limited to a few levels, the non-Pauli resource is idealized, and the fusion comparison is not normalized by expected resource consumption. These restrictions prevent an end-to-end asymptotic fault-tolerance claim, but they do not affect the definition of the selected-event posterior or its role as the common decoder interface.

A particularly important next step is a direct decoder ablation. One should compare the full selected-event decoder against a hard-decision architecture that passes only local Pauli estimates and erasure flags, a decoder that discards rejected-syndrome information, and an approximation that treats the availability mask as statistically independent of the analog record. Such a comparison would isolate the quantitative value of retaining the selection likelihood and clarify which regimes benefit most from syndrome-resolved Pauli-frame inference. A complementary optimization should jointly choose the graph topology, logical representatives, basis-dependent confidence thresholds, branch policy, local GKP lattice, concatenation depth, and fusion boundary under fixed mode and measurement budgets.

Overall, the results establish a unified route from continuous GKP records to the discrete control information required by graph-state quantum computation. The same selected-event closest-coset engine applies to logical Pauli measurements, protected non-Pauli transport, recursive concatenation, and logical fusion. This common interface makes it possible to design larger measurement-based and fusion-based architectures in which analog bosonic information is retained until the final graph-level syndrome and Pauli-frame decision, rather than being compressed prematurely into independent hard errors and erasures.

\appendix

\section{Proof of the branch-compatible graph--GKP compilation}
\label{app:branch_compilation}

\begin{proof}[Proof of Theorem~\ref{thm:branch_compilation}]
The chosen lift satisfies
\begin{equation*}
\Phi(u)+\Phi(v)-\Phi(u\oplus v)\in L_{\rm loc},
\qquad
2\Phi(u)\in L_{\rm loc}.
\end{equation*}
In particular, since $2\Phi(v)\in L_{\rm loc}$,
\begin{equation}
\Phi(u)-\Phi(v)
\equiv
\Phi(u)+\Phi(v)
\equiv
\Phi(u\oplus v)
\pmod{L_{\rm loc}}.
\label{eq:branch_lift_congruence}
\end{equation}

Suppose first that $u\oplus v=s\in Q_b$. Equation~\eqref{eq:branch_lift_congruence} then gives
\begin{equation*}
\Phi(u)-\Phi(v)=\Phi(s)+\lambda_0
\end{equation*}
for some $\lambda_0\in L_{\rm loc}$. Since
\begin{equation*}
L_{G,b}
=
L_{\rm loc}
+
\spanZ\{\Phi(q):q\in Q_b\},
\end{equation*}
it follows that $\Phi(u)-\Phi(v)\in L_{G,b}$ and hence
\begin{equation*}
\Phi(u)+L_{G,b}=\Phi(v)+L_{G,b}.
\end{equation*}

Conversely, suppose that $\Phi(u)-\Phi(v)\in L_{G,b}$. Reducing modulo $L_{\rm loc}$ gives
\begin{equation*}
\overline\Phi(u\oplus v)
=
\pi_{\rm loc}\!\left(\Phi(u)-\Phi(v)\right)
\in \overline\Phi(Q_b),
\end{equation*}
where $\pi_{\rm loc}:L_{\rm loc}^{\perp}\to
L_{\rm loc}^{\perp}/L_{\rm loc}$ is the quotient map. Therefore there exists $s\in Q_b$ such that
\begin{equation*}
\overline\Phi(u\oplus v)=\overline\Phi(s).
\end{equation*}
Faithfulness of $\overline\Phi$ implies
\begin{equation*}
u\oplus v=s\in Q_b.
\end{equation*}

Finally, Eq.~\eqref{eq:branch_fiber_coset} gives
\begin{equation*}
L_{a,b}=q_{a,b}+Q_b
\end{equation*}
for every nonempty logical branch fiber. Hence any two elements of $L_{a,b}$ differ by an element of $Q_b$ and, by the equivalence just proved, map to the same coset of $L_{G,b}$. Conversely, two compatible representatives mapping to the same coset differ by an element of $Q_b$. Thus each nonempty logical branch fiber is represented by one and only one coset of $L_{G,b}$.
\end{proof}

\section{Carry-aware signed compilation}
\label{app:signed_compilation}

\paragraph{Signed Pauli convention and lattice carry.}

We use the standard Hermitian binary Pauli section
\cite{dehaene2003}
\begin{equation*}
P_H(u)
=
i^{x\cdot z}
\bigotimes_{i=1}^{n}X_i^{x_i}Z_i^{z_i},
\qquad
P_{\pm}(u,\epsilon)
=
(-1)^\epsilon P_H(u),
\end{equation*}
for $u=(x\mid z)\in\F_2^{2n}$. Its multiplication cocycle is defined by
\begin{equation}
P_H(u)P_H(v)
=
i^{\omega_P(u,v)}P_H(u\oplus v),
\label{eq:pauli_cocycle_product}
\end{equation}
where, for $v=(x'\mid z')$,
\begin{equation*}
\omega_P(u,v)
=
2z\cdot x'
+x\cdot z
+x'\cdot z'
-(x\oplus x')\cdot(z\oplus z')
\pmod 4.
\end{equation*}
For commuting labels, $\omega_P(u,v)$ is even, and Eq.~\eqref{eq:pauli_cocycle_product} gives the signed-tableau rule
\begin{equation}
P_{\pm}(u,\epsilon)P_{\pm}(v,\delta)
=
P_{\pm}\!\left(
u\oplus v,
\epsilon\oplus\delta
\oplus\frac{\omega_P(u,v)}{2}
\right).
\label{eq:signed_tableau_rule}
\end{equation}

Under the logical embedding of the product GKP code, the abstract Hermitian Pauli $P_H(u)$ is represented on the product code space by the displacement $D(\Phi(u))$. Because the chosen lift is linear only modulo the base lattice, binary addition produces the carry
\begin{equation}
c(u,v)
=
\Phi(u)+\Phi(v)-\Phi(u\oplus v)
\in L_{\rm loc}.
\label{eq:binary_lattice_carry}
\end{equation}
Writing $W=\Phi(u\oplus v)$, the Weyl relation gives
\begin{align*}
D(\Phi(u))D(\Phi(v))
={}&
e^{-i\pi\Phi(u)^TJ_{\rm loc}\Phi(v)}
D(W+c(u,v))
\\
={}&
e^{-i\pi\Phi(u)^TJ_{\rm loc}\Phi(v)}
e^{i\pi W^TJ_{\rm loc}c(u,v)}
D(W)D(c(u,v)).
\end{align*}
Since
\begin{equation*}
e^{i\phi_{\rm loc}(c)}D(c)
\end{equation*}
acts trivially on the product local-code space for every
$c\in L_{\rm loc}$, comparison with
Eq.~\eqref{eq:pauli_cocycle_product} yields the intertwining identity
\begin{multline}
e^{-i\pi\Phi(u)^TJ_{\rm loc}\Phi(v)}
e^{i\pi\Phi(u\oplus v)^TJ_{\rm loc}c(u,v)}
e^{-i\phi_{\rm loc}(c(u,v))}
\\
=
i^{\omega_P(u,v)}.
\label{eq:carry_intertwining}
\end{multline}
This is the only convention-dependent identity needed below; the standard Pauli and Weyl multiplication laws themselves are not rederived.

\paragraph{Extension to the compiled lattice.}

Let the signed graph-code stabilizer with label $s\in Q_G$ be
\begin{equation}
\widehat P(s)
=
(-1)^{\epsilon(s)}P_H(s),
\qquad
\epsilon:Q_G\longrightarrow\F_2,
\label{eq:signed_stabilizer_section}
\end{equation}
where the sign function $\epsilon$ is fixed by the progenitor stabilizer group. Since $Q_G$ is isotropic, its labels commute. Closure of the signed stabilizer group therefore requires
\begin{equation}
\epsilon(s\oplus t)
=
\epsilon(s)\oplus\epsilon(t)
\oplus\frac{\omega_P(s,t)}{2}
\pmod 2
\label{eq:signed_stabilizer_cocycle}
\end{equation}
for all $s,t\in Q_G$.

The corresponding compiled graph generator is assigned the phase
\begin{equation}
\phi_G(\Phi(s))
=
\pi\epsilon(s)
\pmod{2\pi}.
\label{eq:compiled_graph_phase_generator}
\end{equation}
For a symplectically integral lattice, a compatible phase function obeys the quadratic-refinement law
\begin{equation}
\phi(\lambda+\mu)
=
\phi(\lambda)+\phi(\mu)
-\pi\lambda^TJ\mu
\pmod{2\pi}.
\label{eq:lattice_phase_refinement}
\end{equation}

\begin{theorem}[Carry-aware extension of the lattice phase sector]
\label{thm:signed_compilation}
Assume that $\overline\Phi$ is faithful and that
Eq.~\eqref{eq:signed_stabilizer_section} defines an abelian stabilizer group. Every $\lambda\in L_G$ has a unique decomposition
\begin{equation*}
\lambda=\lambda_0+\Phi(s),
\qquad
\lambda_0\in L_{\rm loc},
\quad
s\in Q_G.
\end{equation*}
Moreover,
\begin{align}
\phi_G(\lambda_0+\Phi(s))
={}&
\phi_{\rm loc}(\lambda_0)
+\pi\epsilon(s)
\nonumber\\
&-\pi\lambda_0^TJ_{\rm loc}\Phi(s)
\pmod{2\pi}
\label{eq:compiled_phase_explicit}
\end{align}
defines the unique phase sector extending the local phases and the graph-code signs. Consequently,
\begin{equation}
\mathcal S_G
=
\left\{
e^{i\phi_G(\lambda)}D(\lambda):
\lambda\in L_G
\right\}
\label{eq:compiled_signed_stabilizer_group}
\end{equation}
is a well-defined abelian GKP stabilizer group.
\end{theorem}

\begin{proof}
We first establish the decomposition. Since
\begin{equation*}
L_G/L_{\rm loc}
\simeq
\overline\Phi(Q_G),
\end{equation*}
the quotient class of every $\lambda\in L_G$ can be written as
\begin{equation*}
\lambda+L_{\rm loc}
=
\Phi(s)+L_{\rm loc}
\end{equation*}
for some $s\in Q_G$. Hence
\begin{equation*}
\lambda_0:=\lambda-\Phi(s)\in L_{\rm loc},
\end{equation*}
which proves existence.

If
\begin{equation*}
\lambda_0+\Phi(s)
=
\lambda_0'+\Phi(t),
\end{equation*}
then
\begin{equation*}
\Phi(s)-\Phi(t)\in L_{\rm loc}.
\end{equation*}
Reducing modulo $L_{\rm loc}$ and using the order-two property of the quotient gives
\begin{equation*}
\overline\Phi(s\oplus t)=0.
\end{equation*}
Faithfulness of $\overline\Phi$ implies $s=t$, after which
$\lambda_0=\lambda_0'$. The decomposition is therefore unique.

It remains to verify consistency of the phase assignment. For two graph labels $s,t\in Q_G$, define
\begin{equation*}
A=\Phi(s),
\qquad
B=\Phi(t),
\qquad
w=s\oplus t,
\qquad
W=\Phi(w),
\end{equation*}
and let
\begin{equation*}
c=A+B-W\in L_{\rm loc}.
\end{equation*}
Closure of the signed graph stabilizer, together with
Eq.~\eqref{eq:carry_intertwining}, gives
\begin{align}
\pi\epsilon(w)
={}&
\pi\epsilon(s)+\pi\epsilon(t)
-\pi A^TJ_{\rm loc}B
\nonumber\\
&+\pi W^TJ_{\rm loc}c
-\phi_{\rm loc}(c)
\pmod{2\pi}.
\label{eq:graph_phase_closure}
\end{align}
Using $A+B=c+W$ in Eq.~\eqref{eq:compiled_phase_explicit},
\begin{align*}
\phi_G(A+B)
={}&
\phi_{\rm loc}(c)
+\pi\epsilon(w)
-\pi c^TJ_{\rm loc}W.
\end{align*}
Substitution of Eq.~\eqref{eq:graph_phase_closure}, together with antisymmetry of $J_{\rm loc}$, gives
\begin{align*}
\phi_G(A+B)
={}&
\pi\epsilon(s)+\pi\epsilon(t)
-\pi A^TJ_{\rm loc}B
\\
&+2\pi W^TJ_{\rm loc}c
\pmod{2\pi}.
\end{align*}
Because $W\in L_{\rm loc}^{\perp}$ and
$c\in L_{\rm loc}$,
\begin{equation*}
W^TJ_{\rm loc}c\in\mathbb Z,
\end{equation*}
and the last term vanishes modulo $2\pi$. Therefore,
\begin{equation}
\phi_G(A+B)
=
\phi_G(A)+\phi_G(B)
-\pi A^TJ_{\rm loc}B
\pmod{2\pi}.
\label{eq:graph_graph_refinement}
\end{equation}

The refinement law already holds on $L_{\rm loc}$. In addition,
$\Phi(s)\in L_{\rm loc}^{\perp}$ for every $s\in Q_G$, so all local--graph symplectic pairings are integral. Equation~\eqref{eq:graph_graph_refinement}, bilinearity, and the unique decomposition of every element of $L_G$ therefore extend the refinement law to arbitrary
\begin{equation*}
\lambda_0+\Phi(s),
\qquad
\mu_0+\Phi(t)
\in L_G.
\end{equation*}
The prescribed values determine $\phi_G$ uniquely. Finally, symplectic integrality of $L_G$ implies that the associated displacement generators commute, so
Eq.~\eqref{eq:compiled_signed_stabilizer_group} defines a well-defined abelian GKP stabilizer group.
\end{proof}

\paragraph{Observed signed characters.}

Let
\begin{equation*}
F_b\in\F_2^{2n}
\end{equation*}
denote the known physical Pauli-frame label on the graph blocks immediately before the measurements of branch $b$. It includes a chosen physical representative of the incoming logical controller frame together with all deterministic byproducts accumulated before the branch.

For a measurement-supported signed Pauli $\widehat Q$, let
$\beta_b(\widehat Q)\in\F_2$ denote its measured eigenvalue bit, with eigenvalue $(-1)^{\beta_b(\widehat Q)}$. Products of local measurement factors are accumulated using
Eq.~\eqref{eq:signed_tableau_rule}. For a measured check
$q_{b,k}\in Q_b^{\meas}$, the frame-corrected observed syndrome bit is
\begin{equation}
z_{b,k}
=
\beta_b\!\left(\widehat S(q_{b,k})\right)
\oplus
[q_{b,k},F_b]_{\DV}.
\label{eq:observed_signed_check}
\end{equation}

For a requested logical observable, choose a fixed signed reference
$\widehat Q^0_{a,b}$ with binary label $q^0_{a,b}$. Every operational representative factors as
\begin{equation}
\widehat Q_{a,b}
=
\widehat Q^0_{a,b}\widehat S(s_{a,b}),
\qquad
s_{a,b}\in Q_b^{\meas}.
\label{eq:signed_reference_factorization}
\end{equation}
The corresponding reference-normalized raw bit is
\begin{align}
\chi_b^{\rm raw}(a;q^0_{a,b})
={}&
\beta_b(\widehat Q_{a,b})
\oplus
\beta_b\!\left(\widehat S(s_{a,b})\right)
\nonumber\\
&\oplus
[q^0_{a,b},F_b]_{\DV}.
\label{eq:chi_explicit}
\end{align}
Changing the operational representative multiplies both measured factors in Eq.~\eqref{eq:signed_reference_factorization} by the same signed stabilizer character. Their changes therefore cancel in Eq.~\eqref{eq:chi_explicit}, so the reconstructed raw logical character is independent of the selected representative.

\bibliographystyle{apsrev4-2}
\bibliography{references}


\end{document}